\pdfoutput=1
 \documentclass[a4paper, 12pt]{article}
 \usepackage[T1]{fontenc}
 \usepackage[utf8]{inputenc}

 \usepackage[tmargin=1.2in,bmargin=1.2in,lmargin=1.2in,rmargin=1.2in]{geometry}
 \usepackage{setspace} % before hyperref
 \usepackage[
 activate={true,nocompatibility},
 factor=1000,
 expansion=false,           
 protrusion=true,          
 tracking=true,            
 spacing=true,             
 shrink=5,                
 stretch=5,               
 kerning=true,             
 final                      
 ]{microtype}
 \microtypecontext{spacing=nonfrench}
 
 \usepackage{mathtools, amsthm, amssymb, bm}
 \mathtoolsset{centercolon}
 \usepackage[caption=false]{subfig}
 \usepackage{graphicx, pgfplots, tikz, xcolor}
 \pgfplotsset{compat=1.16}
 \usepackage[skip = 10pt]{caption}
 \allowdisplaybreaks[1]
 \allowdisplaybreaks[1]

 \definecolor{blue}{RGB}{43,147,206}
 \definecolor{GmailBlue}{RGB}{42, 93, 176}
 
 \definecolor{red}{RGB}{221,126,0}
 \definecolor{orange}{RGB}{208,106,11} 
 \definecolor{green}{RGB}{0,158,115} 
 \definecolor{gray}{RGB}{73, 73, 73}
 \definecolor{yellow}{RGB}{230,168,0}
 \definecolor{pink}{RGB}{211,0,214}

 \usepackage[authoryear]{natbib}

\usepackage[
bookmarks=true,
bookmarksnumbered=true,
bookmarksopen=true,
pdfborder={0 0 0},
breaklinks=true,
colorlinks=true,
linkcolor=GmailBlue,
citecolor=GmailBlue,
filecolor=GmailBlue,
urlcolor=GmailBlue,
hypertexnames=false,
plainpages=false,
]{hyperref}

\usepackage{apptools}
\AtAppendix{\counterwithin{lemma}{section}}
\AtAppendix{\counterwithin{proposition}{section}}
\AtAppendix{\counterwithin{theorem}{section}}
\AtAppendix{\counterwithin{corollary}{section}}
\AtAppendix{\counterwithin{remark}{section}}
\AtAppendix{\counterwithin{definition}{section}}

\theoremstyle{plain}

\newtheorem{proposition}{Proposition}
\newtheorem{lemma}{Lemma}

\newtheorem{prop}{\protect\propositionname}
\newtheorem{lem}{\protect\lemmaname}
\newtheorem{cor}{\protect\corollaryname}

\theoremstyle{definition}
\newtheorem{definition}{Definition}

\newtheorem{rem}{\protect\remarkname}
\newtheorem{remark}{Remark}

\providecommand{\corollaryname}{Corollary}
\providecommand{\definitionname}{Definition}
\providecommand{\lemmaname}{Lemma}
\providecommand{\propositionname}{Proposition}
\providecommand{\remarkname}{Remark}

  \newcommand*\diff{\mathop{}\!\mathrm{d}}
 \DeclareMathOperator*{\argmin}{argmin}
 \DeclareMathOperator*{\argmax}{argmax}

 \DeclareMathOperator*{\supp}{supp}
 \DeclareMathOperator*{\vex}{conv}

  \makeatletter
  \let\originalleft\left
  \let\originalright\right
  \renewcommand{\left}{\mathopen{}\mathclose\bgroup\originalleft}
  \renewcommand{\right}{\aftergroup\egroup\originalright}
  \renewcommand*{\NAT@spacechar}{\ }
  \makeatother 
  \usepackage{datetime}
  \newdateformat{specialdate}{\THEDAY~\monthname[\THEMONTH] \THEYEAR}
  \date{\specialdate\today}
 
\begin{document}
	\title{\textbf{Screening in digital monopolies}\thanks{This paper is based on the first chapter of Elia Sartori's PhD dissertation at Princeton University. Elia Sartori wishes to thank Stephen Morris for invaluable guidance and support. We thank Alberto Bennardo, Francis Bloch, Roberto Corrao, Vincenzo Denicolò, Tangren Feng, Joel Flynn, Daniel Gottlieb, Faruk Gul, Paul Heidhues, Andreas Kleiner, Jan Knoepfle, Sophie Kreutzkamp, Franz Ostrizek, Marco Pagnozzi, Eduardo Perez-Richet, Wolfgang Pesendorfer, Giorgio Saponaro, Alex Smolin, and Tymon Tatur for comments and discussions. We also thank Lorenzo Cuomo for invaluable research assistance. Financial support from Unicredit and Universities Foundation is gratefully acknowledged. An earlier version of this paper bore the title ``Competitive provision of digital goods.''}}

	\author{Pietro Dall'Ara\thanks{CSEF and University of Naples Federico II; contact: \href{mailto:pietro.dallara@unina.it}{pietro.dallara@unina.it}.} \and Elia Sartori\thanks{CSEF and University of Naples Federico II; contact: \href{mailto:elia.sartori@unina.it}{elia.sartori@unina.it}.}} 
	\maketitle
 	\begin{abstract}
A defining feature of digital goods is that replication and degradation are costless: once a high-quality good is produced, low-quality versions can be created and distributed at no additional cost. This paper studies quality-based screening in markets for digital goods. Production costs depend only on the highest quality supplied, unlike in standard screening models. The monopolist allocation exhibits two interdependent inefficiencies. First, a productive inefficiency: the monopolist underinvests in the highest quality relative to the efficiency benchmark. Second, due to a distributional inefficiency, certain buyers receive degraded versions of the produced good. Competition exacerbates productive inefficiency, but improves distributional efficiency.
 	\end{abstract}
 	\noindent Keywords: Screening, Monopoly, Digital Goods, Mechanism Design.\\ \noindent JEL codes: D42, D82, L12.
 	\clearpage 
 %	\tableofcontents
 	
 	\section{Introduction}\label{sec:Introduction}
 	Quality-based screening is pervasive across markets: sellers routinely offer multiple versions of the same product at different prices, effectively discriminating among consumers with different tastes for quality. For example, car companies distinguish between models with and without alloy wheels, and tech companies sell premium and basic versions of the same software. The canonical setting in which to study quality-based screening is the Mussa–Rosen model. In this model, the seller optimally allocates to each consumer a quality level such that the marginal cost of providing that quality equals the consumer’s virtual value. The difference between a consumer’s type (her taste for quality) and her virtual value captures the rents that her consumption generates for other types.
 	
 	This ``marginal cost equals adjusted marginal utility'' condition reveals the core force in quality screening: the quality of a type is under-provided in order to reduce the information rents of other types. This equality is due to an important feature of the model: all interdependence across types can be handled by transforming the revenues generated by each type. In particular, the model assumes no production interdependencies: total production cost is simply the sum of the cost of each good produced. We refer to this assumption, implicit in the Mussa--Rosen framework, as ``cost separability.''
 
 	However, cost separability is not realistic in every industry. It is arguably a compelling description of the car industry, in which the cost of supplying a vehicle depends on the version of that specific car, e.g., with or without alloy wheels. A software developer, instead, incurs an upfront cost to create a product of a given quality, after which the product can be replicated---and degraded, if needed---at negligible marginal cost. Hence, the cost of allocating a software product is not well approximated as depending only on the version delivered to each buyer. A better approximation is that the production cost of the entire menu is driven by the development of its top tier.
 	
 	In this paper, we develop a model of screening under an ``invest, then damage'' technology. The seller first chooses a \emph{top-quality} and pays the corresponding production cost. It can then replicate that product and supply lower qualities by degrading the cap-quality product at zero cost. In particular, producing below the cap is costless. The demand side is standard: there is a continuum of buyers with one-dimensional private information (type) and increasing-differences utility.

	Markets for digital goods provide the most natural application of our cost structure. A defining feature of digital production is near-zero replication cost \citep{goldfarb_digital_2019}, and in many settings offering additional tiers amounts to restricting features, limiting performance, or gating access---i.e., creating ``damaged'' versions at negligible incremental cost.\footnote{Section \ref{subsec:interpretation_digital} details the implications of our screening model in the context of the digital economy.} Information markets provide an equally compelling example \citep{bergemann_information_2021}. Producing a rating, or a financial report, largely requires upfront investment, e.g., data collection, analysts, consultants. Once an assessment is produced, it can be distributed to many clients  in both full and intentionally coarsened forms, e.g., redacting sections, or aggregating detail. In both industries, versioning and quality screening are pervasive: software is routinely sold in tiered editions, and information intermediaries such as Equifax offer products that differ in depth and access. Yet, the Mussa--Rosen model---and its core marginal-cost-equals-virtual-value logic---relies on separable costs, limiting its ability to speak to a first-order phenomenon in some of today’s most important markets. These observations motivate the following questions.
	Are the canonical screening insights robust to nonseparable, \emph{top-quality} costs? How does this technology reshape the familiar inefficiencies from asymmetric information? How do predictions change when the seller cannot degrade quality, or when multiple sellers compete?
	 	
 	Two benchmarks organize the analysis. First, a planner that maximizes total surplus allocates the same undamaged quality to all buyer types. The efficient allocation is a constant function of the type, because giving the highest quality of any allocation to all types increases total utility leaving costs unchanged. Second, the efficient quality equates the marginal cost of production to the average marginal utility in the population of buyers.
 	
 	The monopolist allocation features two kinds of inefficiencies. The seller chooses a cap $q^{M}$, and allocates to each type either the cap or a damaged---i.e., lower---quality. We show that a \emph{productive inefficiency} arises: the monopolist underinvests in the highest quality. Specifically, $q^{M}$ is lower than the efficient quality. The productive inefficiency interacts with the \emph{distributional inefficiency} that arises from information rents and occurs via damaging. 
 	
 	The characterization of the monopolist allocation uses a decomposition of the monopolist problem into a family of ``separable'' problems. For a fixed cap $q$, we define the cap-constrained problem as a canonical zero-cost virtual-surplus problem subject to the upper bound $q$ on the feasible quality allocations. The optimal allocation for this auxiliary problem maximizes the virtual surplus truncated at $q$. Thus, a region of high types are \emph{bunched} because they all receive the cap quality, and every low type receives the damaged version of the cap that maximizes virtual surplus conditional on her type.
 	
 	The seller jointly solves the cap-constrained problem and chooses the cap by equating the marginal production costs with the marginal revenue from relaxing the cap. The marginal revenue equals the marginal utility of the lowest type that is bunched, the \emph{marginally bunched} type, weighted by the mass of the bunching region. Intuitively, the extra revenue given by a higher cap comes from charging a higher price for the good purchased by all bunched types, because the allocation of low types is unchanged.
 	
 	The marginally bunched type has a special role: she determines the price of the top quality increment and splits the market into two regions. Above the marginally bunched type, there is distributional efficiency because everyone gets the undamaged quality, but productive inefficiency is active because the top quality is lower than the efficient quality. Below this type, there is downgrading, but no additional productive loss relative to a monopolist who replicates and damages an efficient-quality good.
 	
 	In certain settings, regulatory or technological constraints make versioning infeasible. A ban on screening increases the share of buyers receiving an undamaged good. The marginal revenues in this case account for both the price increment and the inframarginal types, similarly to the screening case. The key difference lies in how the relevant cutoff type is determined. Without screening, the cutoff type makes the seller indifferent between serving her or not; with screening, the cutoff type---marginally bunched---makes the seller indifferent between damaging her quality or not. As a result, the no-screening seller underinvests in quality with respect to the screening case, so productive inefficiency is strengthened.
 	
 	We introduce competition in a stylized framework by considering multiple firms that commit to a cap before a pricing stage. The familiar Bertrand-pricing force drives the price of the second-highest quality produced to zero, so only one firm emerges as a monopolist. Effectively, the equilibrium adds a ``competitive'' constraint in the problem of the endogenous monopolist, stating that: any quality feasible for at least one competitor must be offered for free. The emergence of a monopoly induces a war-of-attrition feature in the production stage, so the production strategies are mixed in equilibrium. We show that every equilibrium induces a greater productive inefficiency and a more prevalent distributional efficiency than the monopolist allocation.
 	
 	\paragraph{Outline} The following paragraph summarizes the related literature. The main model and results are in Section \ref{sec:model}. Section \ref{sec:interpretation} offers interpretations and applications; in Section \ref{sec:extensions} we introduce a no-screening constraint and competition. The Appendix contains a more general setup and all proofs.

 	\paragraph{Related literature}
 	The idea that a seller may want to damage its products for screening purposes was introduced by \citet*{deneckere_damaged_1996}, who consider a seller paying damaging costs, possibly zero, to create inferior versions of an available product. We study a monopolist engaged in quality-based screening \citep*{mussa_monopoly_1978,maskin_monopoly_1984,wilson_nonlinear_1993}, who uses damaging following the logic of \citeauthor*{deneckere_damaged_1996}. The screening models in \citet*{oren_capacity_1985}, \citet*{grubb_selling_2009}, and \citet*{corrao_nonlinear_2023} feature consumers who damage goods through underutilization. The seller in \citet*{hahn_damaged_2006} and \citet*{inderst_durable_2008} offers damaged goods over time to mitigate the Coasean commitment problem. Our model introduces a production cost that depends only on the highest quality supplied. The cost structure is not separable in the sense that the total cost cannot be expressed as the sum of costs incurred by interacting separately with every type, contrary to the cited work.\footnote{A kind of nonseparability arises with inventory costs, in which any distribution covered by the inventory has zero costs and other distributions are infeasible \citep*{loertscher_monopoly_2022,bergemann_screening_2025}. Our costs are sensitive to the highest quality supplied, and so are not of the inventory type; however, the distribution stage of the monopolist problem can be interpreted as featuring inventory cost (Lemma \ref{lem:decomposition}). Capacity costs can be interpreted as not separable across consumers. For instance, the capacity cost in \citet*{boiteux_peak_1960} is not separable across time, in a model without asymmetric information, and the monopoly airline in \citet*{gale_advance_1993} pays capacity costs, over-investing compared to the efficient level.} This model captures a distinctive property of digital goods.\footnote{Research on digital economics spans from the study of platforms \citep*{rochet_platform_2003} to the pricing of cloud computing and language models \citep*{bergemann_economics_2025,bergemann_robust_2025}.}
 	
 	Information goods are freely replicable \citep*{bergemann_information_2021} and can be versioned cheaply \citep*{shapiro_versioning_1998}. In the typical approach to modeling information markets, a seller replicates and garbles Blackwell experiments at zero cost, but information goods exhibit other important properties, such as non-excludability \citep*{admati_monopolistic_1986}, payoff externalities \citep*{ichiharshi_economics_2021,bonatti_selling_2024,rodriguez_strategic_2024}, and multi-dimensionality (\citealp*{bergemann_design_2018}, in contrast with pure vertical differentiation.) Our model complements the literature by considering the production stage and identifying a source of monopoly inefficiency---the productive inefficiency and its interaction with the distributive properties of the monopolist allocation---although abstracting from other properties of information. The literature on excludable public goods studies goods that are freely replicable and excludable through a price system (e.g., \citealp*{moulin_1994_serial, meisner_monetizing_2025}.)
 	
 	Our costs can be viewed as generating a production externality, because the cost of allocating a quality to a type depends on the quality allocated to other types. The literature on screening with externalities considers the demand externalities that arise if the allocation of a buyer affects the utility that can be extracted from other buyers \citep*{segal_contracting_1999,jehiel_multidimensional_1999,segal_robust_2003}. Recently, \citet*{halac_pricing_2024} study the profit guarantee with network goods.
 	
 	In the single-agent interpretation of the model, replication is irrelevant, and our cost structure represents a seller who produces before eliciting the type of the buyer and damaging the produced good. Hence, the comparison with the workhorse screening model identifies the role of the timing of production, because the seller produces after eliciting the type in \citet*{mussa_monopoly_1978}.
 	
 	The game with which we introduce competition has a natural timing given the two-stage interpretation of our monopolist problem. The two-stage game resembles the price-and-quantity competition of \citet*{kreps_quantity_1983}, in which every firm commits to a capacity before competing in prices, given a rationing rule and absent buyer private information. In our model, buyers have single-unit demand and private information, so the nature of competition is starkly different. We discuss alternative models of competition in Section \ref{subsec:competition}.

 	\section{Model and results}\label{sec:model}
 	
 	\subsection{Model}\label{subsec:model}
 	
 	A seller (she) faces a unit mass of buyers. Each buyer (he) has a
 	type $\theta$ that is distributed according to the distribution function
 	$F$; $F$ is twice differentiable, has support $\Theta\coloneqq[0,1]$,
 	and has increasing virtual value $\varphi\colon\theta\mapsto\theta-(1-F(\theta))/F'(\theta)$.
 	The payoff of type $\theta$ from quality $q\in Q\coloneqq\mathbb{R}_{+}$
 	and transfer $t\in\mathbb{R}$ is $u(q,\theta)-t$, in which the utility
 	from quality is $u(q,\theta)=g(q)+\theta q$, for a twice-differentiable,
 	increasing, and strictly concave $g\colon Q\to\mathbb{R}$
 	with $g(0)=0$. An \emph{allocation} is a measurable $\bm{q}\colon\Theta\to Q$,
 	and the cost of producing and distributing according to $\bm{q}$
 	is $c(\sup_{\theta\in\Theta}\bm{q}(\theta))$, for an increasing,
 	strictly convex, and differentiable $c\colon Q\to\mathbb{R}_{+}$;
 	we denote $\sup_{\theta\in\Theta}\bm{q}(\theta)$ by $\sup\bm{q}$.
 	We assume that $g$ and $c$ satisfy the Inada conditions $\lim_{q\to\infty}g'(q)=c'(0)=0$
 	and $\lim_{q\to\infty}c'(q)=g'(0)=\infty$.
 	
 	\subsection{Efficiency}
 	
 	An allocation $\bm q$ is \emph{efficient} if it maximizes total welfare, that is, $\bm q$ maximizes \allowbreak $\int_\Theta u(\bm q (\theta) , \theta ) \diff F(\theta)%
 	-%
 	c(\sup \bm q)$.
 	\begin{prop}
 		\label{prop:efficient}Let $q^{\star}$ be the quality $q$ solving
 		$g'(q)+\int_{\Theta}\theta\diff F(\theta)=c'(q)$. The efficient allocation
 		is given by $\bm{q}^{\star}(\theta)=q^{\star}$ for all $\theta$.
 	\end{prop}
 	The efficiency benchmark in markets for digital goods is characterized
 	by two features. First, the allocation is constant, which follows
 	from the fact that any non-constant allocation can be improved upon
 	by giving the maximum quality of the allocation to every type. Total
 	costs are unchanged, and total utility increases. This feature is
 	a departure from traditional markets in which the efficient allocation
 	is increasing. Second, the quality $q^{\star}$ is characterized by
 	equating the marginal cost of production with its marginal social
 	value: the average marginal utility in the population.
 	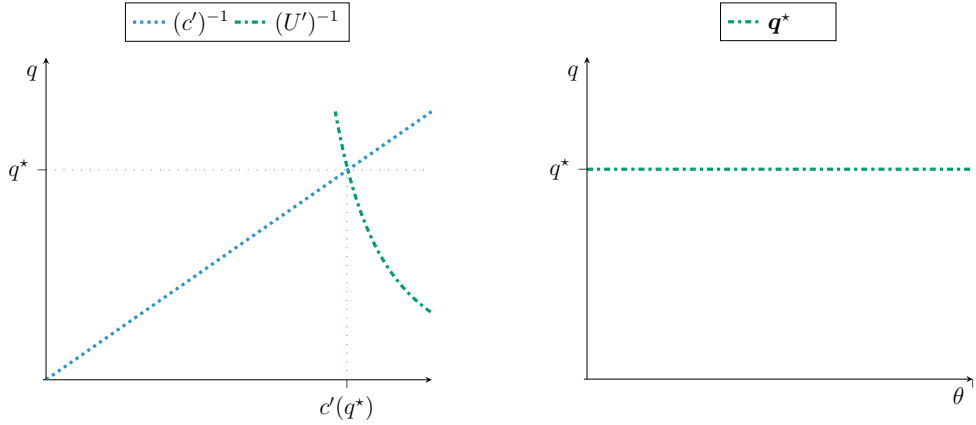
\begin{figure}[t]
 		\centering%
 		\begin{minipage}[t]{0.5\columnwidth}%
 			\centering\subfloat[The quality $q^\star$ equates marginal production cost to average marginal utility. The marginal social value of quality is $U'\colon q\mapsto g'(q)+\int_{\Theta}\theta\diff F(\theta)$.]{
 			\begin{tikzpicture}[scale=0.75]
 				\begin{axis}[
 					y axis line style = {opacity=0, black},
 					axis lines=middle,
 					tick align = outside,
 					ylabel={$q$},
 					ylabel style = {below  left},
 					domain=0.001:4,
 					samples=300,
 					ymin=-0.015, ymax=4.8,
 					xmin=-0.008, xmax=1,
 					restrict y to domain=0:4,
 					xtick={0, 0.78},
 					xticklabels={$0$, $c'(q^\star)$}, 
 					ytick={3.13},
 					yticklabels={$q^\star$},
 					legend style={at={(0.5,1.05)}, anchor=south, legend columns=4},
 					]
 					\addplot[blue, ultra thick, dotted] ({0.25*x}, x);
 					\addlegendentry{$(c')^{-1}$}
 					\addplot[green, ultra thick, dashdotted] ({1/(2*sqrt(x)) + 0.5}, x);
 					\addlegendentry{$(U')^{-1}$}
 					%(U')^{-1}
 					\addplot[loosely dotted, thin, green, domain=0:2] {3.13};
 					\draw[green, loosely dotted] (axis cs:0.78,0) -- (axis cs:0.78,3.13); 
 					%Y AXIS
 					\draw[] (axis cs:0,0) -- (axis cs:0,4.79); 
 					\draw[-{stealth}] (axis cs:0,4.8);
 				\end{axis}
 		\end{tikzpicture}}%
 		\end{minipage}%
 		\begin{minipage}[t]{0.5\columnwidth}%
 			\centering\subfloat[The efficient allocation $\bm q^\star$ gives quality $q^\star$ to every type.]{
 			\begin{tikzpicture}[scale=0.75]
 				\begin{axis}[
 					y axis line style = {opacity=0, black},
 					axis lines=middle,
 					tick align = outside,
 					xlabel={$\theta$},
 					xlabel style={below left},
 					ylabel={$q$},
 					ylabel style={below left},
 					domain=0:1,
 					samples=600,
 					ymin=-0.015, ymax=4.8,
 					xmin=-0.008, xmax=1,
 					restrict y to domain=0:4,
 					xtick={1},
 					xticklabels={$\textcolor{white}{c'(q^\star)}$}, 
 					ytick={3.13},
 					yticklabels={ $q^\star$},
 					legend style={at={(0.5,1.05)}, anchor=south, legend columns=4},
 					]
 					\addplot[green, ultra thick, dashdotted, domain=0:1] {3.13}; \addlegendentry{$\bm q^\star \textcolor{white}{)^-1}$};
 					%Y AXIS
 					\draw[] (axis cs:0,0) -- (axis cs:0,4.79); 
 					\draw[-{stealth}] (axis cs:0,4.8);
 				\end{axis}
 		\end{tikzpicture}
 	}%
 		\end{minipage}\caption{Panel (a) illustrates the efficient quality $q^{\star}$; Panel (b)
 			illustrates the efficient allocation. For these graphs and the following
 			ones, $\theta$ is uniformly distributed, $c(q)=\frac{1}{2}q^{2}$,
 			and $g(q)=\sqrt{{q}}$, unless specified otherwise.}\label{fig:efficiency}
 	\end{figure}
 	Figure \ref{fig:efficiency}
 	illustrates the efficient allocation. We benchmark the monopolist allocation both against the efficient
 	quality $q^{\star}$ and the property that no type receives a damaged
 	good under the efficient allocation.
 	\begin{rem}
 		If the seller observes the type of the buyers, then she can charge
 		price $u(q^{\star},\theta)$ for quality $q^{\star}$ to every type
 		$\theta$. In this way, the seller implements the efficient allocation
 		and extracts all the surplus. In the rest of the paper, we assume
 		that the type of a buyer is her private information.
 	\end{rem}
 	
 	\subsection{The monopolist allocation}
 	
 	The seller maximizes profits by choice of a direct mechanism: a pair
 	of an allocation $\bm{q}$ and a transfer function $t\colon\Theta\to\mathbb{R}$.
 	The seller's problem $\mathcal{P}^{M}$ is
 	\begin{align*}
 		& \sup_{\bm{q},\,t}\int_{\Theta}t(\theta)\diff F(\theta)-c(\sup\bm{q})\ \text{subject to:}\\
 		& u(\bm{q}(\theta),\theta)-t(\theta)\ge u(\bm{q}(\theta'),\theta)-t(\theta')\ \text{{and}}\ u(\bm{q}(\theta),\theta)\ge0,\ \text{for all}\;(\theta,\theta')\in\Theta^{2}.
 	\end{align*}
 	The allocation in a solution to $\mathcal{P}^{M}$ determines the
 	associated transfers by standard arguments, so we identify the solution
 	to the monopolist problem with its allocation, denoted by $\bm{q}^{M}$.
 	The characterization of $\bm{q}^{M}$ does not follow directly from
 	known arguments because we cannot rely on ``pointwise'' analysis.
 	We address this non-separability by transforming $\mathcal{P}^{M}$
 	into a family of separable problems, $(\mathcal{P}(q))_{q\ge0}$,
 	whose values can be weighted by their associated costs and compared
 	across problems.
 	\begin{lem}
 		\label{lem:decomposition}For quality $q$, let the $q$-constrained
 		problem $\mathcal{P}(q)$ be
 		\begin{align*}
 			V(q)=\sup_{\bm{q}}\int_{\Theta}g(\bm q(\theta))+\varphi(\theta)\bm q(\theta)\diff F(\theta)\ \text{subject to:}\\
 			\bm{q}\:\text{is nondecreasing},\;\bm{q}(\theta)\le q,\;\text{for all}\;\theta\in\Theta.
 		\end{align*}
 		The allocation $\bm{q}$ solves $\mathcal{P}^{M}$ if and only if:
 		$\bm{q}$ solves $\mathcal{P}(q^{M})$, where $q^{M}$ solves $\max_{\hat{q}}V(\hat{q})-c(\hat{q})$.
 	\end{lem}
 	The value of the $q$-constrained problem is the maximal revenue that
 	a monopolist obtains by damaging and replicating a quality-$q$ good
 	at no cost. The characterization of $V(q)$ uses the standard approach
 	of replacing the incentive constraints with a monotonicity constraint
 	and obtaining the virtual surplus in the objective. Hence, the ``pointwise''
 	problem $\mathcal{P}(q)$ can be solved using traditional techniques.
 	The cost of doing so is that we need to solve a collection of constrained
 	problems, and compare their values with the quality acquisition costs,
 	which requires to characterize the returns from relaxing the upper-bound
 	constraint.
 	
 	Lemma \ref{lem:decomposition} is important for two reasons. First,
 	the decomposition uncovers a timing in the monopoly provision of digital
 	goods: the seller produces a single good of some quality $q$, and,
 	then, she damages and replicates the quality-$q$ good at no extra
 	costs. Second, since every $\mathcal{P}(q)$ is a separable problem
 	and the acquisition step reduces to a unidimensional maximization,
 	Lemma \ref{lem:decomposition} provides a fundamental simplification
 	of the monopolist problem.
 	\begin{figure}[t]
 		\centering%
 		\begin{minipage}[t]{0.5\columnwidth}%
 			\centering \subfloat[If $q>\bm \beta (0)$, then the marginally bunched type is interior; i.e., $b(q)\in(0, 1)$. Every type $\theta>b(q)$ receives $q$, and type $\theta<b(q)$ receives a damaged quality.]{
 			\begin{tikzpicture}[scale=0.75]
 				\begin{axis}[
 					%axis y line=none,
 					y axis line style = {opacity=0, black},
 					axis lines=middle,
 					tick align = outside,
 					xlabel={ $\theta$},
 					xlabel style={below left},
 					ylabel={},
 					domain=0:1,
 					samples=600,
 					ymin=-0.015, ymax=4.8,
 					xmin=-0.008, xmax=1,
 					restrict y to domain=0:4,
 					xtick={0, 0.295, 0.5},
 					xticklabels={$0$, $\textcolor{black}{b(q)}$, $\varphi^{-1}(0)$}, 
 					ytick={0.25, 1.49,  4},
 					yticklabels={$\bm \beta(0)$, $\textcolor{black}{q}$, $\infty$},
 					legend style={at={(0.5,1.05)}, anchor=south, legend columns=4},
 					]
 					\addplot[black, very thick, domain=0:0.37] {1/(4*(1 - 2*x)^2)}; \addlegendentry{$\bm \beta$};
 					% beta MON	
 					\addplot[red, ultra thick, densely dashed, domain=0.295:1] {1.49};  \addlegendentry{$\bm q$};
 					\addplot[red, loosely dotted, domain=0:0.295] {1.49};
 					\draw[red, loosely dotted] (axis cs:0.295,0) -- (axis cs:0.295,1.49); 
 					\addplot[red,  ultra thick, densely dashed, domain=0:0.295] {min(1/(4*(1 - 2*x)^2), 1.49)};
 					% beta
 					\addplot[black, very thick, domain=0.5:1]{4};
 					\draw[black, loosely dotted] (axis cs:0.5,0) -- (axis cs:0.5,4); 
 					%Y AXIS
 					\draw[densely dotted] (axis cs:0,3.75) -- (axis cs:0,4); 
 					\draw[] (axis cs:0,0) -- (axis cs:0,3.75); 
 					\draw[] (axis cs:0,4) -- (axis cs:0,4.79); 
 					\draw[-{stealth}] (axis cs:0,4.8);
 				\end{axis}
 		\end{tikzpicture}
 	}%
 		\end{minipage}\hfill{}%
 		\begin{minipage}[t]{0.5\columnwidth}%
 			\centering\subfloat[If $q\le\bm \beta (0)$, then full bunching occurs: all types receive $q$. The marginally bunched type is $b(q)=0$.]{\begin{tikzpicture}[scale=0.75]
 					\begin{axis}[
 						%axis y line=none,
 						y axis line style = {opacity=0, black},
 						axis lines=middle,
 						tick align = outside,
 						xlabel={ $\theta$},
 						xlabel style={below left},
 						ylabel={},
 						domain=0:1,
 						samples=600,
 						ymin=-0.015, ymax=4.8,
 						xmin=-0.008, xmax=1,
 						restrict y to domain=0:4,
 						xtick={0, 0.5},
 						xticklabels={$0$, $\varphi^{-1}(0)$}, 
 						ytick={0.17,  4},
 						yticklabels={$q$, $\infty$},
 						legend style={at={(0.5,1.05)}, anchor=south, legend columns=4},
 						]
 						\addplot[black, very thick, domain=0:0.37] {1/(4*(1 - 2*x)^2)}; \addlegendentry{$\bm \beta$};
 						% beta MON	
 						\addplot[red, ultra thick, densely dashed, domain=0:1] {0.17};  \addlegendentry{$\bm q$};
 						% beta
 						\addplot[black, very thick, domain=0.5:1]{4};
 						\draw[black, loosely dotted] (axis cs:0.5,0) -- (axis cs:0.5,4); 
 						%Y AXIS
 						\draw[densely dotted] (axis cs:0,3.75) -- (axis cs:0,4); 
 						\draw[] (axis cs:0,0) -- (axis cs:0,3.75); 
 						\draw[] (axis cs:0,4) -- (axis cs:0,4.79); 
 						\draw[-{stealth}] (axis cs:0,4.8);
 					\end{axis}
 			\end{tikzpicture}
 		}%
 		\end{minipage}\caption{The solution $\bm{q}$ to $\mathcal{P}(q)$ caps the virtual-surplus
 			maximizer $\bm{\beta}$ at $q$. Panel (a) illustrates the solution
 			for $q>\bm{\beta}(0)$; Panel (b) illustrates the solution for $q\le\bm{\beta}(0)$.}
 		\label{fig:beta}
 	\end{figure}
 	To characterize the monopolist allocation,
 	we define the maximizer of the virtual surplus of a zero-cost monopolist,
 	$\bm{\beta}$, and its right-continuous inverse $b$ (Figure \ref{fig:beta}).
 	
 	\begin{definition}
 		The virtual-surplus maximizer is $\bm{\beta}\colon\theta\mapsto\arg\max_{q}g(q)+\varphi(\theta)q$,
 		with $\bm{\beta}(\theta)=\infty$ if $\varphi(\theta)\ge0$; the generalized
 		inverse of the virtual-surplus maximizer is: $b\colon q\mapsto\inf\{\theta\mid\bm{\beta}(\theta)\ge q\}$. 
 	\end{definition}
 	
 	The relationship between the allocation $\bm{\beta}$ and $b$ is
 	governed by
 	\begin{align}
 		g'(q)+\varphi(\theta)=0.\label{eq:b:equation}
 	\end{align}
 	Specifically, $b(q)$ is the type $\theta$ solving the above equation,
 	if $g'(q)\le-\varphi(0)$; and $\bm{\beta}(\theta)$ is the quality
 	$q$ solving the above equation, if $\varphi(\theta)<0$. These objects
 	are derived solely from the preference primitives and fully characterize
 	the monopolist allocation via the decomposition in Lemma \ref{lem:decomposition}. 
 	\begin{prop}
 		\label{prop:monopolist}Let $q^{M}$ be the quality $q$ solving $\left(1-F\left(b(q)\right)\right)\left(b(q)+g'(q)\right)=c'(q)$.
 		The monopolist allocation is given by $\bm{q}^{M}(\theta)=\min\{\bm{\beta}(\theta),q^{M}\}$
 		for all $\theta$. Moreover, it holds that $q^{M}<q^{\star}$. 
 	\end{prop}
 	Here is the intuition for the result. The solution to the $q$-constrained
 	problem is determined by capping the virtual-surplus maximizer at
 	the quality $q$.
 	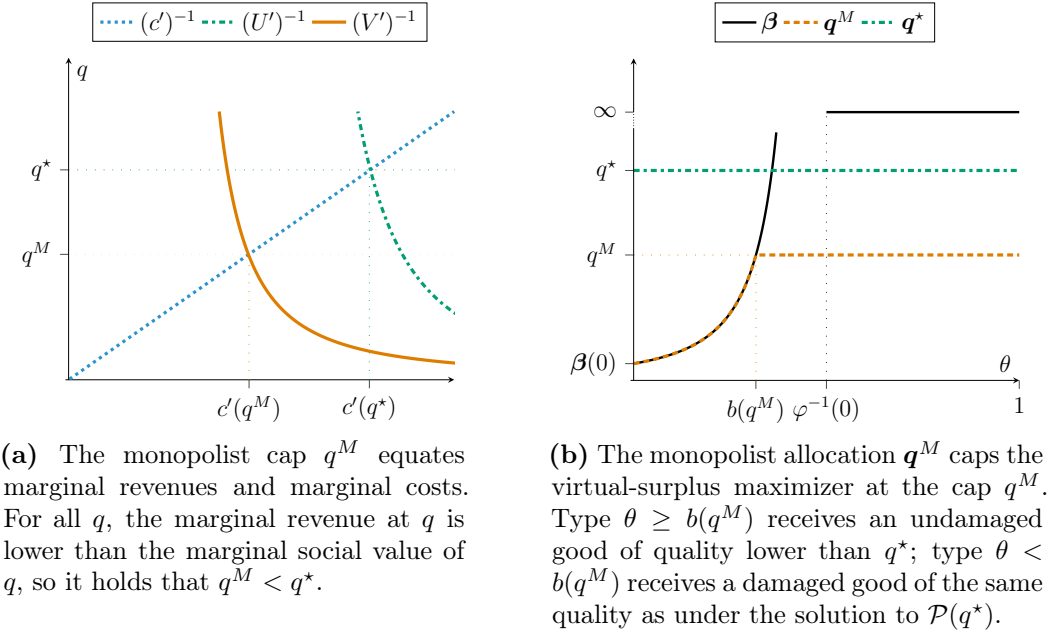
\begin{figure}[t]
 		\centering%
 		\begin{minipage}[t]{0.5\columnwidth}%
 			\centering \subfloat[The monopolist cap $q^M$ equates marginal revenues and marginal costs. For all $q$, the marginal revenue at $q$ is lower than the marginal social value of $q$, so it holds that $q^M<q^\star$.]{
 			\begin{tikzpicture}[scale=0.75]
 				\begin{axis}[
 					y axis line style = {opacity=0, black},
 					axis lines=middle,
 					tick align = outside,
 					ylabel={$q$},
 					domain=0.01:4,
 					samples=300,
 					ymin=-0.015, ymax=4.8,
 					xmin=-0.008, xmax=1,
 					restrict y to domain=0:4,
 					xtick={0, 0.468, 0.78},
 					xticklabels={$0$, $c'(q^M)$, $c'(q^\star)$}, 
 					ytick={1.87, 3.13},
 					yticklabels={$q^M$, $q^\star$},
 					legend style={at={(0.5,1.05)}, anchor=south, legend columns=4},
 					]
 					\addplot[blue, ultra thick, dotted] ({0.25*x}, x);
 					\addlegendentry{$(c')^{-1}$}
 					\addplot[green, ultra thick, dashdotted] ({1/(2*sqrt(x)) + 0.5}, x);
 					\addlegendentry{$(U')^{-1}$}
 					\addplot[red,  ultra thick] ({
 						(1 - 0.5*(1 - 1/(2*sqrt(x)))) * (1/(2*sqrt(x)) + 0.5*(1 - 1/(2*sqrt(x))))
 					}, x);
 					\addlegendentry{$(V')^{-1}$}
 					%qualities
 					\addplot[loosely dotted, ultra thin, red, domain=0:2] {1.87};
 					\addplot[loosely dotted, thin, green, domain=0:2] {3.13};
 					% c prime's 
 					\draw[red, loosely dotted] (axis cs:0.468,0) -- (axis cs:0.468,1.87); 
 					\draw[green, loosely dotted] (axis cs:0.78,0) -- (axis cs:0.78,3.13);
 					%Y AXIS
 					\draw[] (axis cs:0,0) -- (axis cs:0,4.79); 
 					\draw[-{stealth}] (axis cs:0,4.8);
 				\end{axis}
 		\end{tikzpicture}
 		}%
 		\end{minipage}\hfill{}%
 		\begin{minipage}[t]{0.5\columnwidth}%
 			\centering\subfloat[The monopolist allocation $\bm q^M$ caps the virtual-surplus maximizer at the cap $q^M$. Type $\theta\ge b(q^M)$ receives an undamaged good of quality lower than $q^\star$; type $\theta< b(q^M)$ receives a damaged good of the same quality as under the solution to $\mathcal P(q^\star)$.]{
 			\begin{tikzpicture}[scale=0.75]
 				\begin{axis}[
 					%axis y line=none,
 					y axis line style = {opacity=0, black},
 					axis lines=middle,
 					tick align = outside,
 					xlabel={$\theta$},
 					ylabel={},
 					domain=0:1,
 					samples=600,
 					ymin=-0.015, ymax=4.8,
 					xmin=-0.008, xmax=1,
 					restrict y to domain=0:4,
 					xtick={0, 0.317, 0.5, 1},
 					xticklabels={$0$, $b( q^M)$, $\varphi^{-1}(0)$, $1$}, 
 					ytick={0.25, 1.87, 3.13, 4},
 					yticklabels={$\bm \beta(0)$, $  q^M$, $q^\star$, $\infty$},
 					legend style={at={(0.5,1.05)}, anchor=south, legend columns=4},
 					]
 					\addplot[black, very thick, domain=0:0.37] {1/(4*(1 - 2*x)^2)}; \addlegendentry{$\bm \beta$};
 					\addplot[red,  ultra thick, densely dashed, domain=0:0.4] {min(1/(4*(1 - 2*x)^2), 1.87)}; \addlegendentry{$\bm q^M$};
 					\addplot[green,  ultra thick, dashdotted, domain=0:1] {3.13}; \addlegendentry{$\bm q^\star$};
 					% beta
 					\addplot[black, very thick, solid, domain=0.5:1]{4};
 					\draw[black, loosely dotted] (axis cs:0.5,0) -- (axis cs:0.5,4); 
 					% beta MON
 					\addplot[red, ultra thick, densely dashed, domain=0.4:1] {1.87};
 					\addplot[red, loosely dotted, domain=0:0.317] {1.87};
 					\draw[red, loosely dotted] (axis cs:0.317,0) -- (axis cs:0.317,1.87); 
 					%Y AXIS
 					\draw[densely dotted] (axis cs:0,3.75) -- (axis cs:0,4); 
 					\draw[] (axis cs:0,0) -- (axis cs:0,3.75); 
 					\draw[] (axis cs:0,4) -- (axis cs:0,4.79); 
 					\draw[-{stealth}] (axis cs:0,4.8);
 				\end{axis}
 		\end{tikzpicture}
 	}%
 		\end{minipage}\caption{Panel (a) illustrates the monopolist quality $q^{M}$; Panel (b) illustrates
 			the monopolist allocation.}
 		\label{fig:monopoly}
 	\end{figure}
 	This operation implies that $\bm{q}^{M}$ takes the form of $\theta\mapsto\min\{\bm{\beta}(\theta),q^{M}\}$
 	for a cap $q^{M}$ that trades off the revenues $V(q^{M})$ with the
 	acquisition cost $c(q^{M})$, i.e., the value of $\mathcal{P}(q^{M})$
 	with the cost of a relaxation of the associated upper-bound constraint
 	(Figure \ref{fig:monopoly}). The rest of Proposition \ref{prop:monopolist}
 	follows from the characterization of the marginal value $V'$ from
 	a quality investment that we describe in what follows.
 	
	Certain features of the monopolist allocation follow directly from properties of the function $\bm{\beta}$. First, the truncation construction is incentive compatible because $\bm{\beta}$ is increasing on $[0,\varphi^{-1}(0))$. Second, because $\bm{\beta}(0)>0$, the quality $q^{M}$ can be lower than $\bm \beta(0)$, so that the optimal allocation exhibits \emph{full bunching}: all types receive the undamaged quality $q^{M}$.\footnote{Section \ref{subsec:discussion} provides sufficient conditions for full bunching in terms of the primitives of the model, e.g., steep marginal costs.} Finally, because $\bm{\beta}(\theta)=\infty$  for $\theta\ge\varphi^{-1}(0)$ (see Figure \ref{fig:beta}), a positive measure of types in $[\varphi^{-1}(0),1]$ receive the undamaged quality, regardless of the precise level of $q^{M}$. By contrast, lower types receive a damaged version if $q^{M}>\bm{\beta}(0)$; in this case the monopolist assigns damaged quality to types below $b(q^{M})$ and bunches all types above $b(q^{M})$. Therefore, we refer to $b(q^{M})$ as the \emph{marginally bunched type}.
 	
 	The key to understand the expression for $V'(q)$ is the observation
 	that, following an increment in the cap $q$, the seller makes the
 	same revenues from selling all qualities below $q$. The reason is
 	that qualities below $q$ are allocated to the same types and at the
 	same price. The fact that low qualities go to the same types follows
 	from the solution $\theta\mapsto\min\{\bm (\theta),q\}$
 	to $\mathcal{P}(q)$: if the cap is not binding, an increase in the
 	cap is immaterial. Low qualities sell at the same price because rents
 	accumulate from below by incentive compatibility: if the allocation
 	of types below $\theta$ does not change, then the transfer $t(\theta)$
 	does not change. Hence, the extra revenues come from serving types
 	in the bunching region: the increment in $q$ is distributed undamaged
 	to all types above $b(q)$, and with an extra price equal to the marginal
 	utility of type $b(q)$. Therefore, marginal revenues are
 	\begin{align}
 		V'(q)=(1-F(b(q)))(g'(q)+b(q)),\label{eq:marginalrevenues}
 	\end{align}
 	in which the marginal utility of $b(q)$ is weighted by the mass of
 	the bunching region. The function $b$ is crucial. First, $b$ determines
 	production because it shapes the marginal revenues in Equation \ref{eq:marginalrevenues}.
 	Second, $b$ impacts the distribution because, effectively, the monopolist
 	allocation $\bm{q}^{M}$ is defined by the cap $q^{M}$ and two distributional
 	conditions: every quality $q$ below the cap $q^{M}$ goes to its
 	``natural'' type $b(q)$, and the cap goes to all types above
 	$b(q^{M})$.
 	
 	For the ``moreover'' part, we argue that the marginal revenues at
 	$q$ are lower than the average marginal utility from quality $q$.
 	By Markov's inequality, we have $b(q)(1-F(b(q)))<\mathbb{E}\{\theta\}$,
 	so we conclude that $V'(q)<g'(q)+\mathbb{\mathbb{E}}\{\theta\}$.
 	Hence marginal revenues are uniformly below the marginal social value
 	of a quality investment, i.e., $V'$ crosses the marginal cost below
 	$q^{\star}$.
 	
 	Proposition \ref{prop:monopolist} uses the fact that $V$ is concave
 	and differentiable (Proposition \ref{app:prop:mon:alloc}). Under
 	regularity, full bunching does not preclude differentiability of $V$,
 	because $V'$ is continuously pasted at $\bm{\beta}(0)$. The revenue
 	function $V$ is not differentiable if $F$ is not regular. Intuitively,
 	suppose that types in $(\theta',\theta'')$ are bunched ``at'' $q>0$
 	by a monopolist who solves $\mathcal{P}(\overline{q})$ for $\overline{q}>q$,
 	and the monopolist allocation is increasing for all other types. In
 	this case, the extra revenues of adding $\varepsilon>0$ to $q$ come
 	from bunching at the top only types that strictly exceed $\theta''$,
 	whereas the revenues lost from degrading $q$ by $\varepsilon$ come from including types lower than $\theta'$ in the bunching region.
 	We conclude that $\lim_{\varepsilon\to0^{+}}V'(q-\varepsilon)-V'(q+\varepsilon)>0$.
 	In Appendix \ref{app:sec:ironing}, we characterize the monopolist allocation without regularity. The shape of the solution to $\mathcal{P}(q)$ and the fact that $q^{M}<q^{\star}$ hold with more general $F$ and increasing-differences $u$, see Proposition \ref{app:prop:mon:alloc} and \ref{app:prop:mon:ironing}, respectively.

	\begin{remark}
		\label{rem:tariff}Let's argue that the type $b(q)$ maximizes $\theta\mapsto(1-F(\theta))(\theta+g'(q))$.
		The $q$-constrained problem can be equivalently stated as a choice
		of a tariff $T\colon[0,q]\to\mathbb{R}$ determining the price of
		every quality \citep*{guesnerie_complete_1984}. The tariff $T$ is
		optimal only if the \emph{price increment at $q$}, $T'(q)$, solves
		$\max_{p}p\operatorname{Pr}_{}(\{\theta:u_{q}(q,\theta)\ge p\})$,
		by known arguments under regularity conditions \citep*{wilson_nonlinear_1993}.
		Intuitively, every quality increment is priced as if it constitutes
		a separate market. From the optimal tariff corresponding to the solution
		$(\bm{q},t)$ to $\mathcal{P}(q)$---i.e., the mapping $T\colon q\mapsto t(b(q))$---we compute the price increment at $q$, that is $T'(q)=g'(q)+b(q)$.
		Heuristically, if type $\theta$ purchases an interior quality, then the first-order condition holds for the given tariff. Hence, the price increment $g'(q)+b(q)$
		solves $\max_{p}p(1-F(p-g'(q)))$. Equivalently, type $b(q)$ solves
		$\max_{\theta}(1-F(\theta))(\theta+g'(q))$. 
	\end{remark}

 	\paragraph{Inefficiencies}\label{subsec:Inefficiencies}
 	
 	For high types, distributional efficiency holds, because all types
 	above $b(q^{M})$ are allocated an undamaged quality. However, productive
 	efficiency does not hold, because the undamaged quality $q^{M}$ is
 	lower than $q^{\star}$. Under separable costs, the top type does
 	not impose information rents to other types, so his quality is efficient
 	because his information rent is traded-off with the cost of producing
 	only his quality. For digital goods, instead, the monopolist accounts
 	for the productive externality that the top-type quality has for the
 	quality of the other types. Hence, the highest quality determines
 	the information rents of types in the bunching region. As a result,
 	the standard ``efficiency at the top'' observation is weakened to
 	distributional efficiency. Moreover, the distributional efficiency
 	is more ``prevalent'' than with separable costs, in which case efficiency
 	holds only for the top type, because all types above $b(q^{M})$ receive
 	the undamaged quality. 
 	
 	The types who are subject to the two inefficiencies can be read off
 	the graph in Figure \ref{fig:monopoly}. The marginally bunched type
 	partitions the type space into two regions. Each region is associated
 	with an auxiliary economy in which one inefficiency is shut down:
 	(i) the seller produces $q^{M}$ and distributes $q^{M}$ to everyone,
 	and (ii) the seller produces $q^{\star}$ and distributes damaged
 	qualities to maximize revenues, i.e., as in the solution to $\mathcal{P}(q^{\star})$.
 	Types below $b(q^{M})$ receive a damaged good, so they are subject
 	to distributional inefficiency. Nonetheless, they are unaffected by
 	productive inefficiency, because they receive the same quality by
 	a monopolist producing the efficient quality as under $\bm{q}^{M}$.
 	No type above $b(q^{M})$ is subject to distributional inefficiency,
 	because he receives an undamaged good; however, he is subject to productive
 	inefficiency, because he receives a quality exceeding $q^{M}$ by
 	a seller that produces efficiently.
 	\begin{rem}
 		\label{rem:mon:linear}With linear utility, i.e., $u(q,\theta)=\theta q$,
 		the marginally bunched type is $b(q)=\varphi^{-1}(0)$; and we have:
 		$\bm{\beta}(\theta)=0$ for all $\theta<\varphi^{-1}(0)$ (Figure
 		\ref{fig:linear}).
 		\begin{figure}[t]
 			\centering%
 			\noindent\begin{minipage}[t]{1\columnwidth}%
 				\centering \subfloat{
 				\begin{tikzpicture}[scale=0.75]
 					\begin{axis}[
 						%axis y line=none,
 						y axis line style = {opacity=0, black},
 						axis lines=middle,
 						tick align = outside,
 						xlabel={ $\theta$},
 						ylabel={},
 						domain=0:1,
 						samples=200,
 						ymin=-0.015, ymax=4.8,
 						xmin=-0.008, xmax=1,
 						restrict y to domain=0:4,
 						xtick={0, 0.5, 1},
 						xticklabels={$0$, $\varphi^{-1}(0)$, $1$}, 
 						ytick={1, 4},
 						yticklabels={$q^M$, $\infty$},
 						legend style={at={(0.5,1.05)}, anchor=south, legend columns=4},
 						]
 						\addplot[black, very thick, domain=0:0.5] {0}; \addlegendentry{$\bm \beta$};
 						\addplot[red,  ultra thick, densely dashed, domain=0:0.5] {0}; \addlegendentry{$\bm q^M$};
 						% beta
 						\addplot[black, very thick, domain=0.5:1]{4};
 						\draw[black, loosely dotted] (axis cs:0.5,1) -- (axis cs:0.5,4); 
 						% beta MON
 						\addplot[red, ultra thick, densely dashed,domain=0.5:1] {1};
 						\addplot[red, loosely dotted, domain=0:0.5] {1};
 						\draw[red, loosely dotted] (axis cs:0.5,0) -- (axis cs:0.5,1); 
 						%Y AXIS
 						\draw[densely dotted] (axis cs:0,3.75) -- (axis cs:0,4); 
 						\draw[] (axis cs:0,0) -- (axis cs:0,3.75); 
 						\draw[] (axis cs:0,4) -- (axis cs:0,4.79); 
 						\draw[-{stealth}] (axis cs:0,4.8);
 					\end{axis}
 			\end{tikzpicture}}%
 			\end{minipage}\hfill{}\caption{With linear utility, the problem $\mathcal{P}(q)$ is solved by excluding
 				types lower than the zero of the virtual value $\varphi$, and allocating
 				$q$ to higher types.}
 			\label{fig:linear}
 		\end{figure}
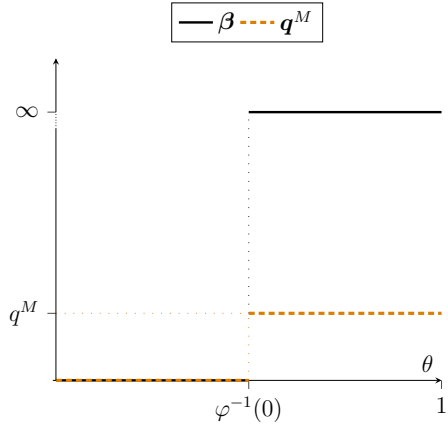
 		Damaging takes the trivial form of excluding low types, because $\bm{q}^{M}(\theta)=q^{M}$
 		if $\theta>\varphi^{-1}(0)$, and $\bm{q}^{M}(\theta)=0$ if $\theta<\varphi^{-1}(0)$. 
 	\end{rem}
 	
 	\subsection{Discussion of the model primitives}\label{subsec:discussion}
 	
 	In markets with our cost structure, richness of versions is driven
 	by preferences alone. In particular, the monopolist allocation with
 	linear preferences (Remark \ref{rem:mon:linear}) can be implemented
 	by offering a single quality at a posted price. The reason is that
 	the curvature of the cost function does not directly impact the distributive
 	properties of the monopolist allocation (Lemma \ref{lem:decomposition}),
 	and only determines the highest quality. Hence, the rich variety of
 	versions observed in markets for digital goods \citep*{bergemann_markets_2019}
 	should be attributed to extra curvature in the utility from quality
 	than that allowed by linear preferences.
 	
 	The type-independent utility shifter given by $g$ is a parsimonious
 	addition to the multiplicative preferences in \citet*{mussa_monopoly_1978}
 	that permits our screening model to be applicable to digital markets.
 	The reason is that the seller finds it more profitable to discriminate
 	across types as we climb the quality ladder, for reasons that are
 	not related to costs: the ratio of the marginal utility of type $\theta$
 	to the marginal utility of a lower type $\theta'$ is increasing in
 	quality, i.e., $\frac{g'(q)+\theta}{g'(q)+\theta'}$ increases in
 	$q$. This pattern aligns with the way digital goods are employed
 	in basic and advanced tasks, whenever buyers are more heterogeneous
 	in advanced tasks than basic ones. In this interpretation, the heterogeneity
 	in relative marginal utilities vanishes around $0$ because quality
 	increments in a neighborhood of $0$ are ``infinitely'' valuable
 	for the accomplishment of the basic task. Instead, the heterogeneity
 	becomes more pronounced as $q$ grows, because quality increments
 	are used almost exclusively in the accomplishment of the professional
 	task.
 	\begin{rem}
 		In a model of separable screening, the curvature of the cost function
 		directly affects the shape of the monopolist allocation. Under linear
 		preferences, the optimality of a single-quality menu arises only if
 		costs are linear. \citet*{sandman_sparse_2025}, and related literature,
 		provides sufficient conditions and intuition for the optimality of
 		single-quality menus in general setups with separable costs.
 	\end{rem}
 	
 	\paragraph{Comparative statics}
 	
 	We study how the monopolist allocation responds to changes in production costs and utility curvature introducing the parameter $\kappa=(\kappa_{c},\kappa_{g})$ that
 	enters the model via $c_{\kappa}(q)\coloneqq\kappa_{c}c(q)$ and $u_{\kappa}(q,\theta)\coloneqq\kappa_{g}g(q)+\theta q$.
 	The value of $\kappa_{c}$ shifts the marginal costs of production,
 	whereas $\kappa_{g}$ shifts the importance of the common curvature
 	of the utility function. Hence, an increase in $\kappa_{g}$ reduces
 	the importance of preference heterogeneity. By Proposition \ref{prop:monopolist},
 	the monopolist allocation is determined as follows. First, we define
 	$b_{\kappa_{g}}\colon Q\to\Theta\ \text{such that}\ \kappa_{g}g'(q)+\varphi(b_{\kappa_{g}}(q))=0\ \text{for all}\ q\in Q$
 	and $b_{\kappa_{g}}(q)=0$ if $\kappa_{g}g'(q)\le-\varphi(0)$, with
 	right-continuous inverse $\bm{\beta}_{\kappa_{g}}$. The monopolist
 	allocation $\bm{q}_{\kappa}^{M}$ is defined the following conditions,
 	
 	\begin{align*}
 		\begin{cases}
 			q^M_\kappa\in Q \ \text{is such that} \ \kappa_c c'(q^M_\kappa) = (1-F(b_{\kappa_g}(q^M_\kappa)))(\kappa_gg'(q^M_\kappa)+b_{\kappa_g}(q^M_\kappa)),\\
 			\text{for all} \ \theta, \ \bm q^M _\kappa (\theta) = \min\{\bm \beta_{\kappa_g}(\theta), q^M_\kappa\}.
 		\end{cases}
 	\end{align*}
 	
 	As $\kappa_{c}$ increases, the acquired quality decreases and the
 	function $b_{\kappa_{g}}$ does not change. Hence, the allocation
 	of all but the low types is pushed downwards, and the bunching region
 	expands: $b_{\kappa_{g}}(q_{\kappa}^{M})$ is decreasing in $\kappa_{c}$.
 	
 	An increase in $\kappa_{g}$ affects the allocation both through the
 	distribution, shifting the virtual-surplus maximizer upwards, and
 	the acquisition condition, via marginal revenues. The acquisition
 	effect has a direct channel---the marginal utility is parametrized
 	by $\kappa_{g}$---and an indirect channel---the function $b_{\kappa_{g}}$
 	depends on $\kappa_{g}$ and determines the marginal revenue. Overall,
 	the monopolist cap increases and the bunching region expands as $\kappa_{g}$
 	grows.
 	\begin{prop}
 		\label{prop:mon:comparative}The cap $q_{\kappa}^{M}$ is decreasing
 		in $\kappa_{c}$ and nondecreasing in $\kappa_{g}$. Moreover, $b_{\kappa_{g}}$
 		is nonincreasing in $\kappa_{g}$ pointwise; $b_{\kappa_{g}}(q_{\kappa}^{M})$
 		is nonincreasing in $\kappa_{g}$, decreasing if interior, and there
 		exists $\overline{\kappa}_{g}$ such that $b_{\kappa_{g}}(q_{\kappa}^{M})=0$
 		for all $\kappa_{g}>\overline{\kappa}_{g}$.
 	\end{prop}
 	The economically relevant case of non-degenerate screening is due
 	to intermediate curvature in the utility from quality: preferences
 	reduce to the linear benchmark if $\kappa_{g}=0$, and screening collapses
 	to a pooling contract as $\kappa_{g}$ grows. Therefore, rich contracts
 	in digital monopolies require some, but not overwhelming, curvature
 	in utility. 
 	
 	The comparative statics implies that full bunching is optimal in two
 	cases: a large $\kappa_{g}$ and a large $\kappa_{c}$. Intuitively,
 	the optimality of full bunching is related to the relative importance
 	of the common curvature in the utility and the strength of the type
 	heterogeneity. First, as $\kappa_{g}$ increases, the common curvature
 	becomes mechanically more important because $\kappa_{g}g'(q)$ increases,
 	and full bunching is more profitable. Due to this channel, full bunching
 	arises with a high monopolist cap because $q_{\kappa}^{M}$ is increasing
 	in $\kappa_{g}$. Second, as $\kappa_{c}$ increases, the highest
 	quality $q_{\kappa}^{M}$ decreases. Hence, the term $\kappa_{g}g'(q)$
 	increases endogenously via the range of qualities that are allocated
 	by the monopolist as $\kappa_{c}$ grows. Due to this channel, full
 	bunching arises with a low cap.
 	
 	The model differs sharply from separable screening in its sensitivity
 	to the type distribution. Because all bunched types receive the cap,
 	a change in the density of high types does not affect the distributive
 	properties of the monopolist allocation. Hence, any distributions
 	$F$ and $G$ that coincide below the threshold given by the marginally
 	bunched type corresponding to $F$ induce the same monopolist allocation.
 	If, in addition, the distributions have the same mean, then the efficient
 	allocation is the same.
 	
 	The comparative statics clarifies the roles of primitives. Costs determine
 	the reach of acquisition via the cap, whereas the curvature of the
 	utility term $g$ shapes both the acquisition and the distribution
 	stages. Distributional perturbations matter via the assignment of
 	quality---below the bunching region and through the mass at the cap---rendering
 	a class of shape changes above some threshold irrelevant for allocation
 	and welfare. These observations stand in contrast to the case of separable-cost
 	screening (Section \ref{subsec:interpretation_separable}).
 	\begin{rem}
 		The main results (Proposition \ref{prop:efficient}, \ref{prop:monopolist},
 		\ref{prop:noscreening}, and \ref{prop:competition}) hold with more
 		general utility functions, as shown in Appendix \ref{app:sec:proofs}. In our comparative
 		statics (Proposition \ref{prop:mon:comparative}), we leverage the
 		simplicity of intuitions arising by expressing type-$\theta$ rents
 		as $\int_{[0,\theta]}\bm{q}(s)\diff s$.
 		The multi-principal model of \citet*{chade_screening_2021} also leverages
 		this tractability. The current $u$ captures more general preferences,
 		given by the gross utility function $q\mapsto g_{1}(q)+\theta g_{2}(q)$
 		for some invertible $g_{2}\in\mathbb{R}^{Q}$ and a concave transformation
 		$g_{1}\in\mathbb{R}^{Q}$ of $g_{2}$, because quality does not admit
 		a natural metric.
 	\end{rem}
 	
 	\section{Interpretation}\label{sec:interpretation}
 	In this section, we discuss our model within the leading interpretation of the digital goods industry, and offer an alternative interpretation. First, we explain that the costs given by $c$ are effectively calibrated to population size in the leading interpretation, consistently with digital markets. Then, we provide an interpretation for the same formulation of the problem of the seller, in which costs and revenues are expressed in the same units. Under this interpretation, the seller faces a single buyer and produces before eliciting the type of the buyer; we compare the monopolist allocation with the one in \citeauthor*{mussa_monopoly_1978} (\citeyear{mussa_monopoly_1978},
 	MR) under this second interpretation.

 	\subsection{Digital goods}\label{subsec:interpretation_digital}
 	 The digital economy offers the most compelling justification for our cost function, because low replication costs are a key feature of the production of digital goods \citep*{goldfarb_digital_2019}. Consider a word processing application offered via subscription, such as Microsoft Word. The developer builds a feature-complete application at a certain cost, and deploys the application across cloud servers. A basic-tier subscriber receives the identical codebase as a premium subscriber, but certain functions---collaboration tools, design templates, or revision history---are deactivated through licensing flags, without additional development or infrastructure cost. For many products, the same technology that enables costless replication makes damaging nearly free. For example, the same computer that can copy a dataset can also delete observations, the compiler that generates executable code can disable functions, the servers that store content available for multiple viewers can stream content at low resolution.

 	A \emph{cost mismatch} arises under the free-damaging and free-replication
 	scenario of the digital-good motivation. Recall the maximand in the
 	monopolist problem, $\int_{\Theta}t(\theta)\diff F(\theta)-c(\sup\bm{q})$.
 	The production cost has the same order of magnitude as the revenue
 	term, and is infinitely larger than the surplus generated by a single
 	buyer. We view this feature as capturing the correct economics for
 	digital goods: the seller chooses how far up the quality ladder to
 	climb, pays the development cost, and can then serve any number of
 	buyers with any lower quality by throttling, feature-gating, or binning.
 	For example, consider a game released in two editions: an all-inclusive
 	``Deluxe'' edition, and a basic edition with fewer features. The game
 	studio only needs to build the Deluxe edition, whereas the basic one
 	is created by removing maps, characters, or game modes. Producing
 	the Deluxe edition is costly, whereas producing the basic edition
 	simply involves withholding ``pieces'' of the existing game, rather
 	than building a separate game, and costs almost nothing. The game
 	studio pays the cost of development once, and can then freely replicate
 	and degrade the Deluxe game for any number of players.
 	
 	This interpretation leads to Lemma \ref{lem:decomposition}: the distribution
 	stage amounts to a pure price discrimination problem on a fixed ladder,
 	and occurs after production has taken place. A key feature of the
 	interpretation is that the development cost is parametrized relative
 	to the population size. Specifically, suppose that we perturb the
 	mass of consumers $\alpha$, away from its unit-mass benchmark in
 	Section \ref{sec:model}.\footnote{Specifically, we consider a seller that maximizes profits $\alpha \int_\Theta t(\theta)\diff F(\theta)-c(\sup \bm q)$, subject to the same constraints as in Section \ref{sec:model}.} The perturbation impacts profits only through
 	the revenue term, contrary to a model à la MR in which $\alpha$ is
 	a profit shifter: both revenues and costs are scaled by $\alpha$
 	in MR. Hence, an increase in $\alpha$ changes the weights of the
 	two addenda in the objective of the seller, reducing the importance of the cost component. Therefore, the impact of an increase in $\alpha$ can be read off the comparative statics in Proposition \ref{prop:mon:comparative} as a decrease in $\kappa_{c}$. With a larger population, the seller increases the quality produced, and determines the allocation moving along the (unchanged) function $\beta$. Conversely, the case where $\alpha$ vanishes, i.e., approximating a single buyer, corresponds $\kappa_c\to \infty$: the monopolist allocation degenerates to full bunching at zero, a formalization of the cost mismatch.
 	
 	This discussion clarifies the role of the production cost in the mechanics of our model: it is a device to discipline the top quality in a model with costless distribution. This selection induces a quality under-provision that interacts with the familiar screening distortions of the distribution
 	stage. The acquisition margin delivers an economic insight that is
 	absent under the ``interim'' perspective that assumes an exogenous
 	top version and studies the distribution stage under free replication
 	and damaging.
 	
 	Because we work with a ``global'' development cost rather than a
 	per-unit cost, the direct comparison with standard screening in this
 	multi-agent interpretation is artificial. We compare our results with
 	the conclusion of a model à la MR under an alternative interpretation
 	in Section \ref{subsec:interpretation_separable}.
 	
 	\subsection{Single agent}\label{subsec:interpretation_separable}
 	
 	In this section, we consider an interpretation of our model in which costs are expressed in the same units
 	as the utility of a single buyer---recall that costs are expressed
 	in per-agent units in the leading interpretation of the model in Section
 	\ref{subsec:interpretation_digital}.
 	
 	Consider a setting with a single buyer, in which production has to precede the type elicitation, in contrast with the single-agent interpretation of MR. When the buyer enters the store, the seller has already produced
 	a version of the good, and screening only occurs through offering damaged versions of the available good.\footnote{Replication is irrelevant with a single buyer, and there is no use of damaging if production occurs after the type elicitation. With a single buyer, it is the interaction of free damaging and pre-production that creates a non-trivial model: adding free damaging but keeping the MR timing does not change the MR solution because there is no need to over produce; if seller has to pre-produce and cannot damage, then she offers a single quality, solving the problem analyzed in Section \ref{subsec:noscreening}.} With this interpretation, the comparison of the allocation $\bm{q}^{M}$ with the monopolist allocation under separable costs does not feature the size mismatch mentioned in Section \ref{subsec:interpretation_digital}. If costs are separable (i.e., in the MR model), the seller maximizes $\int_{\Theta}t(\theta)-c(\bm{q}(\theta))\diff F(\theta)$. The pointwise maximization of virtual surplus yields the optimal allocation, which is $\bm{q}^{MS}\colon\theta\mapsto\max\{(c'-g')^{-1}(\varphi(\theta)),0\}$ (in what follows, we use MS to denote the case of a separable monopolist, i.e., the MR model.) Cost separability provides an accurate description of the production process of several industries. However, there may be a
 	separation between the production and the elicitation stage. For example,
 	in the car industry, dealers do not impact the production process,
 	and, effectively only offer multiple versions that are assembled through
 	add-ons; e.g., sound system, paint finishes, leather seats. In this
 	case, the production costs are such that a baseline ``top version''
 	is made available by the manufacturer, and the dealer screens through damaging.
 	Hence, the comparison between $\bm{q}^{M}$ and $\bm{q}^{MS}$ is attributed
 	to the difference between screening after or before production. 
 	
 	Under the current interpretation, the natural efficiency benchmark
 	equates, for every type, the marginal utility to the marginal cost
 	of production---i.e. $\bm{q}^{E}\colon\theta\mapsto\max\{(c'-g')^{-1}(\theta),0\}$.\footnote{The allocation $q^{\star}$ corresponds to the ``constrained-efficiency'' benchmark in which the planner produces prior to knowing the type of the buyer.} Figure \ref{fig:singleallocation1} guides the comparison among the different allocations.
 	\begin{figure}
 	\centering%
 		\begin{minipage}[t]{0.5\columnwidth}%
 			\centering\subfloat[For this graph: $c(q)=\frac{1}{2}q^2$, so that $b(q^M)>0$.]{\includegraphics[width=0.95\textwidth]{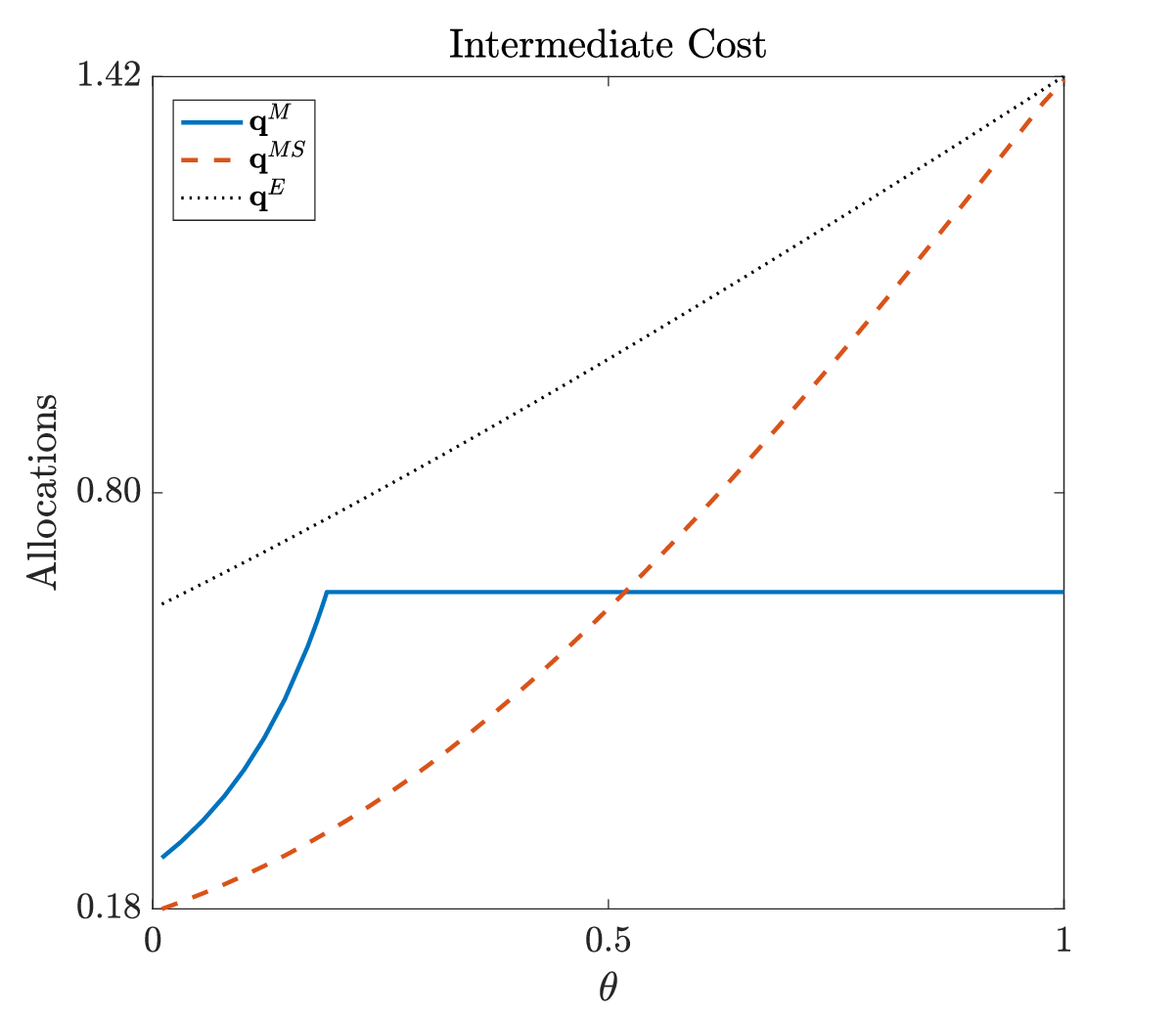}	\label{fig:singleallocation1:right}}
 		\end{minipage}\hfill{}%
 		\begin{minipage}[t]{0.5\columnwidth}%
 			\centering\subfloat[For this graph: $c(q)=2q^2$, so that $b(q^M)=0$.]{\includegraphics[width=1.10\textwidth]{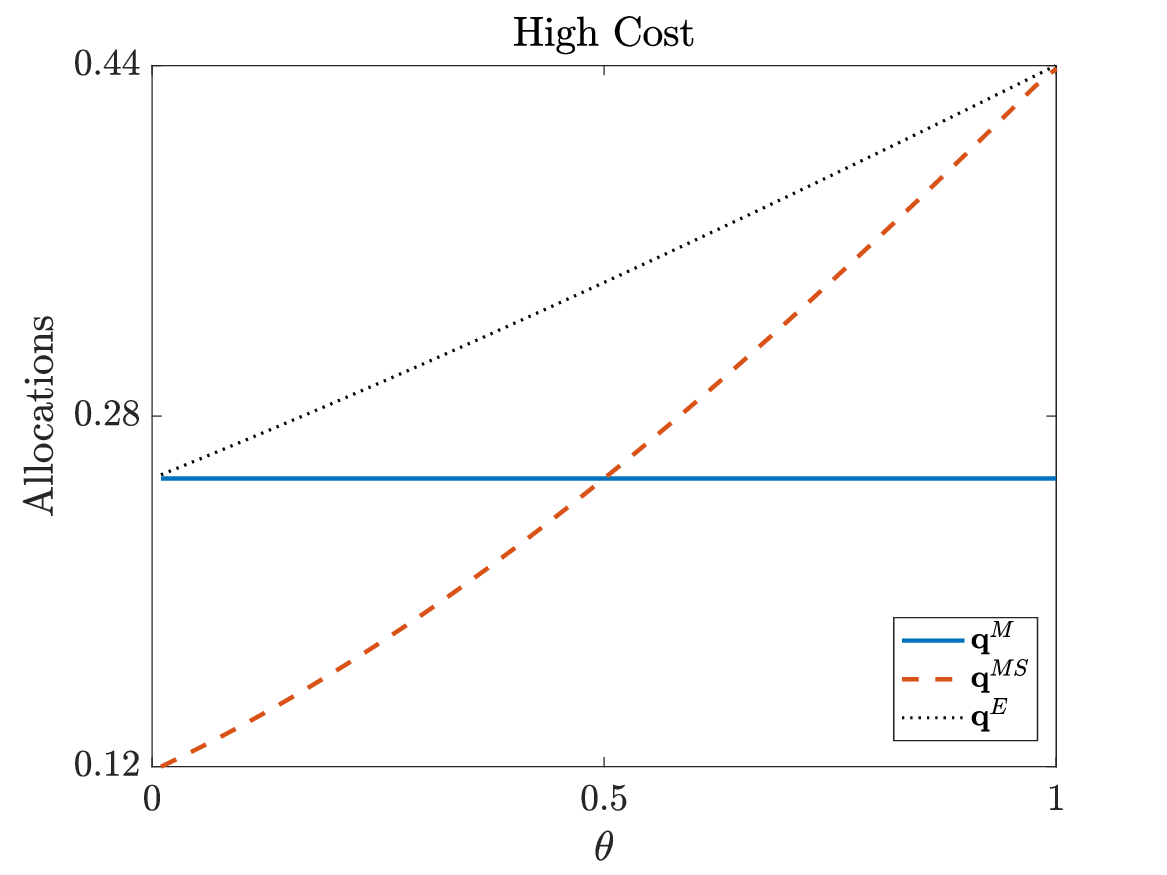}}%
 		\end{minipage}\caption{Panel (a) compares the monopolist allocation with $\bm{q}^{MS}$ if
 			the monopolist does not engage in full bunching, so $b(q^{M})>0$;
 			Panel (b) compares the monopolist allocation with $\bm{q}^{MS}$ if
 			the monopolist optimally offers a pooling contract, so $b(q^{M})=0$.
 			The two figures have different allocations due to different
 			cost functions, see Proposition \ref{prop:mon:comparative}.}
 			\label{fig:singleallocation1}
 	\end{figure}
 	By incentive compatibility, both monopolist allocations are nondecreasing
 	and pointwise below the efficient allocation $\bm{q}^{E}$. If there
 	is no full bunching, the two allocations cross; in particular $\bm{q}^{M}$
 	is higher for low types, and lower for high types, than $\bm{q}^{MS}$.
 	For types that are not bunched, the nonseparable allocation $\bm{q}^{M}$
 	is set to equate the marginal virtual surplus to zero, whereas the
 	cost-separable monopolist equates the marginal virtual surplus to
 	the marginal cost. Importantly, the cost separable monopolist serves
 	efficiently the top type, whereas the non-separable one underprovides
 	quality to the top type; i.e., it holds that $\bm{q}^{MS}(1)=\bm{q}^{E}(1)>q^{\star}>q^{M}=\bm{q}^{M}(1)$.
 	These comparisons are established in the following result.
 	\begin{prop}
 		\label{prop:singleagentallocation} The allocations $\bm{q}^{M}$ and $\bm{q}^{MS}$ cross above $b(q^M)$, that is, it holds that: (i)
 		$q^{M}<\bm{q}^{MS}(1)$; (ii) if $\theta<b(q^{M})$, then: $\bm{q}^{M}(\theta)>\bm{q}^{MS}(\theta)$. Moreover, if $b(q^{M})=0$, then $q^{M}=\bm{q}^{E}(0)$.
 		%The allocations $\bm{q}^{M}$ and $\bm{q}^{E}$ are such that: (i) if $b(q^{M})>0$, then $\bm{q}^{M}<\bm{q}^{E}$ pointwise; (ii) if $b(q^{M})=0$, then $q^{M}=\bm{q}^{E}(0)$.
 	\end{prop}

 	We turn to the welfare comparisons between the two settings.  The seller is worse off by producing before the elicitation stage.
 	Defining the ex-post profits from quality $q$ as $\pi(q,\theta)=g(q)+\varphi(\theta)q-c(q)$,
 	we note that $\bm{q}^{MS}(\theta)$ maximizes $\pi(\cdot,\theta)$,
 	whereas the separable monopolist makes profits $\pi(\bm{q}^{M}(\theta),\theta)+c(\bm{q}^{M}(\theta))-c(q^{M})$.
 	The expression for the profits in the separable case subtracts the
 	nonnegative quantity $c(q^{M})-c(\bm{q}^{M}(\theta))$ to $\pi(\bm{q}^{M}(\theta),\theta)$,
 	which, in turn, is below the ex-post profits in the separable case,
 	$\pi(\bm{q}^{MS}(\theta),\theta)$ (Figure \ref{fig:SAprofits}).
 	\begin{figure}
		\centering\includegraphics[scale=0.4]{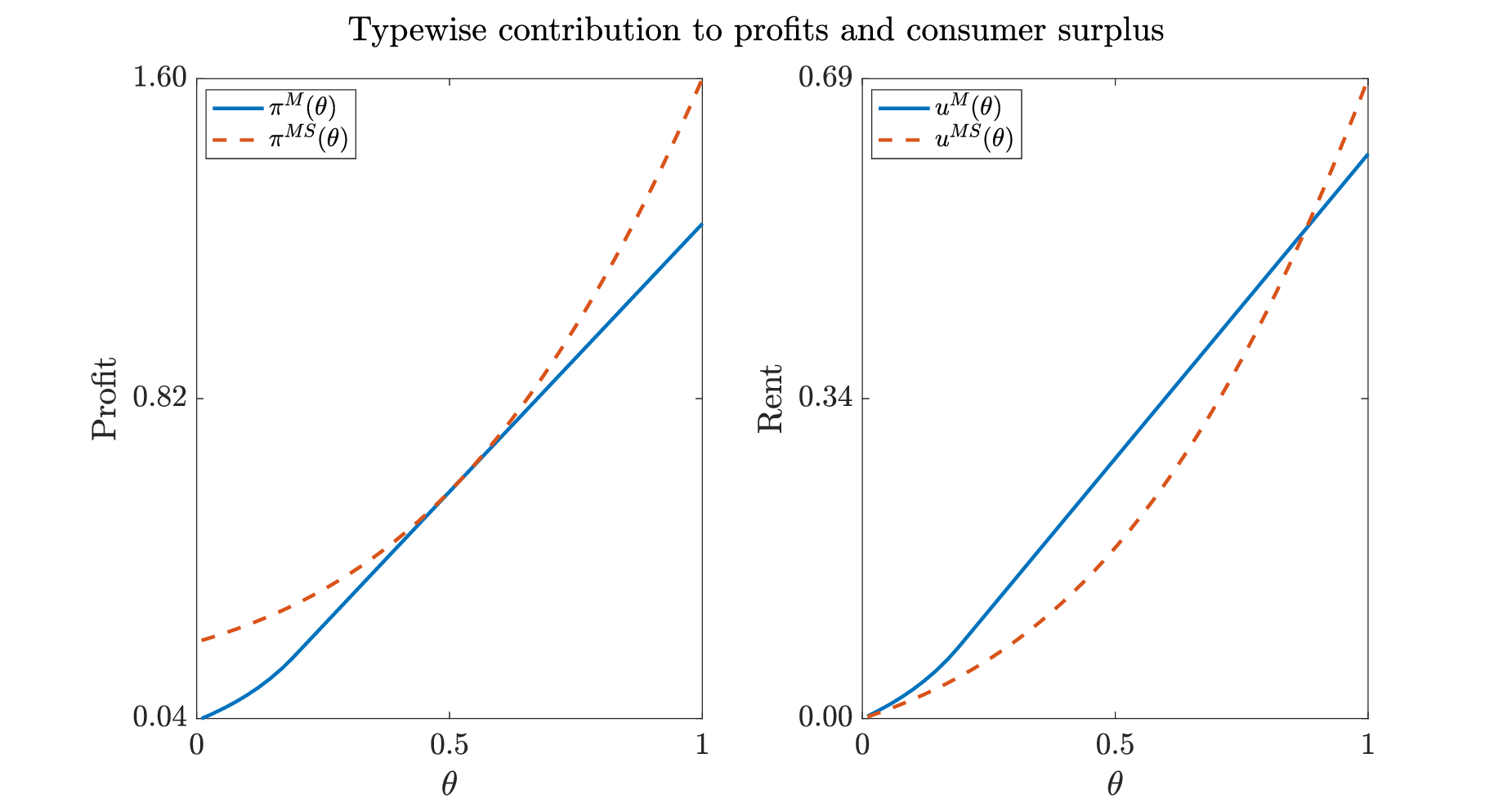}\caption{The left panel shows the profit given the monopolist allocation $\bm{q}^{M}$, and given the MR allocation with separable costs $\bm{q}^{MS}$, defining profits by $\pi^{MS}(\theta)\protect=\pi(\bm{q}^{MS}(\theta),\theta)$
 			and $\pi^{M}(\theta)\protect=\pi(\bm{q}^{M}(\theta),\theta)+c(\bm{q}^{M}(\theta))-c(q^{M})$, for all $\theta$; the right panel shows
 			the consumer surplus, defining information rents by
 			$u^{MS}(\theta)=\protect\int_{[0,\theta]}\bm{q}^{MS}(s)\diff s$ and $u^{M}(\theta)=\protect\int_{[0,\theta]}\bm{q}^{M}(s)\diff s$, for all $\theta$.
 			The seller is better off under separability, whereas the ranking of consumer surplus
 			is ambiguous. The allocation $\bm{q}^{M}$ and $\bm{q}^{MS}$ are shown in Figure \ref{fig:singleallocation1:right}, which uses the same parameters as in this figure.}
 		 		\label{fig:SAprofits}
 	\end{figure}
 	Put differently, the seller is harmed for two reasons: she cannot
 	optimize typewise, and she overpays for low types. The profits are
 	uniformly higher in the separable case, with equality at the type
 	such that the two allocations cross.
 	
 	The implications of nonseparability for consumer surplus are more
 	delicate. The fact that the allocations cross suggests that the ranking
 	of consumer surplus is ambiguous. Low types receive higher rents under
 	$\bm{q}^{M}$ than with separable costs, because of the quality ranking
 	(Proposition \ref{prop:singleagentallocation}). Instead, higher types
 	can receive lower rents under nonseparability. With a uniform type
 	distribution, quadratic costs, and $g(q)=\kappa_{g}\sqrt{q}$, as
 	in Section \ref{subsec:discussion}, the ranking of consumer surplus
 	is determined by $\kappa_{g}$. Specifically, denoting the consumer
 	surplus of allocation $\bm{q}$ by $S(\bm{q})\coloneqq\int_{\Theta}\bm{q}(\theta)(1-F(\theta))\diff\theta$,
 	the difference $S(\bm{q}^{M})-S(\bm{q}^{MS})$ is negative at $\kappa_{g}=0$,
 	increases with $\kappa_{g}$, and is positive for high values of $\kappa_{g}$.
 	The intuition for this ranking builds upon Section \ref{subsec:discussion}.
 	If $u(q,\theta)$ is approximately linear in $q$, then our monopolist
 	excludes low types, as in Figure \ref{fig:linear}, whereas $\bm{q}^{MS}$
 	is increasing, due to the curvature of $c$. This difference explains
 	why the consumer surplus in MR is higher than under nonseparable costs
 	in the case of linear preferences. As $\kappa_{g}$ grows, the bunching
 	region expands by Proposition \ref{prop:mon:comparative}, and the
 	ranking reverses; we observe numerically that the surpluses, as functions of $\kappa_g$ cross at $0.3$, so, nonseparability makes buyers better off if and only if $\kappa_g$ exceeds $0.3$. Interestingly, our full-bunching monopolist exhibits efficiency at the bottom: $\bm{q}^{E}(0)=q^{M}$. The reason is that the seller determines the quality $q^{M}$ by setting
 	the price increment ``at'' $q^{M}$ equal to the marginal
 	utility of the lowest type (Remark \ref{rem:tariff}), which is
 	 the way in which the quality $\bm{q}^{E}(0)$ is determined.

 	\section{Extensions}\label{sec:extensions}
 	
 	\subsection{Non-discriminating monopolist}\label{subsec:noscreening}
 	
 	In certain settings, technological or regulatory constraints make
 	it infeasible for the seller to degrade the good. If the monopolist
 	allocation induces full bunching, then a no-damaging constraint is irrelevant
 	for the monopolist problem.\footnote{Formally, in this section we consider the problem given by $\mathcal P^M$ with an additional constraint that says that the image of a feasible allocation can include 0 and, at most, a single positive quality, chosen optimally; i.e., there exists $q\in Q$ such that, for all $\theta$, $\bm{q}(\theta)\subseteq\{0,q\}$.} The constraint is irrelevant also in the case of linear preferences of Remark \ref{rem:mon:linear}.
 	
 	To characterize the allocation $\bm{q}_{N}^{M}$ of a non-discriminating
 	monopolist, fix a produced quality $q$. The monopolist effectively
 	chooses a type $b_{N}(q)$ such that: lower types are excluded and
 	higher types buy quality $q$. The implied price of the quality-$q$
 	good makes type $b_{N}(q)$ indifferent between buying and the outside
 	option. The resulting revenues are $V_{N}(q)=\max_{\theta}(1-F(\theta))(g(q)+\theta q)$.
 	The resulting marginal revenues account both for the price increment
 	and the inframarginal types, because, by the envelope theorem, we
 	have $V_{N}'(q)=(1-F(b_{N}(q)))(g'(q)+b_{N}(q))$.
 	
 	The structure of the marginal value of relaxing the $q$-constraint
 	for the non-discriminating seller is the same as for the screening
 	seller. The key difference lies in how the relevant cutoff type is
 	determined. Without screening, the cutoff type makes the monopolist
 	indifferent between serving and excluding him; i.e.,~$b_{N}(q)$
 	is the type $\theta$ solving $g(q)+\varphi(\theta)q=0$. With screening,
 	the cutoff type makes the monopolist indifferent between damaging
 	his quality or not; i.e.,~$b(q)$ is the type $\theta$ solving $g'(q)+\varphi(\theta)=0$.
 	The screening ban mechanically shuts down damaging, so intuitively
 	the monopolist sells the undamaged good to a larger set of types than
 	under screening, for fixed cap $q$.\footnote{In particular, we have $g'(q)<g(q)/q$, for all $q>0$, by
 		strict concavity of $g$; hence, we have $b_{N}(q)\le b(q)$, with
 		strict inequality if $b(q)$ is interior, for all $q>0$.}
 	
 	\begin{figure}[t]
 		\centering%
 		\begin{minipage}[t]{0.5\columnwidth}%
 			\centering \subfloat[Without damaging, the monopolist sells an undamaged quality ($q^M_N$) to more types than with damaging, excludes low types, and strengthens  productive inefficiency.]{
 			\begin{tikzpicture}[scale=0.75]
 				\begin{axis}[
 					%axis y line=none,
 					y axis line style = {opacity=0, black},
 					axis lines=middle,
 					tick align = outside,
 					xlabel={ $\theta$},
 					ylabel={},
 					domain=0:1,
 					samples=800,
 					ymin=-0.015, ymax=4.8,
 					xmin=-0.008, xmax=1,
 					restrict y to domain=0:4,
 					xtick={0, 0.2,0.5, 1},
 					xticklabels={$0$, $b_N(q_N^M)$, $\varphi^{-1}(0)$, $1$}, 
 					ytick={0.25, 1.1, 1.49,4},
 					yticklabels={$\bm \beta(0)$, $q_N^M$, $  q^M$, $\infty$},
 					%							yticklabel style = {font=\small},
 					%							xticklabel style = {font=\small},
 					%legend pos=north east, %legend style={font=\small},
 					legend style={at={(0.5,1.05)}, anchor=south, legend columns=4},
 					%semithick
 					extra tick style={% changes for all extra ticks
 						tick align=outside,
 						tick pos=left,
 					},
 					extra x tick style={% changes for extra x ticks
 						major tick length=1.2\baselineskip
 					},
 					extra x ticks={0.295},
 					extra x tick labels={$b( q^M)$},
 					]
 					\addplot[black, very thick, domain=0:0.37] {1/(4*(1 - 2*x)^2)}; \addlegendentry{$\bm \beta$};
 					\addplot[red, ultra thick, densely dashed, domain=0:0.295] {min(1/(4*(1 - 2*x)^2), 1.49)}; \addlegendentry{$\bm q^M$};
 					\addplot[purple, ultra thick, dashdotted, domain=0.2:1] {1.1}; \addlegendentry{$\bm q_N^M$};
 					% beta
 					\addplot[black, very thick, solid, domain=0.5:1]{4};
 					\draw[black, loosely dotted] (axis cs:0.5,0) -- (axis cs:0.5,4); 
 					% beta MON
 					\addplot[red, ultra thick, densely dashed, domain=0.295:1] {1.49};
 					\addplot[red, loosely dotted, domain=0:0.295] {1.49};
 					\draw[red, loosely dotted] (axis cs:0.295,0) -- (axis cs:0.295,1.49); 
 					% beta No-Dama
 					\addplot[purple, ultra thick, dashdotted, domain=0:0.2] {0};
 					\addplot[purple, loosely dotted, domain=0:0.2] {1.1};
 					\draw[purple, loosely dotted] (axis cs:0.2,0) -- (axis cs:0.2,1.1); 
 					%Y AXIS
 					\draw[densely dotted] (axis cs:0,3.75) -- (axis cs:0,4); 
 					\draw[] (axis cs:0,0) -- (axis cs:0,3.75); 
 					\draw[] (axis cs:0,4) -- (axis cs:0,4.79); 
 					\draw[-{stealth}] (axis cs:0,4.8);
 				\end{axis}
 		\end{tikzpicture}
	 	\label{fig:nodamaging}}%
 		\end{minipage}\hfill{}%
 		\begin{minipage}[t]{0.5\columnwidth}%
 			\centering\subfloat[In a subgame 
 			with highest and second-highest qualities $x$ and $y$, resp., every type 
 			$\theta$ purchases quality \protect${\bm q^C {[x, y]}(\theta)}\protect$.]{
 			\begin{tikzpicture}[scale=0.75]
 				\begin{axis}[
 					%axis y line=none,
 					axis lines=middle,
 					tick align = outside,
 					xlabel={$\theta$},
 					domain=0:1,
 					samples=800,
 					ymin=-0.015, ymax=2.3,
 					xmin=-0.008, xmax=0.55,
 					restrict y to domain=0:4,
 					xtick={0, 0.15, 0.25},
 					xticklabels={$0$, $b(y)$, $b(x)$}, 
 					ytick={0.25, 0.5, 1, 1.7},
 					yticklabels={$\bm \beta(0)$, $y$, $x$, $  q^M$},
 					legend style={at={(0.5,1.05)}, anchor=south, legend columns=4},
 					extra tick style={% changes for all extra ticks
 						tick align=outside,
 						tick pos=left,
 					},
 					extra x tick style={% changes for extra x ticks
 						major tick length=1.2\baselineskip
 					},
 					extra x ticks={0.308},
 					extra x tick labels={$b( q^M)$},
 					extra y tick style={% changes for extra x ticks
 						major tick length=1.2\baselineskip
 					},
 					]
 					%\addplot[black, very thick, domain=0:0.37] {1/(4*(1 - 2*x)^2)}; \addlegendentry{$\bm \beta$};
 					\addplot[red,  ultra thick, densely dashed, domain=0:1] {min(1/(4*(1 - 2*x)^2), 1.7)}; \addlegendentry{$\bm q^M$};
 					\addplot[black, ultra thick, dashdotted, domain=0.15:0.25] {min(1/(4*(1 - 2*x)^2), 1.49)}; \addlegendentry{$\bm q ^C[x, y]$};
 					% beta MON
 					\addplot[red, loosely dotted, domain=0:0.308] {1.7};
 					\draw[red, loosely dotted] (axis cs:0.308,0) -- (axis cs:0.308,1.7); 
 					% beta COmp
 					\addplot[black, ultra thick, dashdotted, domain=0.25:1] {1};
 					\addplot[black, ultra thick, dashdotted , domain=0:0.15] {0.5};
 					\addplot[black, loosely dotted, domain=0:0.25] {1};
 					\draw[black,  loosely dotted] (axis cs:0.25,0) -- (axis cs:0.25,1);
 					\draw[black,  loosely dotted] (axis cs:0.15,0) -- (axis cs:0.15,0.5);
 				\end{axis}
 		\end{tikzpicture}
 	\label{fig:competition}}%
 		\end{minipage}\caption{Panel (a) illustrates the monopolist allocation $\bm{q}_{N}^{M}$
 			under a no-damaging constraint; Panel (b) illustrates the allocation
 			$\bm{q}^{C}[x,y]$ that prevails in a subgame in which the highest
 			and second-highest quality are, resp., $x$ and $y$.}
 		\label{fig:extensions}
 	\end{figure}
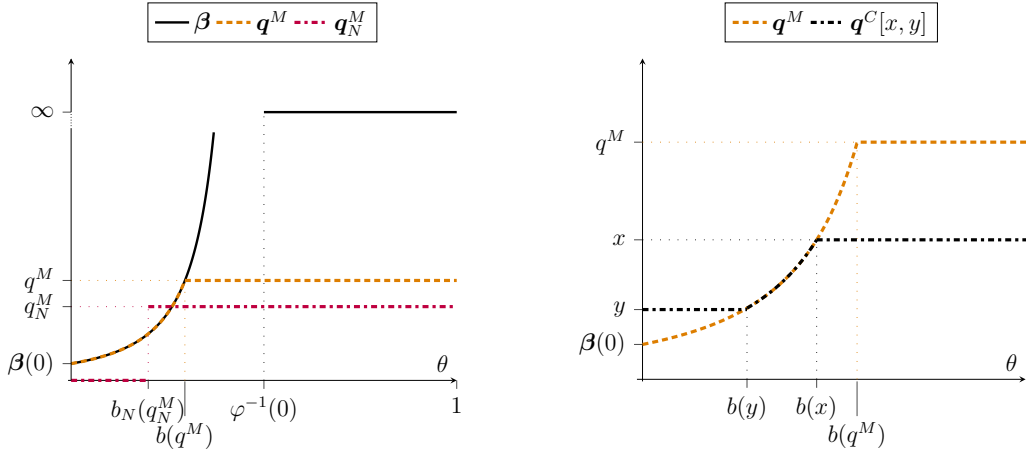
 	Figure \ref{fig:nodamaging} illustrates the optimal allocation. Here is
 	the intuition why the monopolist supplies a lower top quality than
 	under screening.  
 	
 	The marginally bunched type $b(q)$ is the type $\theta$
 	that maximizes the marginal utility $u_{q}(q,\theta)$ weighted by
 	the mass of types above $\theta$, by Remark \ref{rem:tariff}. However,
 	the marginal revenues of the non-discriminating seller account for
 	the fact that $q$ is sold to non-excluded types at a price that increases
 	with the marginal utility of \emph{cutoff} type $b_{N}(q)$. Hence,
 	productive inefficiency gets stronger with the addition of the no-damaging
 	constraint: $V'(q)-V_{N}'(q)\ge0$ for all $q$, with strict inequality
 	whenever the screening seller does not engage in full bunching.
 	
 	Therefore,
 	the highest quality without screening, $q_{N}^{M}$, is lower than
 	$q^{M}$. By the conclusion of the preceding paragraph, distributional
 	efficiency is more prevalent than with screening. In particular, the
 	strengthening of productive inefficiency implies that a larger region
 	of types gets the undamaged good, i.e., $b_{N}(q_{N}^{M})\le b(q^{M})$. 
 	\begin{prop}
 		\label{prop:noscreening}If $b(q^{M})>0$, then it holds that $q_{N}^{M}<q^{M}$
 		and $b_{N}(q_{N}^{M})<b(q^{M})$.
 	\end{prop}
 	The ban may induce exclusion, which does not occur under screening,
 	and makes low types worse off. In particular, the ban makes types
 	in $[0,b_{N}(q_{N}^{M})]$ worse off because these types are excluded
 	under $\bm{q}_{N}^{M}$. Instead, the ban makes low types better off
 	whenever the non-discriminating monopolist does not exclude the bottom
 	type. In particular, the ban makes types in $[0,b(q_{N}^{M})]$ better
 	off if $b_{N}(q_{N}^{M})=0$, because $\bm{q}_{N}^{M}$ allocates
 	$q_{N}^{M}$ to all types lower than $b(q_{N}^{M})$ and $q_{N}^{M}>\bm{q}^{M}(\theta)$
 	for $\theta<b(q_{N}^{M})$.
 	
 	\subsection{Competition}\label{subsec:competition}
 	In this section, we extend our analysis to a competitive setting.
 	We show that a competitive market yields a stochastic allocation that
 	``shrinks'' the monopolist one: the lowest quality is higher, and is
 	offered for free, more types receive an undamaged good, but the highest quality is lower. The reduction of the damaging
 	inefficiency and the increase of the productive inefficiency push welfare in opposite directions. Whether welfare is higher under monopoly than competition depends on the relative strenghts of the two effects, which, in turn, are determined by the shape of the cost function.
 	
 	\paragraph{Description of the game}
 	
 	We study a two-stage game of perfect information among countably many
 	replicas, indexed by $i\in\{1,2,\dots\}$, of the seller in Section
 	\ref{sec:model}. In the first stage of the game, every firm $i$
 	simultaneously chooses an investment strategy---potentially mixed,
 	so firm $i$ chooses a distribution over $Q$---and pays the production
 	cost of the realized quality $q_{i}$.\footnote{The definition of the game and the technical details are in the Appendix \ref{app:sec:game:definitions}.}
 	Then, the profile of realized investments becomes public information
 	among firms. In the second stage, firms compete à la Bertrand: every
 	firm $i$ simultaneously chooses a pricing function $p_{i}$ over
 	the qualities that are feasible; i.e., in $[0,q_{i}]$. Finally, every
 	buyer treats the firms as supplying homogeneous goods, and observes
 	the pricing functions; effectively the buyers face the tariff given
 	by the lower envelope of the pricing functions, $q\mapsto\min_{i}p_{i}(q)$.
 	Every buyer maximizes his payoff purchasing a feasible quality $q$
 	from a firm $i$ or buying nothing for the outside option of 0 payoff.
 	The revenues of firm $i$ derive from the demand for the qualities
 	it offers at the lowest price, and costs are sunk at the second stage.
 	We study the subgame-perfect equilibria of this two-stage game of
 	complete information, in which firms choose a distribution of their
 	investment and a pricing function conditional on the realized investments.
 	
 	\paragraph{Results}
 	
 	We can study competition as a constrained monopolist problem, due
 	to the structure of the game. Consider the subgame starting at a cap
 	profile with first- and second-highest qualities given by $x$ and
 	$y$, respectively. Bertrand competition drives to zero the revenues
 	from qualities in $[0,y]$. Hence, only one firm earns positive revenues:
 	the firm $i^{*}$ that produces $x$ and sells qualities in the spectrum
 	$[y,x]$ over which it holds market power. We establish that the ensuing
 	allocation $\bm{q}^{C}[x,y]$ admits an intuitive characterization:
 	it is obtained by ``slicing'' the virtual-surplus maximizer both
 	from above at $x$ and from below at $y$, yielding $\bm{q}^{C}[x,y]\colon\theta\mapsto\max\{\min\{\bm{\beta}(\theta),x\},y\}$ (Figure \ref{fig:competition}).
 	The quality $y$ is provided for free to all types below $b(y)$.\footnote{If a buyer is indifferent across firms, then he purchases from the
 		lowest-index firm. In equilibrium, the firm serving type $\theta\le b(y)$
 		is the firm with lowest index $i$ such that $q_{i}\ge y$.} Qualities above $y$, instead, are available exclusively from the
 	``interim'' monopolist, the only firm earning positive revenues,
 	given by $V(x)-V(y)$.
 	
 	As an implication, the investment game reduces to an all-pay contest
 	in which the second-highest bid takes away part of the gains from
 	winning. From the point of view of firm $i$, the investment of competitors
 	acts as a fixed cost in its production problem ``at'' the highest
 	quality of the opponents---equal to $V(\max_{j\ne i}q_{j})$. Only
 	by ``winning'' does firm $i$ earn positive revenues. This fixed
 	cost affects only the ``entry'' decision, and not investment conditional
 	on entry. Therefore, the best response is either $q^{M}$ or $0$.
 	
 	The following result characterizes the equilibrium allocations. A
 	firm is \emph{active} in an equilibrium if it produces a positive
 	quality with positive probability.
 	\begin{prop}
 		\label{prop:competition}For every $n\in\{1,2,\dots\}$, there exists
 		an equilibrium with $n$ active firms. In every equilibrium in pure
 		strategies, only one firm is active. In every equilibrium with $n\ge2$
 		active firms, the investment strategy of every active firm is the
 		continuous distribution function with support $[0,q^{M}]$ given by
 		$H_{n}(q)=(c'(q)/V'(q))^{1/(n-1)}$, for all $q\in[0,q^{M}]$.
 	\end{prop}
 	The monopolist allocation is an equilibrium allocation, and, in particular, is the unique allocation induced by  pure-strategy equilibria. This result arises due to the implicit
 	commitment of the timing and information structure of the game. Specifically,
 	an active firm can be thought of as committing to fight ``entry,''
 	because its pricing function conditions on the investment of competitors.
 	This type of commitment is absent in models of competitive screening
 	with simultaneous contract posting described in what follows.
 	
 	We refer to equilibria with $n\ge2$ active firms as ``competitive''.
 	A competitive equilibrium induces the random allocation given by $\bm{q}^{C}[X,Y]$,
 	in which $X$ and $Y$ are the first- and second-order statistics
 	of $n$ i.i.d.~draws from $H_{n}$, respectively. The resulting allocation
 	is a contraction of the monopolist allocation.
 	This observation leads to a qualitative effect of competition on the
 	inefficiencies of the monopolist allocation. First, competition exacerbates
 	the productive inefficiency; i.e., $X<q^{M}$ with probability 1.
 	Second, competition can increase the allocation of low types reducing
 	the damaging inefficiency.
 	\begin{cor}\label{comp:cor}
 		In every equilibrium with $n\ge2$ active firms, with probability
 		one: the highest distributed quality is strictly lower than $q^{M}$.
 		Moreover, if the monopolist does not engage in full bunching, then
 		the lowest distributed quality is strictly greater than $\bm{\beta}(0)$
 		with positive probability.
 	\end{cor}
 	
 	 Competition exacerbates productive inefficiency, but improves distributional efficiency. Hence, Corollary \ref{comp:cor} does not establish that competition ``qualitatively'' dominates monopoly, or vice versa, in terms of welfare. However, quantitatively, the comparison remains undetermined: we shut down either effect of competition on the
 	inefficiency via appropriate cost functions (Appendix \ref{app:supp:sec:competition}). As a consequence, there are primitives under which competition dominates monopoly, and vice versa. Specifically, welfare amounts
 	to consumer surplus, because all firms make $0$ profits in a competitive equilibrium. Moreover, the relevant welfare comparison is
 	between duopoly and monopoly as welfare decreases in the number of
 	active firms, starting from $n\ge2$ active firms. If full bunching
 	occurs in $\bm{q}^{M}$---e.g., for steep $c$ by Proposition \ref{prop:mon:comparative}---then
 	the productive inefficiency is the dominant difference between the
 	monopoly and the duopoly allocation; in this case, welfare is higher
 	under monopoly than under duopoly. If, instead, $c$ approximates a
 	fixed-cost function, then the distributive effect dominates, and duopoly
 	dominates monopoly in terms of welfare.
 	\begin{rem}
 		\label{rem:free}A positive quality is supplied at zero
 		price in every competitive equilibrium, because the second-order statistic
 		of the equilibrium caps is positive. However, a positive mass of types
 		consumes this free quality if and only if $y>\bm{\beta}(0)$, which
 		occurs with positive probability. If $y>\bm{\beta}(0)$, then all types
 		in $[0,b(y)]$ consume quality $y$, receiving a strictly higher
 		quality than under the monopolist allocation. The lowest monopolist
 		quality $\bm{\beta}(0)$ can exceed the second-highest quality $y$; in this case, the availability of $y$ simply lowers the transfer of all types, and the competitive allocation is
 		the same as the monopolist allocation for low types (i.e., types in
 		$[0,b(x)]$ receive the same quality under $\bm{q}^{M}$ and under
 		$\bm{q}^{C}[x,y]$.)
 	\end{rem}
 	\paragraph*{Alternative models of competition}
 	
 	In our model, firms compete for buyers after the investment stage,
 	which determines feasible contracts, is complete and publicly observed.
 	Given this timing assumption, an ``incumbent'' firms effectively
 	commits to fight an ``entrant'' in stage 2. Models of competitive
 	screening can lead to non-existence of equilibria if firms can offer
 	any contract taking as given the contracts offered by the opponents,
 	differently from our model \citep*{rothschild_equilibrium_1976}.
 	
 	The timing of our model is also adopted in the quality-screening duopoly
 	of \citet*{champsaur_multiproduct_1989}, in which firms first commit
 	to a quality spectrum at no costs, and then price feasible qualities.
 	Firms make positive profits in equilibrium by playing disjoint quality
 	intervals, which are ruled out in our model in which feasible quality
 	sets take the form $[0,q]$. The space of feasible quality intervals
 	is a key difference between the setups: in \citet*{champsaur_multiproduct_1989},
 	a firm choosing $[a,b]$ in stage 1 effectively commits not to sell
 	qualities worse than $a$ in stage 2, whereas this commitment is absent
 	in our model.
 	
 	There are alternative ways to model multiple firms and privately informed
 	buyers, see \citet*{stole_price_2007} for a survey. In \citet*{garrett_competitive_2019},
 	asymmetric information is two-sided: consumers are imperfectly informed
 	about the posted contracts. Their model generates menu dispersion
 	and competition can raise prices of low-quality goods. In \citet*{johnson_multiproduct_2003},
 	an entrant and an incumbent move simultaneously, but the entrant can
 	only produce qualities lower than a given threshold. This exogenous
 	constraint mimics the equilibrium in the pricing stage of our game.
 	The notion of ``relevance'' introduced by \citet*{chade_screening_2021}
 	does not induce equilibrium existence with our nonseparable costs.
 	In their ``vertical'' oligopoly, with a contract-posting game à
 	la \citeauthor*{rothschild_equilibrium_1976}, a form of market power
 	induces equilibrium existence: each firm has a cost advantage over
 	a quality interval. However, if producing $q$ makes all qualities
 	lower than $q$ available for free, then a cost advantage for
 	low qualities does not play any role; in fact, costs are separable
 	in the model of \citeauthor*{chade_screening_2021}.
 	
 	\section{Conclusion}
 	
 	This paper extends the workhorse screening model by studying an alternative
 	cost structure in which costs depend on the maximum quality. The optimum
 	for the seller involves both the usual virtual-surplus analysis, and
 	the association of each quality to a marginally bunched type, in order
 	to determine the marginal revenue of a quality investment. A natural
 	extension is to consider costs such that providing a quality to a
 	type depends both on her own quality and on the entire
 	allocation. This analysis can be used to model,
 	e.g., economies of scale, and damaging costs.
 	
 	Innovation in digital markets proceeds dynamically, and firms adjust
 	research and production decisions over time, for example as decisions
 	of competitors are observed, or new information about demand is gathered.
 	An extension that captures a dynamic model of competition can inform
 	policies related to patent protection, innovation, and the regulation
 	of natural monopolies in digital markets.
 	
 	\newpage 
 	\appendix
 	
 	\section{Proofs}
 	
 	\subsection{General model}
 	
 	The following model nests the model in the main body of the paper.
 	The set of qualities is $Q\coloneqq[0,\overline{q}]$, an allocation
 	is a measurable $\bm{q}\colon\Theta\to Q$, the set of allocations
 	is $\bm{Q}$, $F$ is a distribution function whose support is $\Theta\coloneqq[0,1]$,
 	and is $\mathcal{C}^{2}$. The utility $u\colon Q\times\Theta\to\mathbb{R}$
 	and the virtual surplus $J\colon(q,\theta)\mapsto u(q,\theta)-\frac{1-F(\theta)}{F'(\theta)}u_{2}(q,\theta)-k(q)$
 	satisfy increasing differences, and are: $\mathcal{C}^{2}$, concave
 	in $q$ for all $\theta\in\Theta$, strictly quasiconcave in $q$
 	for all $\theta\in\Theta\setminus S$, for a countable $S\subseteq\Theta$,
 	and a convex, $\mathcal{C}^{2}$, and nondecreasing $k\colon Q\to\mathbb{R}$.
 	The function $c\colon Q\to\mathbb{R}$ is increasing, strictly convex,
 	differentiable, and satisfies: 
 	\begin{align*}
 		\lim_{q\to0}{\Big (\int_{\Theta}J_{1}(q,{ \theta)\diff F(\theta)-c'(q)\Big )>0>\lim_{q\to\overline{q}}{ \Big (\int_{\Theta}J_{1}(q,\theta)\diff F(\theta)-c'(q)\Big)}}}.
 	\end{align*}
 	We maintain the convention that $\inf\emptyset=1$, and use Iverson brackets: $[P]=1$ if the statement $P$ is true, and $[P]=0$ otherwise.
 	
 	The present model nests the model in the main body of the paper except
 	for the finite quality upper bound $\overline{q}$, because $Q$ is
 	unbounded in the paper. The only role of $\overline{q}$ is to compactify
 	the quality space, as shown in the rest of the Appendix (Proposition
 	\ref{app:prop:eff:alloc}, \ref{app:prop:mon:alloc}). Hence, the
 	results in the main text are implied by the results that are proved
 	in what follows.
 	
 	\paragraph*{Preliminaries}
 	
 		The welfare of allocation $\bm q$ is $W(\bm q) = \int_{\Theta}u(\theta,\bm{q}(\theta))-k(\bm{q}(\theta))\diff F(\theta)-c(\sup\bm{q})$. The allocation $\bm{q}$ is \emph{efficient} if $\bm{q}$ solves $\max_{\bm{q}\in\bm{Q}} W(\bm q)$.	Given $q\in Q$, we let $U(q)$ be the value of the problem ${\mathcal{P}}^{\star}(q)$,
 	\begin{align*}
 		U(q)=\sup_{\bm{q}\in\bm{Q}}\int_{\Theta}u(\bm{q}(\theta),\theta)-k(\bm{q}(\theta))\diff F(\theta)\ \text{subject to:}\ \bm{q}(\theta)\le q\ \text{for all}\ \theta\in\Theta.
 	\end{align*}
 	The \emph{surplus maximizer} is the largest selection
 	from $\theta\mapsto\argmax_{q\in Q}u(\theta,q)-k(q)$; we denote the surplus maximizer by $\bm \alpha$. The generalized
 	inverse $a\colon q\mapsto\inf\{\theta\in\Theta\mid \bm{\alpha}(\theta)\ge q\}$
 	of $\bm{\alpha}$ is nondecreasing because $u$ has increasing differences
 	\citep*{topkis_minimizing_1978}, and satisfies: $a(q)=0$ if $q<\bm{\alpha}(0)$,
 	and $a(q)=1$ if $q>\bm{\alpha}(1)$. Moreover, $a$ is continuously
 	differentiable at $q\in(\bm{\alpha}(0),\bm{\alpha}(1))$, with $a'(q)=-\frac{u_{11}(q,a(q))-k''(q)}{u_{12}(q,a(q))}$
 	by the implicit function theorem.
 	
 	The set of direct mechanisms that are incentive-compatible and individually
 	rational is $\bm{M}=\{(\bm{q},t)\in\bm{Q}\times\mathbb{R}^{\Theta}\mid u(\theta,\bm{q}(\theta))-t(\theta)\ge u(\theta,\bm{q}(\hat{\theta}))-t(\hat{\theta}),\,u(\theta,\bm{q}(\theta))-t(\theta)\ge0\:\text{for all}\:(\theta,\hat{\theta})\in\Theta^{2}\}$.
 	The allocation $\bm{q}$ is \emph{monopolist} if there exists $t\colon\Theta\to\mathbb{R}$
 	such that $(\bm{q},t)$ solves $\sup_{(\bm{q},\,t)\in\bm{M}}\int_{\Theta}t(\theta)\diff F(\theta)-c(\sup\bm{q})$.
 	For $q\in Q$,  we let $V(q)$ be the value of the problem $\mathcal{P}(q)$, 
 	\begin{align*}
 		V(q)=\sup_{(\bm{q},\,t)\in\bm{M}}\int_{\Theta}t(\theta)\diff F(\theta)\ \text{subject to:}\ \bm{q}(\theta)\le q\ \text{for all}\ \theta\in\Theta.
 	\end{align*}
 	The \emph{virtual surplus maximizer} is the largest selection from $\theta\mapsto\argmax_{q\in Q}J(q,\theta)$; we denote the virtual surplus maximizer by $\bm \beta$. The generalized
 	inverse $b\colon q\mapsto\inf\{\theta\in\Theta\mid \bm{\beta}(\theta)\ge q\}$
 	of $\bm{\beta}$ is nondecreasing because $J$ satisfies increasing
 	differences, and satisfies: $b(q)=0$ if $q<\bm{\beta}(0)$, and $b(q)=1$
 	if $q>\bm{\beta}(1)$. Moreover, $b$ is continuously differentiable
 	at $q\in(\bm{\beta}(0),\bm{\beta}(1))$, with $b'(q)=-\frac{J_{11}(q,b(q))}{J_{12}(q,b(q))}$.
 	
 	\subsection{Proofs for Section \ref{sec:model}}
 	
 	\label{app:sec:proofs}
 	
 	Proposition \ref{prop:efficient}, Lemma \ref{lem:decomposition},
 	and Proposition \ref{prop:monopolist} are implied by Proposition
 	\ref{app:prop:eff:alloc}, Lemma \ref{app:lem:mon:dec}, and Proposition
 	\ref{app:prop:mon:alloc}, respectively; a technical argument in the
 	proof of Proposition \ref{app:prop:eff:alloc} and \ref{app:prop:mon:alloc}
 	is relegated to Appendix \ref{app:sec:extra}. %%This is different in the Appendix for Insights.\
 	
 	\begin{lemma}\label{app:lem:eff:dec} The allocation $\bm{q}$ is
 		efficient if and only if: $\bm{q}$ solves $\mathcal{P}^{\star}(q^{\star})$
 		for a quality $q^{\star}\in\argmax_{q\in Q}U(q)-c(q)$. \end{lemma} 
 	\begin{proof}
 		We first show that $\bm{q}$ is efficient only if $\bm{q}$ solves
 		$\mathcal{P}^{\star}(q^{\star})$ for a quality $q^{\star}\in\argmax_{q\in Q}U(q)-c(q)$.
 		Let $\bm{q}$ be efficient. We want to show that: there exists $q^{\star}\in Q$
 		such that: (1) $\sup\bm{q}\le q^{\star}$, (2) for all $\tilde{\bm{q}}\in\bm{Q}$
 		with $\sup\tilde{\bm{q}}\le q^{\star}$, $\int_{\Theta}u(\bm{q}(\theta),\theta)-k(\bm{q}(\theta))\diff F(\theta)\ge\int_{\Theta}u(\tilde{\bm{q}}(\theta),\theta)-k(\tilde{\bm{q}}(\theta))\diff F(\theta)$,
 		and (3) for all $q\in Q$, $U(q^{\star})-c(q^\star)\ge U(q)-c(q)$. We claim that
 		$\sup\bm{q}$ satisfies the three properties. First, consider property
 		(2). As an implication of efficiency and the fact that $c$ is increasing,
 		we have that: for every $\tilde{\bm{q}}\in\bm{Q}$, if $c(\tilde{\bm{q}})<c(\bm{q})$,
 		then $\int_{\Theta}u(\bm{q}(\theta),\theta)-k(\bm{q}(\theta))\diff F(\theta)>\int_{\Theta}u(\tilde{\bm{q}}(\theta),\theta)-k(\tilde{\bm{q}}(\theta))\diff F(\theta)$.
 		Hence, property (2) holds. Moreover, $U(\sup\bm{q})=\int_{\Theta}u(\bm{q}(\theta),\theta)-k(\bm{q}(\theta))$,
 		so (3) and (1) hold.
 		
 		We proceed to show that $\bm{q}$ is efficient if $\bm{q}$ solves
 		$\mathcal{P}^{\star}(q^{\star})$ for a quality $q^{\star}\in\argmax_{q\in Q}U(q)-c(q)$.
 		Towards establishing an intermediate claim, let $X\coloneqq\bm{Q}$,
 		$Y\coloneqq Q$, for all $(x,y)\in X\times Y$, $f(x,y)\coloneqq\int_{\Theta}u(x(\theta),\theta)-k(x(\theta))\diff F(\theta)-c(y)$,
 		and $h(x)\coloneqq\sup x(\Theta)$. Denote by $v_{1}$ the value of
 		the efficient allocation, $v_{1}\coloneqq\sup\{f(x,h(x))\mid x\in X\}$,
 		and $v_{2}$ the value of the decomposed problem, $v_{2}\coloneqq\sup\{\sup\{f(x,y)\mid x\in X,\,h(x)=y\}\mid y\in Y\}$.
 		By construction, we have that $v_{1}=\sup\{f(x,y)\mid x\in X,\,y\in Y,\,h(x)=y\}$.
 		We claim that $v_{1}=v_{2}$. First, we show that $v_{2}\ge v_{1}$. For
 		all $y\in Y,\,x\in X$, if $h(x)=y$, then 
 		\begin{align*}
 			f(x,y)\le\sup\{f(\tilde{x},y)\mid\tilde{x}\in X,\,h(\tilde{x})=y\}.
 		\end{align*}
 		Thus, for all $y\in Y,\,x\in X$, if $h(x)=y$, then $f(x,y)\le v_{2}$.
 		Therefore, we have that $v_{2}\ge v_{1}$. Second, we show that $v_{1}\ge v_{2}$.
 		For all $y\in Y,\,x\in X$, if $h(x)=y$, then 
 		\begin{align*}
 			f(x,y)\le v_{1}.
 		\end{align*}
 		Therefore, it holds that $v_{1}\ge v_{2}$. We conclude that $v_{1}=v_{2}$.
 		
 		Let $\bm{q}$ solve $\mathcal{P}^{\star}(q^{\star})$ for a quality
 		$q^{\star}\in\argmax_{q\in Q}U(q)-c(q)$. By construction, we have
 		that $f(\bm{q},\sup\bm{q})\ge f(\bm{q},q^{\star})$, $v_{2}=f(\bm{q},q^{\star})$,
 		and $v_{1}\ge f(\bm{q},\sup\bm{q})$. We conclude that: 
 		\begin{align*}
 			v_{2}\ge f(\bm{q},\sup\bm{q})\ge v_{1}.
 		\end{align*}
 		Therefore, $\bm{q}$ is efficient. 
 	\end{proof}
 
 	\begin{lemma}
 		There exists a unique quality $q^\star$ that maximizes $q\mapsto U(q)-c(q)$. Moreover, $q^{\star}$ is
 		the unique quality $q$ such that $\int_{[a(q),1]}u_{1}(q,\theta)-k'(q)\diff F(\theta)=c'(q)$.
 	\end{lemma}
 	\begin{proof}
 	\emph{Claim: $U$ is concave, continuously differentiable, and $U'(q)=\int_{[a(q),1]}u_{1}(q,\theta)-k'(q)\diff F(\theta)$
 			for all $q\in(0,\overline{q})$.} First, $U$ is differentiable at
 		$q<\bm{\alpha}(0)$ with $U'(q)=\int_{[0,1]}u_{1}(q,\theta)-k'(q)\diff F(\theta)$
 		if $q>0$; $U$ is differentiable at $q\in(\bm{\alpha}(0),\bm{\alpha}(1))$
 		with $U'(q)=\int_{[a(q),1]}u_{1}(q,\theta)-k'(q)\diff F(\theta)$
 		because $a$ is continuously differentiable; finally, $U$ is differentiable
 		at $q>\bm{\alpha}(1)$ with $U'(q)=0$ if $q<\overline{q}$. The previous
 		derivatives are continuously pasted at $\bm{\alpha}(0)$ and $\bm{\alpha}(1)$,
 		so $U$ is continuously differentiable. For the claim, it suffices
 		to establish that $U'(q_{2})-U'(q_{1})\le0$ for all $q_{1},q_{2}\in Q$
 		with $q_{2}>q_{1}$. It holds that 
 		\begin{align*}
 			U'(q_{2})-U'(q_{1})=\int_{[a(q_{2}),1]}u_{1}(q_{2},\theta)-u_{1}(q_{1},\theta)-k'(q_{2})+k'(q_{1})\diff F(\theta)\\
 			-\int_{[a(q_{1}),a(q_{2}))}u_{1}(q_{1},\theta)-k'(q_{1})\diff F(\theta).
 		\end{align*}
 		Moreover, $u_{1}(q_{2},\theta)-u_{1}(q_{1},\theta)-k'(q_{2})+k'(q_{1})\le0$
 		by concavity of $q\mapsto u(q,\theta)-k(q)$ for all $\theta\in\Theta$,
 		and $u_{1}(q_{1},\theta)-k'(q_{1})\ge0$ for all $\theta\ge a(q_{1})$
 		by definition of $a$ and $\bm{\alpha}$. Hence, $U$ is concave.
 		
 		\emph{Claim: $U$ is right continuous at $0$.} For all $q>0$, we
 		have: $U(q)\ge\int_{\Theta}u(0,\theta)-k(0)\diff F(\theta)$; moreover,
 		if $q$ is sufficiently small, we have $U(q)\le\int_{\Theta}u(q,\theta)-k(0)\diff F(\theta)$
 		by our assumption on $c'$. Hence, we have $U(q)\to\int_{\Theta}u(0,\theta)-k(0)\diff F(\theta)$
 		as $q\to0$ by continuity of $u$.
 		
 		By the properties of $c'$, if $q\in\argmax_{q\in Q}U(q)-c(q)$ then
 		$q\in(0,\overline{q})$, so the proof is complete. 
 	\end{proof}
 
 	\begin{proposition}\label{app:prop:eff:alloc}
 		The allocation $\bm{q}$ is efficient if and only if: there exists
 		an allocation $\bm{\gamma}$ such that $\bm{\gamma}(\theta)=\bm{\alpha}(\theta)$
 		almost everywhere and $\bm{q}(\theta)=\min\{\bm{\gamma}(\theta),q^{\star}\}$
 		for all $\theta$.
 	\end{proposition}
 	\begin{proof}
 		The proof has two steps. First, we solve $\mathcal{P}^{\star}(q)$
 		for all $q\in Q$. Then, the result follows from Lemma \ref{app:lem:eff:dec}.
 		
 		It suffices to show that: for fixed $q\in Q$, $\bm{q}$ solves $\mathcal{P}^{\star}(q)$ iff $\bm{q}(\theta)=\min\{\bm{\gamma}(\theta),q\}$ for all $\theta$
 			and some $\bm{\gamma}\in\bm{Q}$ such that $\bm{\gamma}(\theta)=\bm{\alpha}(\theta)$
 			almost everywhere.
 			
 		By Lemma \ref{app:lem:general}, $\bm{q}$ solves
 		$\mathcal{P}^{\star}(q)$ iff $\bm{q}(\theta)=\min\{\bm{\gamma}(\theta),q\}$
 		for all $\theta$ and some allocation $\bm{\gamma}$ such that $\bm{\gamma}(\theta)\in\argmax_{\hat{q}\in Q}u(\hat{q},\theta)-k(\hat{q})$
 		almost everywhere. We have $\bm{\gamma}(\theta)=\max\argmax_{\hat{q}\in Q}u(\hat{q},\theta)-k(\hat{q})$
 		a.e.~by strict quasiconcavity of $u(\cdot,\theta)-k(\cdot)$, so
 		the claim holds.
 	\end{proof}
 
 	\begin{lemma}\label{app:lem:mon:dec} The allocation $\bm{q}$ is
 		monopolist if and only if: $\bm{q}$ solves $\mathcal{P}(q^{M})$
 		for a quality $q^{M}\in\argmax_{q\in Q}V(q)-c(q)$. \end{lemma} 
 	\begin{proof}
 		This proof follows the same steps as in the proof of Lemma \ref{app:lem:eff:dec}.
 		Two preliminary observations follow from standard arguments \citep*{carroll_contract_2023}.
 		First, we have 
 		\begin{align*}
 			V(q)=&\max_{\bm{q}\in\bm{Q}}\int_{\Theta}J(\bm{q}(\theta),\theta)\diff F(\theta)\ \text{subject to:}\ \bm{q}(\theta)\le q\ \text{for all}\ \theta\in\Theta,\\& \bm{q}\ \text{is nondecreasing};
 		\end{align*}
 		and $(\bm{q},t)$ solves $\mathcal{P}(q)$ for some $t$ iff $\bm{q}$
 		solves the above problem. Second, we have that: $\bm{q}$ is monopolist
 		iff $\bm{q}$ solves 
 		\begin{align*}
 			\max_{\bm{q}\in\bm{Q}}\int_{\Theta}J(\bm{q}(\theta),\theta)\diff F(\theta)-c(\sup \bm q)\ \text{subject to:}\ \bm{q}\ \text{is nondecreasing}.
 		\end{align*}
 		
 		We first show that $\bm{q}$ is monopolist only if $\bm{q}$ solves
 		$\mathcal{P}(q^{M})$ for a quality $q^{M}\in\argmax_{q\in Q}V(q)-c(q)$.
 		Let $\bm{q}$ be monopolist. We want to show that: there exists $q^{M}\in Q$
 		such that: (1) $\sup\bm{q}\le q^{M}$, (2) for all nondecreasing $\tilde{\bm{q}}\in\bm{Q}$
 		with $\sup\tilde{\bm{q}}\le q^{M}$, $\int_{\Theta}J(\bm{q}(\theta),\theta)\diff F(\theta)\ge\int_{\Theta}J(\tilde{\bm{q}}(\theta),\theta)\diff F(\theta)$,
 		and (3) for all $q\in Q$, $V(q^{M})-c(q^M)\ge V(q)-c(q)$. We claim that $\sup\bm{q}$
 		satisfies the three properties. First, consider property (2). For
 		every $\tilde{\bm{q}}\in\bm{Q}$, we have 
 		\begin{align*}
 			\int_{\Theta}J(\bm{q}(\theta),\theta)\diff F(\theta)-c(\sup\bm{q})\ge\int_{\Theta}J(\tilde{\bm{q}}(\theta),\theta)\diff F(\theta)-c(\sup\tilde{\bm{q}}).
 		\end{align*}
 		Because $c$ is increasing and $\bm q$ is monopolist, we have that: for every $\tilde{\bm{q}}\in\bm{Q}$,
 		if $c(\tilde{\bm{q}})<c(\bm{q})$, then $\int_{\Theta}J(\bm{q}(\theta),\theta)\diff F(\theta)>\int_{\Theta}J(\tilde{\bm{q}}(\theta),\theta)\diff F(\theta)$.
 		Hence, property (2) holds. Moreover, $V(\sup\bm{q})=\int_{\Theta}J(\bm{q}(\theta),\theta)$,
 		so (3) and (1) hold.
 		
 		We proceed to show that $\bm{q}$ is monopolist if $\bm{q}$ solves
 		$\mathcal{P}(q^{M})$ for a quality $q^{M}\in\argmax_{q\in Q}V(q)-c(q)$.
 		Towards establishing an intermediate claim, let $X\coloneqq\{\bm{q}\in\bm{Q}\mid\bm{q}\ \text{is nondecreasing}\}$,
 		$Y\coloneqq Q$, for all $(x,y)\in X\times Y$, $f(x,y)\coloneqq\int_{\Theta}J(x(\theta),\theta)\diff F(\theta)-c(y)$,
 		and $h(x)\coloneqq\sup x(\Theta)$. We denote by $v_{1}$ the profit
 		of the monopolist allocation, $v_{1}\coloneqq\sup\{f(x,h(x))\mid x\in X\}$,
 		and $v_{2}$ the value of the decomposed problem, $v_{2}\coloneqq\sup\{\sup\{f(x,y)\mid x\in X,\,h(x)=y\}\mid y\in Y\}$.
 		By the same steps as in the proof of Lemma \ref{app:lem:eff:dec},
 		we conclude that $v_{1}=v_{2}$.
 		
 		Let $\bm{q}$ solve $\mathcal{P}(q^{M})$ for a quality $q^{M}\in\argmax_{q\in Q}V(q)-c(q)$.
 		By construction, we have that $f(\bm{q},\sup\bm{q})\ge f(\bm{q},q^{M})$,
 		$v_{2}=f(\bm{q},q^{M})$, and $v_{1}\ge f(\bm{q},\sup\bm{q})$. We
 		conclude that 
 		\begin{align*}
 			v_{2}\ge f(\bm{q},\sup\bm{q})\ge v_{1}.
 		\end{align*}
 		Therefore, $\bm{q}$ is monopolist. 
 	\end{proof}
 
 	 	\begin{lemma}
 		There exists a unique quality $q^M$ that maximizes $q\mapsto V(q)-c(q)$. Moreover, $q^M$ is the unique quality $q$ such that $\int_{[b(q),1]}J_{1}(q,\theta)\diff F(\theta)=c'(q)$. Moreover, it holds that $0<q^{M}<q^{\star}$.
 		\end{lemma}
 		\begin{proof}
 			\emph{Claim: $V$ is concave, continuously differentiable, and $V'(q)=\int_{[b(q),1]}J_{1}(q,\theta)\diff F(\theta)$
 				for all $q\in(0,\overline{q})$.} First, $V$ is differentiable at
 			$q<\bm{\beta}(0)$ with $V'(q)=\int_{[0,1]}J_{1}(q,\theta)\diff F(\theta)$
 			if $q>0$; $V$ is differentiable at $q\in(\bm{\beta}(0),\bm{\beta}(1))$
 			with $V'(q)=\int_{[b(q),1]}J_{1}(q,\theta)\diff F(\theta)$ because
 			$b$ is continuously differentiable; finally, $V$ is differentiable
 			at $q>\bm{\beta}(1)$ with $V'(q)=0$ if $q<\overline{q}$. The previous
 			derivatives are continuously pasted at $\bm{\beta}(0)$ and $\bm{\beta}(1)$.
 			For the claim, it suffices to establish that $V'(q_{2})-V'(q_{1})\le0$
 			for all $q_{1},q_{2}\in Q$ with $q_{2}>q_{1}$. It holds that 
 			\begin{align*}
 				V'(q_{2})-V'(q_{1})=\int_{[b(q_{2}),1]}J_{1}(q_{2},\theta)-J_{1}(q_{1},\theta)\diff F(\theta)-\int_{[b(q_{1}),b(q_{2})]}J_{1}(q_{1},\theta)\diff F(\theta).
 			\end{align*}
 			Moreover, $J_{1}(q_{2},\theta)-J_{1}(q_{1},\theta)\le0$ by concavity
 			of $q\mapsto J(q,\theta)$ for all $\theta\in\Theta$, and $J_{1}(q_{1},\theta)\ge0$
 			for all $\theta\ge b(q_{1})$ by definition of $b$, $\bm{\beta}$.
 			Hence, $V$ is concave.
 			
 			\emph{Claim: $V$ is right continuous at $0$.} For all $q>0$, we
 			have: $\int_{\Theta}u(q,\theta)-\frac{1-F(\theta)}{F'(\theta)}u_{2}(0,\theta)-k(0)\diff F(\theta)\ge V(q)\ge\int_{\Theta}J(0,\theta)\diff F(\theta)$.
 			Hence, we have $V(q)\to\int_{\Theta}J(0,\theta)\diff F(\theta)$ as
 			$q\to0$.
 			
 			\emph{Claim: $\bm{\beta}(\theta)\le\bm{\alpha}(\theta)$ for all $\theta\in\Theta$,
 				and, additionally, $\bm{\beta}(\theta)<\bm{\alpha}(\theta)$ if: $\bm{\alpha}(\theta)>0$,
 				$\bm{\beta}(\theta)<\overline{q}$, and $\theta<1$.} It suffices
 			to observe that $u_{1}(q,\theta)-k'(q)-J_{1}(q,\theta)=\frac{1-F(\theta)}{F'(\theta)}u_{12}(q,\theta)$
 			for all $(q,\theta)\in Q\times\Theta$.
 			
 			\emph{Claim}: $q^{M}<q^{\star}<\overline{q}$. First, we observe that,
 			if $q\in(0,\overline{q})$, then: $a(q)<b(q)$, and, moreover, 
 			\begin{align*}
 				U'(q)-V'(q)=\int_{[b(q),1]}\frac{1-F(\theta)}{F'(\theta)}u_{12}(q,\theta)\diff F(\theta)+\int_{[a(q),b(q))}u_1(q,\theta)-k'(q)\diff F(\theta).
 			\end{align*}
 			It follows that $U'(q)-V'(q)\ge\int_{[a(q),b(q)]}u_1(q,\theta)-k'(q)\diff F(\theta)>0$
 			by the definition of $a$. Hence, we have that $q^{M}<q^{\star}<\overline{q}$.
 			
 			To complete the proof, it remains to show that $q^{M}>0$. We let
 			the right-continuous inverse of $\bm{\beta}$ be $T^{+}\colon q\mapsto\inf\{\theta\in\Theta \mid \bm{\beta}(\theta)>q\}$.
 			Observe that: $\lim_{q\to0}V'(q)=\int_{[T^{+}(0),1]}J_{1}(0,\theta)\diff F(\theta)$.
 			If $1>T^{+}(0)>0$, by increasing differences we have
 			\begin{align*}
 				\frac{1}{1-F(T^{+}(0))}\lim_{q\to0}V'(q)>\int_{[0,1]}J_{1}(0,\theta)\diff F(\theta)
 			\end{align*}
 			We note that $1>T^{+}(0)$, because $J_{1}(0,1)>0$ and $\bm{\beta}$
 			is continuous. If, instead, $T^{+}(0)=0$, then $\lim_{q\to0}V'(q)=\int_{[0,1]}J_{1}(0,\theta)\diff F(\theta)$.
 			By the properties of $c'$, it follows that $q^{M}>0$.\emph{Claim: $V$ is concave, continuously differentiable, and $V'(q)=\int_{[b(q),1]}J_{1}(q,\theta)\diff F(\theta)$
 				for all $q\in(0,\overline{q})$.} First, $V$ is differentiable at
 			$q<\bm{\beta}(0)$ with $V'(q)=\int_{[0,1]}J_{1}(q,\theta)\diff F(\theta)$
 			if $q>0$; $U$ is differentiable at $q\in(\bm{\beta}(0),\bm{\beta}(1))$
 			with $V'(q)=\int_{[b(q),1]}J_{1}(q,\theta)\diff F(\theta)$ because
 			$b$ is continuously differentiable; finally, $V$ is differentiable
 			at $q>\bm{\beta}(1)$ with $V'(q)=0$ if $q<\overline{q}$. The previous
 			derivatives are continuously pasted at $\bm{\beta}(0)$ and $\bm{\beta}(1)$.
 			For the claim, it suffices to establish that $V'(q_{2})-V'(q_{1})\le0$
 			for all $q_{1},q_{2}\in Q$ with $q_{2}>q_{1}$. It holds that 
 			\begin{align*}
 				V'(q_{2})-V'(q_{1})=\int_{[b(q_{2}),1]}J_{1}(q_{2},\theta)-J_{1}(q_{1},\theta)\diff F(\theta)-\int_{[b(q_{1}),b(q_{2})]}J_{1}(q_{1},\theta)\diff F(\theta).
 			\end{align*}
 			Moreover, $J_{1}(q_{2},\theta)-J_{1}(q_{1},\theta)\le0$ by concavity
 			of $q\mapsto J(q,\theta)$ for all $\theta\in\Theta$, and $J_{1}(q_{1},\theta)\ge0$
 			for all $\theta\ge b(q_{1})$ by definition of $b$, $\bm{\beta}$.
 			Hence, $V$ is concave.
 			
 			\emph{Claim: $V$ is right continuous at $0$.} For all $q>0$, we
 			have: $\int_{\Theta}u(q,\theta)-\frac{1-F(\theta)}{F'(\theta)}u_{2}(0,\theta)-k(0)\diff F(\theta)\ge V(q)\ge\int_{\Theta}J(0,\theta)\diff F(\theta)$.
 			Hence, we have $V(q)\to\int_{\Theta}J(0,\theta)\diff F(\theta)$ as
 			$q\to0$.
 			
 			\emph{Claim: $\bm{\beta}(\theta)\le\bm{\alpha}(\theta)$ for all $\theta\in\Theta$,
 				and, additionally, $\bm{\beta}(\theta)<\bm{\alpha}(\theta)$ if: $\bm{\alpha}(\theta)>0$,
 				$\bm{\beta}(\theta)<\overline{q}$, and $\theta<1$.} It suffices
 			to observe that $u_{1}(q,\theta)-k'(q)-J_{1}(q,\theta)=\frac{1-F(\theta)}{F'(\theta)}u_{12}(q,\theta)$
 			for all $(q,\theta)\in Q\times\Theta$.
 			
 			\emph{Claim}: $q^{M}<q^{\star}<\overline{q}$. First, we observe that,
 			if $q\in(0,\overline{q})$, then: $a(q)<b(q)$, and, moreover, 
 			\begin{align*}
 				U'(q)-V'(q)=\int_{[b(q),1]}\frac{1-F(\theta)}{F'(\theta)}u_{12}(q,\theta)\diff F(\theta)+\int_{[a(q),b(q))}u_1(q,\theta)-k'(q)\diff F(\theta).
 			\end{align*}
 			It follows that $U'(q)-V'(q)\ge\int_{[a(q),b(q)]}u_1(q,\theta)-k'(q)\diff F(\theta)>0$
 			by the definition of $a$. Hence, we have that $q^{M}<q^{\star}<\overline{q}$.
 			
 			To complete the proof, it remains to show that $q^{M}>0$. We let
 			the right-continuous inverse of $\bm{\beta}$ be $T^{+}\colon q\mapsto\inf\{\theta\in\Theta\mid\bm{\beta}(\theta)>q\}$.
 			Observe that: $\lim_{q\to0}V'(q)=\int_{[T^{+}(0),1]}J_{1}(0,\theta)\diff F(\theta)$.
 			If $1>T^{+}(0)>0$, by increasing differences we have
 			\begin{align*}
 				\frac{1}{1-F(T^{+}(0))}\lim_{q\to0}V'(q)>\int_{[0,1]}J_{1}(0,\theta)\diff F(\theta)
 			\end{align*}
 			We note that $1>T^{+}(0)$, because $J_{1}(0,1)>0$ and $\bm{\beta}$
 			is continuous. If, instead, $T^{+}(0)=0$, then $\lim_{q\to0}V'(q)=\int_{[0,1]}J_{1}(0,\theta)\diff F(\theta)$.
 			By the properties of $c'$, it follows that $q^{M}>0$.
 		\end{proof}
 	
 	\begin{proposition}\label{app:prop:mon:alloc}
 		The allocation $\bm{q}$ is monopolist if and only if: there exists
 		a nondecreasing allocation $\bm{\gamma}$ such that $\bm{\gamma}(\theta)=\bm{\beta}(\theta)$
 		almost everywhere and $\bm{q}(\theta)=\min\{\bm{\gamma}(\theta),q^{M}\}$
 		for all $\theta$. 
 	 \end{proposition} 
 	\begin{proof}
 		First, we solve $\mathcal{P}(q)$
 		for all $q\in Q$. Then, the result follows from Lemma
 		\ref{app:lem:mon:dec}.
 		
 		\emph{Claim: for fixed $q\in Q$, $\bm{q}$ solves $\mathcal{P}(q)$
 			iff $\bm{q}(\theta)=\min\{\bm{\gamma}(\theta),q\}$ for all $\theta$
 			and some nondecreasing $\bm{\gamma}\in\bm{Q}$ with $\bm{\gamma}(\theta)=\bm{\beta}(\theta)$
 			almost everywhere.} For sufficiency, let $\bm{q}\in\bm{Q}$ be such
 		that $\bm{q}(\theta)=\min\{\bm{\gamma}(\theta),q\}$ for all $\theta\in\Theta$
 		and for some nondecreasing allocation $\bm{\gamma}$ with $\bm{\gamma}(\theta)=\bm{\beta}(\theta)$
 		almost everywhere. Then, by Lemma \ref{app:lem:general}, $\bm{q}$ solves the problem without the
 		monotonicity constraint; i.e., $\bm q$ solves
 		  \begin{align} \label{app:eq:monunconstrained}
 		 	&\max_{\bm{q}\in\bm{Q}}\int_{\Theta}J(\bm{q}(\theta),\theta)\diff F(\theta)\ \text{subject to:}\ \bm{q}(\theta)\le q\ \text{for all}\ \theta\in\Theta,
 		 \end{align}
 		(see the preliminary observation in the proof of Lemma \ref{app:lem:mon:dec}.) Additionally, $\bm q$ is nondecreasing, so $\bm{q}$ solves $\mathcal{P}(q)$.

 		For necessity, let $\bm{q}$ solve $\mathcal{P}(q)$ and be such
 		that: there does not exist a nondecreasing allocation $\bm{\gamma}$
 		with $\bm{q}(\theta)=\min\{\bm{\gamma}(\theta),q\}$ for all $\theta$
 		and $\bm{\gamma}(\theta)=\bm{\beta}(\theta)$ almost everywhere. By
 		Lemma \ref{app:lem:general} and the strict-quasiconcavity property
 		of $J(\cdot,\theta)$, $\bm{q}$ does not solve the problem without
 		the monotonicity constraint; i.e., $\bm q$ does not solve the maximization in (\ref{app:eq:monunconstrained}). We observe that the function $\theta\mapsto\min\{\bm{\beta}(\theta),q\}$
 		is an allocation that solves the problem without the monotonicity
 		constraint---in (\ref{app:eq:monunconstrained})---by Lemma \ref{app:lem:general}, and satisfies the monotonicity
 		constraint, so $\bm{q}$ does not solve $\mathcal{P}(q)$. The claim
 		follows from the preliminary observation in the proof of Lemma \ref{app:lem:mon:dec}.
 	\end{proof}
 	
 	\subsection{Proofs for Section \ref{sec:interpretation}}
 	\begin{proof}[Proof of Proposition \ref{prop:mon:comparative}]
 		For this proof, we assume that $u$ is as in the main text. Define,
 		for all $q\in Q$, $G(q,\kappa_{c},\kappa_{g}):=(1-F(b_{\kappa_{g}}(q)))(b_{\kappa_{g}}(q)+\kappa_{g}g'(q))-\kappa_{c}c'(q)$.
 		By Proposition \ref{app:prop:mon:alloc}, $q_{\kappa}^{M}$ is interior,
 		and is the unique quality $q$ such that $G(q,\kappa_{c},\kappa_{g})=0$.
 		
 		\emph{Claim 1: The cap $q_{\kappa}^{M}$ is decreasing in $\kappa_{c}$.} By strict concavity of $q\mapsto V(q)-\kappa_{c}c(q)$, we have that $G_{2}(q,\kappa_{c},\kappa_{g})=-c'(q)$ and $G_{1}(q,\kappa_{c},\kappa_{g})<0$.
 	 By the implicit
 		function theorem: $\frac{\partial}{\partial\kappa_{c}}q_{\kappa}^{M}=-G_{2}(q_{\kappa}^{M},\kappa_{c},\kappa_{g})/G_{1}(q_{\kappa}^{M},\kappa_{c},\kappa_{g})$.
 		Hence, $q_{\kappa}^{M}$ is decreasing in $\kappa_{c}$.
 		
 		\emph{Claim 2: The cap $q_{\kappa}^{M}$ is nondecreasing in $\kappa_{g}$.}
 		Fix $q$ with $b_{\kappa_{g}}(q)\in(0,1)$. It holds that 
 		\begin{align*}
 			G_{3}(q,\kappa_{c},\kappa_{g}) & =(1-F(b_{\kappa_{g}}(q)))(g'(q)+\frac{\partial}{\partial\kappa_{g}}b_{\kappa_{g}}(q))\\
 			&-F'(b_{\kappa_{g}}(q))\Big (\frac{\partial}{\partial\kappa_{g}}b_{\kappa_{g}}(q)\Big ) (b_{\kappa_{g}} (q)+\kappa_{g}g'(q))\\
 			& =(1-F(b_{\kappa_{g}}(q)))g'(q),
 		\end{align*}
 		using $\kappa_{g}g'(q)+\varphi(b_{\kappa_{g}}(q))=0$ for the second
 		equality, and $G_{1}(q,\kappa_{c},\kappa_{g})<0$ by strict concavity
 		of $q\mapsto V(q)-\kappa_{c}c(q)$. By the implicit function theorem:
 		$\frac{\partial}{\partial\kappa_{g}}q_{\kappa}^{M}=-G_{3}(q_{\kappa}^{M},\kappa_{c},\kappa_{g})/G_{1}(q_{\kappa}^{M},\kappa_{c},\kappa_{g})$.
 		Hence, $q_{\kappa}^{M}$ is nondecreasing in $\kappa_{g}$.
 		
 		\emph{Claim 3: $b_{\kappa_{g}}$ is nonincreasing in $\kappa_{g}$
 			pointwise.} Fix $q$ with $b_{\kappa_{g}}(q)\in(0,1)$. From $\kappa_{g}g'(q)+\varphi(b_{\kappa_{g}}(q))=0$,
 		by the implicit function theorem: $\frac{\partial}{\partial\kappa_{g}}b_{\kappa_{g}}(q)=-g'(q)/\varphi'(b_{\kappa_{g}}(q))$.
 		Hence, for every $q\in Q$, $b_{\kappa_{g}}(q)$ is nonincreasing
 		in $\kappa_{g}$.
 		
 		\emph{Claim 4: $b_{\kappa_{g}}(q_{\kappa}^{M})$ is nonincreasing
 			in $\kappa_{g}$, decreasing if interior.} First, we observe that $G_1(q, \kappa_g, \kappa_c) = (1-F(b_{\kappa_g}(q_\kappa^M)))\kappa_gg''(q^M_\kappa)-c''(q^M_\kappa)$, which follows from applying the definition of $b_{\kappa_g}$, or, alternatively, by using the observations in Remark \ref{rem:tariff} and the envelope theorem. Second, we differentiate to obtain $\frac{\partial }{\partial \kappa_g} b_{\kappa_g}(q^M_\kappa) = \frac{g'(q^M_\kappa)}{\varphi '(b_{\kappa_g}(q	^M_\kappa))} \frac{c''(q^M_\kappa)}{G_1(q^M_\kappa, \kappa_g, \kappa_c)}<0$, whenver $b_{\kappa_g}(q^M_\kappa)>0$.
 		
 		\emph{Claim 5: there exists $\overline{\kappa}_{g}\ge0$ such that
 			$b_{\kappa_{g}}(q_{\kappa}^{M})=0$ for all $\kappa_{g}>\overline{\kappa}_{g}$.
 		}By the previous claims, $b_{\kappa_{g}}(q_{\kappa}^{M})\to0$ as
 		$\kappa_{g}\to\infty$. If the claim does not hold, then $\kappa_{g}g'(q_{\kappa}^{M})\to-\varphi(0)$,
 		which is impossible by concavity of $g$ and monotonicity of $q_{\kappa}^{M}$
 		in $\kappa_{g}$. 
 	\end{proof}
 	\begin{proof}[Proof of Proposition \ref{prop:singleagentallocation}]
 		For this proof, we assume that $u$ is as in the main text. We show
 		that, if $0\le\theta<b(q^{M})$, then $\bm{q}^{MRE}(\theta)>\bm{q}^{M}(\theta)$.
 		By hypothesis, 
 		\begin{align*}
 			c'(\bm{\beta}(\theta)) & <c'(q^{M})=(1-F(b(q^{M})))u_{1}(q^{M},b(q^{M})),
 		\end{align*}
 		in which the strict inequality follows from Proposition \ref{app:prop:mon:alloc}.
 		By strict concavity of $V$, we have 
 		\begin{align*}
 			(1-F(b(q^{M})))u_{1}(q^{M},b(q^{M}))<(1-F(\theta))u_{1}(\bm{\beta}(\theta),\theta)\le u_{1}(\bm{\beta}(\theta),\theta),
 		\end{align*}
 		We conclude that $u_{1}(\bm{\beta}(\theta),\theta)-c'(\bm{\beta}(\theta))>0$.
 		Therefore, $\bm{\beta}(\theta)<\bm{q}^{MRE}(\theta)$.
 		
 		The rest of the proof follows from the argument sketched in the text
 		and is omitted. 
 	\end{proof}
 	
 	\subsection{Proofs for Section \ref{sec:extensions}}
 	
 	Proposition \ref{prop:noscreening} and \ref{prop:competition} are
 	implied by Proposition \ref{app:prop:mon:noscreen} and \ref{app:prop:competition},
 	respectively. 
 	
 	\subsubsection{Non-discriminating monopolist}
 	
 	The zero of the virtual surplus is $b_{N}\colon q\mapsto\inf\{\theta\in\Theta\mid J(q,\theta)\ge0\}$.
 	The allocation $\bm{q}$ is \emph{no screening} if $\bm{q}$ solves
 	\begin{align*}
 		\max_{(\bm{q},\,t)\in\bm{M}}\int_{\Theta}t(\theta)\diff F(\theta)-c(\sup\bm{q})\ \text{subject to:}\ \bm{q}(\Theta)\subseteq\{0,\hat{q}\}\ \text{for some}\ \hat{q}\in Q.
 	\end{align*}
 	We let $V_{N}(q)$ be the value of the problem $\mathcal{P}_{N}(q)$,
 	for quality $q\in Q$, 
 	\begin{align*}
 		& V_{N}(q)\coloneqq\max_{(\bm{q},\,t)\in\bm{M}}\int_{\Theta}t(\theta)\diff F(\theta)\ \text{subject to:}\ \bm{q}(\Theta)\subseteq\{0,\hat{q}\}\ \text{for some}\ \hat{q}\in Q,\\
 		& \bm{q}(\theta)\le q\ \text{for all}\ \theta\in\Theta.
 	\end{align*}

 	\begin{proposition}\label{app:prop:mon:noscreen} Assume that $J(0,\theta)=0$
 		for all $\theta$, and $J(q,\cdot)$ is increasing for all $q>0$;
 		let $q_{N}^{M}$ be the unique quality $q$ such that $\int_{[b_{N}(q),1]}J_{1}(q,\theta)\diff F(\theta)=c'(q)$.
 		The allocation $\bm{q}_{N}^{M}$ is no screening if and only if: $\bm{q}_{N}^{M}(\theta)=[\theta\ge b_{N}(q_{N}^{M})]q_{N}^{M}$
 		for all $\theta\in\Theta\setminus\{b_{N}(q_{N}^{M})\}$, and $\bm{q}_{N}^{M}(b_{N}(q_{N}^{M}))\in\{0,q_{N}^{M}\}$.
 		Moreover, it holds that: 
 		\begin{enumerate}
 			\item $0<q_{N}^{M}\le q^{M}$; 
 			\item  if $b(q^{M})>b_{N}(q^{M})$, then $q_{N}^{M}<q^{M}$.
 		\end{enumerate}
 	\end{proposition}
 	\begin{proof}
 		The proof has three main steps. First, we solve $\mathcal{P}_{N}(q)$
 		for all $q\in Q$; second, we show that $V_{N}$ is continuously differentiable
 		and concave; third, we show that $q_{N}^{M}\le q^{M}$, with strict
 		inequality if $q^{M}>0$. The following preliminary observations hold
 		by known arguments: 
 		\begin{align*}
 			& V_{N}(q)=\max_{\bm{q}\in\bm{Q}}\int_{\Theta}J(\bm{q}(\theta),\theta)\diff F(\theta)\ \text{subject to:}\ \bm{q}(\theta)\le q\ \text{for all}\ \theta\in\Theta,\\
 			& \bm{q}\ \text{is nondecreasing},\ \bm{q}(\Theta)\subseteq\{0,\hat{q}\}\ \text{for some}\ \hat{q}\in Q;
 		\end{align*}
 		$(\bm{q},t)$ solves $\mathcal{P}_{N}(q)$ for some $t$ iff $\bm{q}$
 		solves the above problem; the allocation $\bm{q}$ is no screening
 		iff: $\bm{q}$ solves $\mathcal{P}_{N}(q_{N}^{M})$ for a quality
 		$q_{N}^{M}\in\argmax_{q\in Q}V_{N}(q)-c(q)$.
 		
 		\emph{Claim: $b(q)\ge b_{N}(q)$ for all $q\in Q$.} It suffices to
 		show that $\bm{\beta}(\theta)\ge q$ implies $J(q,\theta)\ge0$ for
 		fixed $(\theta,q)\in\Theta\times Q$. By concavity of $J(\cdot,\theta)$,
 		we have $J(q,\theta)=\int_{[0,q]}J_{1}(\tilde{q},\theta)\diff\tilde{q}$,
 		and $\bm{\beta}(\theta)\ge q$ holds iff $J_{1}(q,\theta)\ge0$ holds.
 		Hence, $\bm{\beta}(\theta)\ge q$ implies $J(q,\theta)\ge0$.
 		
 		\emph{Claim: for fixed $q\in Q$, $\bm{q}$ solves $\mathcal{P}_{N}(q)$
 			iff $\bm{q}(\theta)=[\theta\ge b_{N}(q)]q$ for all $\theta\in\Theta$.}
 		We first argue $\bm{q}$ solves $\mathcal{P}_{N}(q)$ only if $\bm{q}(\theta)=[\theta\ge b_{N}(q)]q$
 		almost everywhere. Fix $q\in(0,\overline{q}]$, a solution $\bm{q}_{N}^{M}$
 		to $\mathcal{P}_{N}(q)$, and suppose that there exists $(\theta',\theta'')\subseteq\Theta$
 		such that: $\bm{q}'(\theta)\ne\bm{q}_{N}^{M}(\theta)$ a.e.~on $(\theta',\theta'')$,
 		for $\bm{q}'\colon\theta\mapsto[\theta\ge b_{N}(q)]q$. Define $S\coloneqq\{\theta\in(\theta',\theta'')\mid \bm{q}'(\theta)=q\}$
 		and $\Delta\coloneqq\int_{(\theta',\theta'')}J(\bm{q}'(\theta),\theta)-J(\bm{q}_{N}^{M}(\theta),\theta)\diff F(\theta)$,
 		it holds that 
 		\begin{align*}
 			\Delta=\int_{S}\int_{[\bm{q}_{N}^{M}(\theta),q]}J_{1}(\tilde{q},\theta)\diff\tilde{q}\diff F(\theta)-\int_{(\theta',\theta'')\setminus S}J(\bm{q}_{N}^{M}(\theta),\theta)\diff F(\theta).
 		\end{align*}
 		For $\hat{q}\in[\bm{q}_{N}^{M}(\theta),q]$, it holds that 
 		\begin{align*}
 			J_{1}(\hat{q},\theta)\ge J_{1}(q,\theta)\ge J_{1}(q,b_N(q))\ge0
 		\end{align*}
 		a.e.~on $S$, in which the inequalities follow from: concavity of
 		$J(\cdot,\theta)$, increasing differences,
 		and the definition of $S,\:b_{N}(q)$, from left to right. Moreover,
 		it holds that $J(\bm{q}_{N}^{M}(\theta),\theta)\le0$ a.e.~on $(\theta',\theta'')\setminus S$.
 		Moreover, using increasing differences and the fact that $J(\hat{q},\cdot)$
 		is increasing for $\hat{q}>0$, we have that: either $\int_{S}\int_{[\bm{q}_{N}^{M}(\theta),q]}J_{1}(\tilde{q},\theta)\diff\tilde{q}\diff F(\theta)>0$,
 		or $\int_{(\theta',\theta'')\setminus S}J(\bm{q}_{N}^{M}(\theta),\theta)\diff F(\theta)<0$,
 		or both. It follows that $\Delta>0$. By incentive compatibility,
 		the solution to $\mathcal{P}_{N}(q)$ is nondecreasing, by the binary-image
 		constraint, the solution does not differ from $\theta\mapsto[\theta\ge b_{N}^{M}(q)]$
 		if $\theta\ne b_{N}(q)$.
 		
 		For the other direction, let $\bm{q}$ solve $\mathcal{P}_{N}(q)$
 		and not be equal to $\bm{q}'$ (identified in the previous paragraph)
 		almost everywhere. By the previous argument, $\bm{q}'$ attains a
 		strictly higher value of the maximand in $\mathcal{P}_{N}(q)$ than
 		$\bm{q}$. Moreover, $\bm{q}'$ is nondecreasing and satisfies $\bm{q}'(\Theta)\subseteq\{0,q\}$.
 		The claim follows.
 		
 		\emph{Claim: $b_{N}$ is nondecreasing.} We proceed in three steps.
 		First, we use the definition of $\bm{\beta}$ and the fact that $b$
 		upper bounds $b_{N}$ to show that: $J(q+\varepsilon,b_{N}(q))\le J(q,b_{N}(q)),\:(\varepsilon,q)\in(0,\overline{q}-q)\times(0,\overline{q})$.
 		Suppose that $J(q+\varepsilon,b_{N}(q))>J(q,b_{N}(q))$. Then, there
 		exists $q'>q$ such that $J_{1}(q',b_{N}(q))>0$. Then, by concavity,
 		either $\bm{\beta}(b_{N}(q))=\overline{q}$, or there exists $q''\ge q'$
 		such that $J_{1}(q'',b_{N}(q))\le0$, or both. In all cases: $\bm{\beta}(b_{N}(q))>q$,
 		which contradicts the definition of $b$. Second, we observe that:
 		if $J(q+\varepsilon,b_{N}(q))<J(q,b_{N}(q))$, then $b_{N}(q+\varepsilon)>b_{N}(q)$
 		by increasing differences. Third, we show that $J(q+\varepsilon,b_{N}(q))=J(q,b_{N}(q))$
 		only if $b_{N}(q+\varepsilon)=b_{N}(q)$ and distinguish among three
 		exhaustive and mutually exclusive cases, for given $q\in Q$: (A)
 		$J(q,0)\ge0$; (B) $J(q,1)\le0$; (C) $b_{N}(q)\in(0,1)$. In case
 		(A), $b_{N}(q+\varepsilon)=b_{N}(q)$ because, otherwise: $b_{N}(q+\varepsilon)>b_{N}(q)=0$,
 		so $J(q+\varepsilon,b_{N}(q+\varepsilon))>J(q+\varepsilon,b_{N}(q))\ge0$,
 		which contradicts the definition of $b_{N}$. In case (B), $b_{N}(q+\varepsilon)=b_{N}(q)$
 		because, otherwise: $b_{N}(q+\varepsilon)<b_{N}(q)=1$, so $J(q+\varepsilon,b_{N}(q+\varepsilon))<J(q+\varepsilon,b_{N}(q))\le0$,
 		which contradicts the definition of $b_{N}$. In case (C), $b_{N}(q+\varepsilon)=b_{N}(q)$
 		by increasing differences. Hence, $b_{N}$ is nondecreasing.
 		
 		\emph{Claim: $V_{N}$ is continuously differentiable, concave, and
 			$V_{N}'(q)=\int_{[b_{N}(q),1]}J_{1}(q,\theta)\diff F(\theta)$}. First,
 		it holds that $b_{N}$ satisfies: $b_{N}(q)=0$ if $q\le c$, and
 		$b_{N}(q)=1$ if $q\ge d$, for some $c\le d$, because $b_{N}$ is
 		nondecreasing. Moreover, $b_{N}$ is continuously differentiable at
 		$q\in(c,d)$, with $b_{N}'(q)=-\frac{J_{1}(q,b_{N}(q))}{J_{2}(q,b_{N}(q))}$,
 		by the implicit function theorem. As an implication, by the same argument
 		used in the proof of Proposition \ref{app:prop:mon:alloc}, $V_{N}$
 		is continuously differentiable at $q\in(0,\overline{q})$ with derivative
 		$\int_{[b_{N}(q),1]}J_{1}(q,\theta)\diff F(\theta)$. For the claim,
 		it suffices to establish that $V'_{N}(q_{2})-V'_{N}(q_{1})\le0$ for
 		all $q_{1},q_{2}\in Q$ with $q_{2}>q_{1}$. It holds that 
 		\begin{align*}
 			V'_{N}(q_{2})-V'_{N}(q_{1})=\int_{[b_{N}(q_{2}),1]}J_{1}(q_{2},\theta)-J_{1}(q_{1},\theta)\diff F(\theta)-\int_{[b_{N}(q_{1}),b_{N}(q_{2})]}J_{1}(q_{1},\theta)\diff F(\theta).
 		\end{align*}
 		Moreover, $J_{1}(q_{2},\theta)-J_{1}(q_{1},\theta)\le0$ by concavity
 		of $q\mapsto J(q,\theta)$ for all $\theta\in\Theta$, and $J_{1}(q_{1},\theta)\ge0$
 		for all $\theta\ge b_{N}(q_{1})$ by concavity of $J(\cdot,\theta)$
 		and the definition of $b_{N}$. Hence, $V_{N}'$ is concave; $V_{N}$
 		is right continuous at 0 by the same argument used in the proof of
 		Proposition \ref{app:prop:mon:alloc} to show that $V$ is right continuous
 		at 0.
 		
 		\emph{Claim: $q_{N}^{M}\le q^{M}$, with strict inequality if $b(q^{M})>b_{N}(q^{M})$.}
 		We show that $\Delta\coloneqq V'(q)-V'_{N}(q)\le0$. It holds that
 		$\Delta=-\int_{[b_{N}(q),b(q)]}J_{1}(q,\theta)\diff F(\theta)$, and,
 		by definition of $b(q)$, $J_{1}(q,\theta)\le0$ for all $\theta\le b(q)$.
 		Moreover, if $b(q)>b_{N}(q)$ and $q>0$, then $\Delta>0$. It remains
 		to show that $q_{N}^{M}>0$, which follows from the same argument
 		used in the proof of Proposition \ref{app:prop:mon:alloc} to show
 		that $q^{M}>0$. 
 	\end{proof}
 	
 	\subsubsection{Competition}
 	
 	\label{app:sec:game:definitions}\label{app:sec:competition:auxiliary}
 	
 	We study a game played by $N$ firms facing the same buyer population
 	as in the main text, i.e., we assume that $u$  and $c$ are as in the main text; recall that $c'(0)=0$	We use the following notation and results.
 	
 	\paragraph{Setup}
 	
 	We denote by $\mathcal{N}$ the set of players $\{1,\dots,N\}$, and
 	we let $Q=\mathbb{R}_{+}$ be the set of qualities. Fix a profile
 	of qualities $(\overline{q}_{1},\dots,\overline{q}_{N})\in Q^{N}$.
 	The demand of type $\theta$ given the profile of tariffs $(p_{1},\dots,p_{N})\in\times_{i\in\mathcal{N}}\mathbb{R}^{Q}$
 	satisfying $p_{i}(q)=\infty$ if $q>\overline{q}_{i}$, $i\in\mathcal{N}$,
 	is: $D_{(p_{1},\dots,p_{N})}(\theta)\coloneqq\argmax_{q\in Q}u(\theta,q)-\min\{p_{1}(q),\dots,p_{N}(q)\}$.
 	The player serving $\theta$ given the profile of tariffs $(p_{1},\dots,p_{N})$
 	is $\iota_{(p_{1},\dots,p_{N})}(\theta)\coloneqq\min\argmin_{i\in\mathcal{N}}p_{i}(D_{(p_{1},\dots,p_{N})}(\theta))$;
 	if multiple firms offer the same price for the demand, we break indifferences
 	in favor of the low-index firms. Note that every type $\theta>0$
 	has single-valued demand.
 	
 	\paragraph{Timing and actions}
 	
 	First, every player $i$ simultaneously chooses a quality $\overline{q}_{i}\in Q$.
 	In the second stage of the game, the \emph{pricing} stage, every firm
 	$i$ observes the qualities $(\overline{q}_{1},\dots,\overline{q}_{N})$
 	and simultaneously chooses a pricing function $p_{i}\colon Q\to\mathbb{R}$
 	that only prices feasible qualities, i.e., with $p_{i}(q_{i})=\infty$
 	if $q_{i}>\overline{q}_{i}$. Then, every buyer's type $\theta$ observes
 	the pricing functions and, afterwards, he chooses a firm $i$, purchases
 	a quality $q_{i}$ at price $p_{i}(q_{i})$, and nothing from the
 	other firms, or does not purchase any good for a utility of 0.
 	
 	\paragraph{Strategies, payoffs, and equilibria}
 	
 	Firm $i$ chooses a cap in stage 1 and a pricing function in stage
 	2. There is symmetric information among firms in stage 2, so every
 	player $i$ quotes a pricing function conditional on the known cap
 	quality of every opponent. Hence, the set of pure strategies for $i$
 	is $S_{i}\coloneqq Q\times\bm{P}_{i}$, letting $\bm{P}_{i}\subseteq(\mathbb{R}^{Q})^{Q^{N}}$
 	be the set of ``conditional'' pricing functions of firm $i$; i.e.,
 	we define $\mathbf{P}_{i}=\left\{ P\colon Q^{N}\to\mathbb{R}^{Q}\mid q>\overline{q}_{i}\ \text{implies}\ P[\dots,\overline{q}_{i-1},\overline{q}_{i},\overline{q}_{i+1},\dots](q)>g(\overline q)+\overline q \right\} $.
 	The revenues of firm $i$ given a profile of pricing functions $(p_{1},\dots,p_{N})$
 	are
 	
 	\begin{align*}
 		R_{i}(p_{1},\dots,p_{N})\coloneqq\int_{\Theta}p_{i}(D_{(p_{1},\dots,p_{N})}(\theta))[\iota_{(p_{1},\dots,p_{N})}(\theta)=i]\diff F(\theta).
 	\end{align*}
 	
 	The \emph{payoff} of firm $i$ from the profile of pure strategies
 	$s=(\dots,(\overline{q}_{i}^{s},P_{i}^{s}),\dots)\in\times_{i=1}^{N}S_{i}$
 	is $\Pi_{i}(s)\coloneqq R_{i}(P_{1}^{s}[\overline{q}^{s}],\dots,P_{N}^{s}[\overline{q}^{s}])-c(\overline{q}_{i}^{s})$.
 	A strategy for player $i$ is a pair of a cap distribution and conditional
 	pricing function, $(H_{i},P_{i})$, in which $H_{i}$ is a distribution
 	function with support contained in $Q$ and $P_{i}\in\mathbf{P}_{i}$
 	for all $i$. A strategy profile $((H_{1},P_{1}),\dots,(H_{N},P_{N}))$
 	is an \emph{equilibrium} if it is a subgame-perfect Nash equilibrium
 	of the game, that is, the following two conditions hold. 
 	\begin{enumerate}
 		\item The conditional pricing functions maximize conditional profits, that
 		is, for all $i$, for all cap profiles $\overline{\overline{q}}=(\overline{q}_{1},\dots,\overline{q}_{N})$,
 		\begin{align*}
 			R_{i}(P_{i}[\overline{\overline{q}}],P_{-i}[\overline{\overline{q}}])\ge R_{i}(p_{i}',P_{-i}[\overline{\overline{q}}])\quad\text{for all}\ p_{i}'\in\mathbb{R}^{Q}.
 		\end{align*}
 		\item The cap distributions maximize expected profits, that is, for all $i$,  for all
 		$\overline{q}_{i}$ in the support of $H_{i}$ we have 
 		\begin{align*}
 			&\int_{Q^{N-1}}R_{i}(P_{i}[\overline{\overline{q}}],P_{-i}[\overline{\overline{q}}])\diff H_{-i}(\overline{q}_{-i})-c(\overline{q}_{i})\\
 			&\ge\int_{Q^{N-1}}R_{i}(P_{i}[\overline{q}_{i}',\overline{q}_{-i}],P_{-i}[\overline{q}_{i}',\overline{q}_{-i}])\diff H_{-i}(\overline{q}_{-i})-c(\overline{q}_{i}')
 		\end{align*}
 		for all $\overline{q}_{i}'\in Q$, letting $H_{-i}$ denote the joint
 		distribution of the caps of players other than $i$ in $((H_{1},P_{1}),\dots,(H_{N},P_{N}))$.
 	\end{enumerate}
 	
 	\paragraph{Auxiliary monopoly problem}
 	
 	The allocation $\bm{q}\in\bm{Q}$ is \emph{$x$-$y$ second best},
 	for qualities $x,\,y\in Q$ with $x\le y$, if there exists a transfer
 	function $t$ such that: $(\bm{q},t)$ solves  the problem $\mathcal P^{x, y}$
 	\begin{align*}
 		V^{x,y}=\sup_{(\bm{q},t)\in\mathcal{M}}\int t(\theta)\diff F(\theta)\ \text{subject to:}\ y\le\bm{q}(\theta)\le x\ \text{for all}\ \theta\in\Theta.
 	\end{align*}
 	\begin{lemma} Let $x,\,y\in Q$ with $x\le y$. The allocation $\bm{q}$
 		is $x$-$y$ second best if and only if: there exists a nondecreasing
 		allocation $\bm{\gamma}$ such that $\bm{\gamma}(\theta)=\bm{\beta}(\theta)$
 		almost everywhere and $\bm{q}(\theta)=\max\{\min\{\bm{\gamma}(\theta),x\},y\}$
 		for all $\theta$.\end{lemma} 
 	\begin{proof}
 		This result follows from the same argument as the proof of
 		Proposition \ref{app:prop:mon:alloc}.
 		
 		Two preliminary observations follow from standard arguments \citep*{carroll_contract_2023}.
 		First, we have 
 		\begin{align*}
 			V^{x,y}=&\max_{\bm{q}\in\bm{Q}}\int_{\Theta}J(\bm{q}(\theta),\theta)\diff F(\theta)\ \text{subject to:}\ y\le \bm{q}(\theta)\le x\ \text{for all}\ \theta\in\Theta,\\& \bm{q}\ \text{is nondecreasing};
 		\end{align*}
 		and $(\bm{q},t)$ solves $\mathcal{P}^{x, y}$ for some $t$ iff $\bm{q}$
 		solves the above problem.
 		
 		We establish that $\bm{q}$
 		solves $\mathcal{P}^{x, y}$ if and only if: $\bm{q}(\theta)=\max\{\min\{\bm{\gamma}(\theta),x\},y\}$
 		for all $\theta$ and some nondecreasing $\bm{\gamma}\in\bm{Q}$ with
 		$\bm{\gamma}(\theta)=\bm{\beta}(\theta)$ almost everywhere.
 		
 		For necessity, let $\bm{q}\in\bm{Q}$ be such that $\bm{q}(\theta)=\max\{\min\{\bm{\gamma}(\theta),x\},y\}$
 		for all $\theta\in\Theta$ and for some nondecreasing allocation $\bm{\gamma}$
 		with $\bm{\gamma}(\theta)=\bm{\beta}(\theta)$ almost everywhere.
 		Then,  by Lemma \ref{app:lem:general}, $\bm{q}$ solves the above problem
 		without the monotonicity constraint, i.e., $\bm q$ solves
 		\begin{align}\label{app:eq:xyunconstrained}
 		\max_{\bm{q}\in\bm{Q}}\int_{\Theta}J(\bm{q}(\theta),\theta)\diff F(\theta)\ \text{subject to:}\ y\le \bm{q}(\theta)\le x\ \text{for all}\ \theta\in\Theta.
 		\end{align}
		Additionally, $\bm q$ is nondecreasing, so $\bm{q}$ is $x$-$y$ second best.
 		
 		For sufficiency, let $\bm{q}$ be  $x$-$y$ second best, and such that: there does not exist a nondecreasing
 		allocation $\bm{\gamma}$ with $\bm{q}(\theta)=\max\{\min\{\bm{\gamma}(\theta),x\},y\}$
 		for all $\theta$ and $\bm{\gamma}(\theta)=\bm{\beta}(\theta)$ almost
 		everywhere. By Lemma \ref{app:lem:general} and strict quasiconcavity
 		of $J(\cdot,\theta)$, $\bm{q}$ does not solve the preceding problem---in \ref{app:eq:xyunconstrained}. We observe
 		that, by Lemma \ref{app:lem:general},  the function $\theta\mapsto\max\{ \min\{\bm{\beta}(\theta),x\}, y\}$ is
 		an allocation that solves the problem in \ref{app:eq:xyunconstrained},
 		and satisfies the monotonicity constraint,
 		so $\bm{q}$ does not solve $\mathcal{P}(q)$. The claim follows from
 		the preliminary observation in the proof of Lemma \ref{app:lem:mon:dec}. 
 	\end{proof}
 	
 	\paragraph{Auxiliary results about the pricing game}
 	
 	The pricing game given the quality profile $(q_{1},\dots,q_{N})\in Q^{N}$
 	is the strategic-form game $\Gamma(q_{1},\dots,q_{N})=(\mathcal{N},(\mathbb{R}^{[0,q_{i}]},R_{i}(\cdot))_{i\in\mathcal{N}}$.
 	The profile of pricing functions $(p_{1},\dots,p_{N})$ is a \emph{$(q_{1},\dots,q_{N})$
 		equilibrium} if $(p_{1},\dots,p_{N})$ is a Nash equilibrium of $\Gamma(q_{1},\dots,q_{N})$.
 	
 	We study the subgame starting at the given production profile $(q_{1},\dots,q_{N})\in Q^{N}$
 	with $\max\{q_{1},\dots,q_{N}\}=:x$ and $y:=\max\{q_{1},\dots,q_{N}\}\setminus\{x\}$.
 	\begin{lemma}\label{app:lem:uniquerevenues} Let $q_{i}\in Q$ for
 		all $i\in\mathcal{N}$. For every $(q_{1},\dots,q_{N})$ equilibrium
 		$(p_{1},\dots,p_{N})$, it holds that:
 		\begin{align*}
 			R_{i}(p_{1},\dots,p_{N})=(V(q_{i})-V(\max\{q_{1},\dots,q_{n}\}\setminus\{q_{i}\}))_{+},
 		\end{align*}
 		and 
 		\begin{align*}
 			D_{(p_{1},\dots,p_{N})}(\theta)=\max\{\min\{\bm{\gamma}(\theta),x\},y\},
 		\end{align*}
 		for all $\theta$, in which $\bm{\gamma}$ is a nondecreasing allocation
 		such that $\bm{\gamma}(\theta)=\bm{\beta}(\theta)$ almost everywhere.
 	\end{lemma} 
 	\begin{proof}
 		\emph{(0) Quantities less than $y$ come at zero price.} Fix a $(q_{1},\dots,q_{N})$
 		equilibrium $(p_{1},\dots,p_{N})$. Suppose there exists a positive
 		measure type set $A\subseteq\Theta$ such that $D_{(p_{1},\dots,p_{N})}(\theta)\le y$
 		and purchase at positive price: $p_{\iota_{(p_{1},\dots,p_{N})}(\theta)}D_{(p_{1},\dots,p_{N})}(\theta)>0$
 		for all $\theta\in A$ and $\iota_{(p_{1},\dots,p_{N})}(\theta)=:i$.
 		
 		Firm $j$, that invested $q_{j}\ge y$ and is not $i$, has the following
 		action available: 
 		\begin{align*}
 			p'(q)=\begin{cases}
 				p_{j}(q), & \text{if}\ q\notin D_{(p_{1},\dots,p_{N})}(A),\\
 				p_{i}(q)-\varepsilon, & \text{if}\ q\in D_{(p_{1},\dots,p_{N})}(A).
 			\end{cases}
 		\end{align*}
 		for sufficiently small $\varepsilon>0$. Conditional on type $\theta\in A$,
 		playing $p'$ induces strictly higher profits than $p_{i}$. Conditional
 		on type $\theta\notin A$, playing $p'$ induces lower profits than
 		$p_{i}$ only if $\iota_{(p_{1},\dots,p_{N})}(\theta)=i$, and by
 		a negligible amount if $\varepsilon$ is sufficiently small.
 		
 		\emph{(1) Preliminaries: using the results from the $x$-$y$ second
 			best problem.} We now compute the value of the $x$-$y$ second best
 		problem, defined above. For simplicity, we assume that $g'(0)>-\varphi(0)$.
 		For quantities $0\le y<x\le\overline{\overline{q}}$, it holds that
 		\begin{align*}
 			V^{x,y}=\int_{[\underline{\theta},a]}P(y)\diff F(\theta)+\int_{(a,b)}g(\bm{\beta}(\theta))+\bm{\beta}(\theta)\varphi(\theta)\diff F(\theta)+\int_{[b,\overline{\theta}]}g(x)+x\varphi(\theta)\diff F(\theta),
 		\end{align*}
 		for a pricing function $P$, cutoff types $\underline{\theta}\le a<b\le\overline{\theta}$,
 		and a quality allocation $\bm{q}$ that agrees with $\theta\mapsto\min\{\bm{\beta}(\theta),x\}$
 		on $[a,\overline{\theta}]$. (We note that $a=b(y)$ and $b=b(x)$).
 		
 		Hence, we have 
 		\begin{align*}
 			V^{x,y}=V(x)-V(y)+\int_{[\underline{\theta},b(y)]}P(y)\diff F(\theta).
 		\end{align*}
 		for a pricing function $P$.
 		
 		\emph{(2) Bertrand-type argument.} By step (1), the payoff to $i$
 		if $q_{i}=x>y$ in equilibrium is 
 		\begin{align*}
 			R_{i}(p_{1},\dots,p_{N})=V(x)-V(y)+\int_{[\underline{\theta},b(y)]}[\iota_{(p_{1},\dots,p_{N})}=i]P_{i}(D_{(p_{1},\dots,p_{N})}(\theta))\diff F(\theta).
 		\end{align*}
 		By the argument in part (0), any quality less than $y$ involves zero
 		profits, i.e., the last term on the right-hand side of the above equality
 		is 0. So, the formula holds. 
 	\end{proof}
 	
 	\paragraph{Equilibria of the production game}	\label{app:sec:competitionproofproof}

 	Every Nash equilibrium of $\Gamma(q_{1},\dots,q_{N})$ induces the
 	same revenues, i.e., for all $(q_{1},\dots,q_{N})$ equilibria $(p_{1}^{*},\dots,p_{N}^{*})$
 	and $(p_{1}^{**},\dots,p_{N}^{**})$, we have $R_{i}(p_{1}^{*},\dots,p_{N}^{*})=R_{i}(p_{1}^{**},\dots,p_{N}^{**})$
 	(Lemma \ref{app:lem:uniquerevenues}). We define $\overline{R}_{i}(q_{1},\dots,q_{N})$
 	the unique equilibrium revenues of $i\in\mathcal{N}$ in the pricing
 	game $\Gamma(q_{1},\dots,q_{N})$. By Lemma \ref{app:lem:uniquerevenues},
 	we have 
 	\begin{align*}
 		\overline{R}_{i}(q_{1},\dots,q_{N})=(V(q_{i})-V(\max\{q_{1},\dots,q_{n}\}\setminus\{q_{i}\}))_{+}.
 	\end{align*}
 	
 	The \emph{production game} is the strategic-form game $\Gamma=\left(\mathcal{N},(Q,\overline{R}_{i}(\cdot)-c(\cdot))_{i\in\mathcal{N}}\right)$.
 	By mixed strategies, we refer to elements of $\Delta Q$, the set
 	of probability distributions over $Q$ identified by distribution
 	functions, and we extend payoffs of $\Gamma$ to mixed-strategy profiles
 	as usual; i.e., defining 
 	\begin{align*}
 		\Pi_{i}\colon &(\Delta Q)^{N}\to\mathbb{R}\\
 		&(H_{1},\dots,H_{N})\mapsto\int_{Q}\cdots\int_{Q}\overline{R}_{i}(q_{1},\dots,q_{N})-c(q_{i})\diff H_{1}(q_{1})\cdots\diff H_{N}(q_{N}).
 	\end{align*}

 	\begin{lemma}[Pure-strategy equilibria]\label{app:prop:pureequilibria}
 		The quality profile $(q_{1},\dots,q_{N})$ is a Nash equilibrium in
 		pure strategies of the production game if and only if: there exists
 		$i\in\mathcal{N}$ such that $q_{i}=q^{M}$ and $q_{j}=0$ for all
 		$j\ne i$.\end{lemma} 
 	\begin{proof}
 		Let $q_{i},~q_{j}>0$ for $i\ne j$. Then the player choosing $\min\{q_{i},q_{j}\}$
 		has a strictly profitable deviation to $0$. Thus, in every equilibrium,
 		there exists at most one player choosing a positive quality $q_{i}$.
 		
 		Let $(q_{1},\dots,q_{N})$ be a Nash equilibrium. Then, $\max\{q_{1},\dots,q_{n}\}=q^{M}$,
 		by the characterization of the quality $q^{M}$.
 		
 		It is left to verify that $\overline{R}_{2}(q^{M},0,\dots,0)\ge\overline{R}_{2}(q^{M},q,\dots,0)$
 		for all $q>0$. First, we observe that $\overline{R}_{2}(q^{M},0,\dots,0)=0=\overline{R}_{2}(q^{M},q,0,\dots,0)$
 		for $q\le q^{M}$, by the argument mentioned above. Second, we compute
 		the payoff from qualities higher than $q^{M}$: $V(q)-V(q^{M})-c(q)$.
 		By previous results, $V'(q)<c'(q)$ if $q>{q}^{M}$. So,
 		$q\mapsto\overline{R}_{2}(q^{M},q,0,\dots,0)$ is decreasing on $(q^{M},\overline{q})$.
 		Hence, our claim follows. 
 	\end{proof}
 	
 	A Nash equilibrium in mixed strategies $(H_{1},\dots,H_{N})$ has
 	\emph{multiple active players} if $H_{i}(q)\ge H_{j}(q)>0$ for $q>0$
 	and $i,~j\in\mathcal{N}$ with $i\ne j$. A Nash equilibrium in mixed
 	strategies $(H_{1},\dots,H_{N})$ is \emph{symmetric among active
 		players} if $H_{i}(q)=H_{j}(q)$ for all $q\in Q$ whenever $H_{i}(q)\ge H_{j}(q)>0$
 	for some $q>0$. Note that there are only two kinds of equilibria
 	in mixed strategies that are symmetric among active players: those
 	with at least $2$ active players and the monopolistic equilibria.
 	
\begin{lemma}[Necessary conditions for mixed-strategy equilibria]\label{app:lem:mixedequilibrianecessary}
	If $(H_{1},\dots,H_{N})$ is a Nash equilibrium in mixed strategies with multiple
	active players, then
	\[
	\prod_{j\in\mathcal{N}\setminus\{i\}}H_{j}(q)=
	\begin{cases}
		\dfrac{c'(q)}{V'(q)}, & \text{if }q\in(0,q^{M}),\\[0.6em]
		1, & \text{if }q\ge q^{M},
	\end{cases}
	\]
	for all $i$ such that $H_{i}(q)>0$ for some $q>0$.
\end{lemma}

\begin{proof}
	Fix a mixed-strategy equilibrium $(H_{1},\dots,H_{N})$ with $n\ge2$ active firms. Fix an active firm $i$. Let $ H_{-i}(q):=\prod_{j\in\mathcal N\setminus\{i\}}H_{j}(q)$ be the distribution function of the maximum of the opponents' caps. 
	
	\emph{Step 1: Payoff representation and one-sided derivatives.}
	Using integration by parts and the fact that $V(0)=0$, we obtain
	\begin{equation*}
		\Pi_{i}(q,H_{-i})
		=
		\int_{0}^{q}V'(y)H_{-i}(y)\,dy-c(q).
	\end{equation*}
	Fix $q>0$. Define the left limit $H_{-i}(q^-):=\lim_{t\uparrow q}H_{-i}(t)$. The following
	one-sided derivatives exist:
	\begin{align*}
		&\lim_{\varepsilon\downarrow0}\frac{\Pi_{i}(q+\varepsilon,H_{-i})-\Pi_{i}(q,H_{-i})}{\varepsilon}
		=
		V'(q)H_{-i}(q)-c'(q),\\
		&\lim_{\varepsilon\downarrow0}\frac{\Pi_{i}(q,H_{-i})-\Pi_{i}(q-\varepsilon,H_{-i})}{\varepsilon}
		=
		V'(q)H_{-i}(q^-)-c'(q).
	\end{align*}

	\emph{Step 2: Indifference on the support implies a differential equation}. If firm $i$ mixes, then $\Pi_{i}(\cdot,H_{-i})$ is constant on $\supp H_{i}$. Fix a compact interval $[a,b]\subseteq\supp H_{i}$ with $0<a<b$. Then, by Step 1, for all	$a\le q_{1}<q_{2}\le b$, we have
	\[
	0=\Pi_{i}(q_{2},H_{-i})-\Pi_{i}(q_{1},H_{-i})
	=
	\int_{q_{1}}^{q_{2}}\bigl(V'(y)H_{-i}(y)-c'(y)\bigr)\,dy,
	\]
	If there existed an open subinterval $(q_{1},q_{2})\subseteq(a,b)$ on which $V'(y)H_{-i}(y)-c'(y)>0$ everywhere, then the integral would be strictly positive; similarly if $V'(y)H_{-i}(y)-c'(y)<0$ everywhere, the integral would be strictly negative. Therefore,
	\begin{equation}\label{eq:ODE-ae}
		V'(y)H_{-i}(y)=c'(y)\ \text{for a.e.~}y\in(a,b).
	\end{equation}

	\emph{Step 3: Exclusion of mass points on $(0,\infty)$.}
	Suppose $H_{i}$ has an atom at some $q_{0}>0$, that is, $H_{i}(q_{0})-H_{i}(q_{0}^-)>0$. For every sufficiently small $\varepsilon>0$, using step 1, we have
	\[
	V_i'(q_0, H_{-i})H_{-i}(q_0) - c'(q_0) \le 0 \le 	V_i'(q_0, H_{-i})H_{-i}(q_0^-) - c'(q_0). 
	\]
	Hence, no firm other than $i$ has an atom at $q_0$. Moreover, it is the case that there exists an active firm $j$ such that $q_0\in \supp H_j$---otherwise, firm $i$ has a strictly profitable deviation. By the one-sided derivatives in step 1, $H_{-j}$ does not have an atom at $q_0$. Therefore, $H_i$ does not exhibit an atom at $q_0$.

	\emph{Step 4. The support of an active firm has no gaps.}
	Suppose, toward a contradiction, that $\supp H_i\cap(0,\infty)$ is not an interval.
	Then there exist $0<q_1<q_2<q_3$ such that
	\[
	q_1,q_3\in\supp H_i
	\quad\text{and}\quad
	[q_1,q_3]\cap\supp H_i=\{q_1,q_3\}.
	\]
	By indifference, $\Pi_i(q_1)=\Pi_i(q_3)$, and using Step 1,
	\begin{equation}\label{eq:gap}
		0
		=
		\int_{q_1}^{q_3} V'(y)H_{-i}(y)-c'(y) \diff y .
	\end{equation}
	Therefore the integral in \eqref{eq:gap} cannot be zero, a contradiction.
	Thus $\supp H_i\cap(0,\infty)$ is an interval.

	\emph{Conclusion.} In any equilibrium with multiple active players, the support of $H_{i}$ is an interval and $H_{i}$ is continuous on $\mathbb{R}_{++}$. As an implication, $\max\supp H_{i}=q^{M}$. We find the mass at $0$ by solving the differential equation downward. The mass is null, because $c'(0)=0$. Moreover, we have that $0=\min\supp H_{i}$, because $c'(q)>0$ for $q>0$. 	
	Note that our argument also establishes that the equilibrium $(H_1, \dots, H_n)$ is symmetric among active firms, and exhibits continuous distributions for active firms.
\end{proof}

 	\begin{lemma}[Necessary and sufficient conditions for mixed-strategy
 		equilibria]\label{app:prop:mixedequilibria} If $(H_{1},\dots,H_{N})$
 		is a Nash equilibrium in mixed strategies of the production game, then one, and only one, of the
 		following statements holds. 
 		\begin{enumerate}
 			\item There exists a number of active players $n\in\mathcal{N}\setminus\{1\}$
 			such that: for a set of active players $I\subseteq\mathcal{N}$ with
 			$|I|=n$ we have 
 			\begin{align*}
 				H_{i}\colon q\mapsto\begin{cases}
 					\left(\frac{c'(q)}{V'(q)}\right)^{\frac{1}{n-1}}, & \text{if}\ q\in[0,q^{M}],\\
 					1, & \text{if}\ q>q^{M},
 				\end{cases}
 			\end{align*}
 			for all $i\in I$. 
 			\item There exists an active player $i$ such that: $H_{i}\colon q\mapsto[q\ge q^{M}]$
 			and $H_{j}\colon q\to0$ for all $j\in\mathcal{N}\setminus\{i\}$. 
 		\end{enumerate}
 		Moreover, the strategy profiles described in both parts constitute
 		Nash equilibria in mixed strategies that are symmetric among active
 		players. \end{lemma} 
 	\begin{proof}
 		Given the preceding results, it is left to verify that the strategy profile described in part 1
 		constitutes a Nash equilibrium in mixed strategies. It follows from
 		the definition of ${q}^{M}$ and by a similar argument as
 		the one used to characterize the pure-strategy equilibria. 
 	\end{proof}
 	
	\paragraph{Proof of Proposition \ref{prop:competition}}

 	\begin{proposition}\label{app:prop:competition} For all $n\le N$,
 		there exists an equilibrium with $n$ active firms. Moreover, every competitive $n$ equilibrium induces the random allocation given
 		by $\bm{q}[\hat{x},\hat{y}]$, in which the random variables $\hat{x}$
 		and $\hat{y}$ are, respectively, the first- and second-order statistics
 		of the collection of $n$ i.i.d.~random variables, each with distribution
 		function $H_{n}$ given by 
 		\begin{align}
 			H_{n}\colon q\mapsto\begin{cases}
 				\left(\frac{c'(q)}{V'(q)}\right)^{\frac{1}{n-1}}, & \text{if}\ q\in[0,q^{M}],\\
 				1, & \text{if}\ q\in(q^{M},\infty).
 			\end{cases}\label{eq:CDF}
 		\end{align}
 	\end{proposition}
 	\begin{proof}
 		This proof uses previously established results (Lemma \ref{app:prop:mixedequilibria},
 		Lemma \ref{app:lem:uniquerevenues}, and Lemma \ref{app:lem:mixedequilibrianecessary})
 		and the definitions in Section \ref{app:sec:game:definitions}.
 		
 		\emph{Step 1: Existence. }The first part of the result (i.e., existence)
 		is established in Lemma \ref{app:prop:mixedequilibria}.
 		
 		\emph{Step 2: Quality allocation. }We establish that every allocation
 		that is induced in equilibrium takes the stated form. Fix a value
 		of $\hat{x}:=\max\{\overline{q}_{1},\dots,\overline{q}_{N}\}$ and
 		$\hat{y}:=\max\{\overline{q}_{1},\dots,\overline{q}_{N}\}\setminus\{\hat{x}\}$
 		in $Q$, given the subgame starting after the choice of caps $(\overline{q}_{1},\dots,\overline{q}_{N})$.
 		The result follows from the last part of Lemma \ref{app:lem:uniquerevenues}.
 		
 		\emph{Step 3: First- and second-order statistics. }It is left to establish
 		that the stated distribution of $\hat{x}$ and $\hat{y}$ is correctly
 		induced by the equilibrium distribution of stage-1 caps $(\overline{q}_{1},\dots,\overline{q}_{N})$.
 		This result is established by Lemma \ref{app:lem:mixedequilibrianecessary}.
 		In particular, as a result of the definitions in Section \ref{app:sec:game:definitions},
 		the profile of strategies $((H_{1},P_{1}),\dots,(H_{N},P_{N}))$ is
 		a competitive equilibrium if, and only if, the profile
 		of distribution functions $(H_{1},\dots,H_{N})$ is a mixed-strategy
 		Nash equilibrium of the production game with multiple active players
 		that is symmetric among active players.
% 		\footnote{We note that, by Lemma \ref{app:prop:mixedequilibria}, for
% 			all $n>2$ there exists a mixed-strategy Nash equilibrium of the production
% 			game with multiple active players that is symmetric among active players
% 			such that: for a set $I\subseteq\mathcal{N}$ with $|I|=n$, 		for all $i\in I$, we have
% 			\begin{align*}
% 				H_{i}\colon q\mapsto\begin{cases}
% 					\left(\frac{c'(q)}{V'(q)}\right)^{\frac{1}{n-1}}, & \text{if}\ q\in[0,{q}^{M}],\\
% 					1, & \text{if}\ q>{q}^{M},
% 				\end{cases}
% 			\end{align*}
% 	Moreover, if we have a mixed-strategy Nash equilibrium
% 			of the production game with multiple active players that is symmetric
% 			among active players, then there exists $n>2$ such that: for a set
% 			$I\subseteq\mathcal{N}$ with $|I|=n$, 		for all $i\in I$, we have 
% 			\begin{align*}
% 				H_{i}\colon q\mapsto\begin{cases}
% 					\left(\frac{c'(q)}{V'(q)}\right)^{\frac{1}{n-1}}, & \text{if}\ q\in[0,{q}^{M}],\\
% 					1, & \text{if}\ q>{q}^{M}.
% 				\end{cases}
% 			\end{align*}
% 		}
 	\end{proof}
 	\begin{remark} \label{rem:monopolyeq} The monopoly allocation arises
 		as the unique outcome of the game under pure strategies. First, we
 		argue that firm $i$ producing $q^{M}$ and all other firms staying
 		idle constitutes an equilibrium. Firm $i$ is a monopolist when her
 		opponents are idle, so her best response is $q^{M}$ due to Proposition
 		\ref{app:prop:mon:alloc}. If firm $i$ produces $q^{M}$, then firm
 		$j$ can invest in qualities beyond $q^{M}$. Let's verify that becoming
 		a monopolist for better qualities than $q^{M}$ is not a best response
 		when an opponent produces $q^{M}$. If firm $j$ acquires quality
 		$q>q^{M}$, her profits can be expressed as $\int_{[q^{M},q]}V'(\tilde{q})-c'(\tilde{q})\diff\tilde{q}-c(q^{M})$.
 		So, $j$ is better off staying idle than investing in quality $q>{q}^{M}$.
 		
 		Let's verify that there are no competitive equilibria in pure strategies.
 		Specifically, we claim that, in all equilibria in pure strategies,
 		one firm produces $q^{M}$ and all other firms stay idle. First, in
 		any equilibrium, the firm producing the lower quality earns no revenues
 		and can save on acquisition costs by becoming idle. All firms staying
 		idle does not occur in equilibrium as well, since any firm $i$ prefers
 		to be a monopolist (this observation follows from Proposition \ref{app:prop:mon:alloc}.)
 		Thus, the only pure-strategy equilibria are the monopoly ones. \end{remark}
 	
 	%\clearpage 
 	\section{Supplementary material}\label{app:sec:extra}
 	
 	\subsection{Auxiliary results}
 	
 	\begin{lemma}\label{app:lem:general} Let $(y,x)\in Q^{2}$, $K\colon Q\times\Theta\to\mathbb{R}$
 		be continuously differentiable, and $K(\cdot,\theta)$ be concave for every
 		$\theta\in\Theta$. Then, $\bm{q}\in\bm{Q}$ solves 
 		\begin{align*}
 			\mathcal{Q}:\max_{\bm{q}\in\bm{Q}}\int_{\Theta}K(\bm{q}(\theta),\theta)\diff F(\theta)\ \text{subject to:}\ y\le\bm{q}(\theta)\le x\ \text{for all}\ \theta\in\Theta,
 		\end{align*}
 		if and only if: there exists an allocation $\bm{\gamma}$ such that
 		$\bm{\gamma}(\theta)\in\argmax_{q\in Q}K(q,\theta)$ almost everywhere
 		and $\bm{q}(\theta)=\max\{\min\{\bm{\gamma}(\theta),x\},y\}$ for
 		all $\theta$. \end{lemma} 
 	\begin{proof}
 		Fix $(y,x)\in Q^{2}$. We establish the ``if'' direction first.
 		
 		Fix $\bm{q}'\in\bm{Q}$ such that $y\le\bm{q}'(\theta)\le x$ for
 		all $\theta$, and $\bm{\gamma}\in\bm{Q}$ with $\bm{\gamma}(\theta)\in\argmax_{q\in Q}K(q,\theta)$
 		for a.e.~$\theta$. Define $\Delta=\int_{\Theta}K(\bm{q}(\theta),\theta)-K(\bm{q}'(\theta),\theta)\diff F(\theta)$
 		for $\bm{q}(\theta)=\max\{\min\{\bm{\gamma}(\theta),x\},y\}$, and
 		$S=\{\theta\in\Theta\mid\bm{q}(\theta)=x\}$, $T=\{\theta\in\Theta\mid \bm{q}(\theta)=y\}$.
 		We note that 
 		\begin{multline*}
 			\Delta=\int_{T}K(y,\theta)-K(\bm{q}'(\theta),\theta)\diff F(\theta)+\int_{\Theta\setminus(S\cup T)}K(\bm{\gamma}(\theta),\theta)-K(\bm{q}'(\theta),\theta)\diff F(\theta)\\
 			+\int_{S}K(x,\theta)-K(\bm{q}'(\theta),\theta)\diff F(\theta).
 		\end{multline*}
 		By definition of $\bm{\gamma}$, we have $K(\bm{\gamma}(\theta),\theta)-K(\bm{q}'(\theta),\theta)\ge0$
 		a.e.~on $\Theta\setminus(S\cup T)$. Additionally, we have $K(x,\theta)-K(\bm{q}'(\theta),\theta)=\int_{[\bm{q}'(\theta),x]}K_{1}(\tilde{q},\theta)\diff\tilde{q}$,
 		which is nonnegative a.e.~on $S$ by definition of $\bm{\gamma}$
 		and concavity of $K(\cdot,\theta)$; similarly, we have $K(y,\theta)-K(\bm{q}'(\theta),\theta)=-\int_{[y,\bm{q}'(\theta)]}K_{1}(\tilde{q},\theta)\diff\tilde{q}$,
 		which is nonnegative a.e.~on $T$ by definition of $\bm{\gamma}$
 		and concavity of $K(\cdot,\theta)$. Hence, we have $\Delta\ge0$,
 		so $\bm{q}$ solves $\mathcal{Q}$.
 		
 		It remains to establish the ``only if'' direction. Fix a solution
 		$\bm{q}^{*}$ to $\mathcal{Q}$, and suppose that there exists $(\theta',\theta'')\subseteq\Theta$
 		such that: (for every ${\bm{\gamma}}\in\bm{Q}$ with $\bm{\gamma}(\theta)\in\argmax_{q\in Q}K(q,\theta)$
 		a.e.,~we have $\bm{q}'(\theta)\ne\bm{q}^{*}(\theta)$ on $(\theta',\theta'')$,
 		for $\bm{q}'\colon\theta\mapsto\max\{\min\{\bm{\gamma}(\theta),x\},y\}$).
 		Define $S=\{\theta\in(\theta',\theta''):\bm{q}'(\theta)=x\}$, $T=\{\theta\in(\theta',\theta''):\bm{q}'(\theta)=y\}$,
 		and $\Delta=\int_{(\theta',\theta'')}K(\bm{q}'(\theta),\theta)-K(\bm{q}^{*}(\theta),\theta)\diff F(\theta)$;
 		it holds that 
 		\begin{multline*}
 			\Delta=\int_{T}\int_{[\bm{q}^{*}(\theta),y]}K_{1}(\tilde{q},\theta)\diff\tilde{q}\diff F(\theta)+\int_{(\theta',\theta'')\setminus(S\cup T)}K(\bm{q}'(\theta),\theta)-K(\bm{q}^{*}(\theta),\theta)\diff F(\theta)\\
 			+\int_{S}\int_{[\bm{q}^{*}(\theta),x]}K_{1}(\tilde{q},\theta)\diff\tilde{q}\diff F(\theta)
 		\end{multline*}
 		By definition of $\bm{\gamma}$, we have $K(\bm{q}'(\theta),\theta)-K(\bm{q}^{*}(\theta),\theta)\ge0$
 		a.e.~on $(\theta',\theta'')\setminus(S\cup T)$, Additionally, we
 		have: $\int_{[\bm{q}^{*}(\theta),x]}K_{1}(\tilde{q},\theta)\diff\tilde{q}\ge0$
 		a.e.~on $S$ by definition of $\bm{q}'$, ${\bm{\gamma}}$, and concavity
 		of $K(\cdot,\theta)$; similarly, we have: $\int_{[\bm{q}^{*}(\theta),y]}K_{1}(\tilde{q},\theta)\diff\tilde{q}\ge0$
 		a.e.~on $T$ by definition of $\bm{q}'$, ${\bm{\gamma}}$, and concavity
 		of $K(\cdot,\theta)$. It follows that $\Delta\ge0$. To show that
 		$\Delta>0$, we use the fact that $\bm{q}^{*}$ differs from $\theta\mapsto\min\{{\bm{\gamma}}(\theta),q\}$
 		on $(\theta',\theta'')$, for every allocation ${\bm{\gamma}}$ with
 		$\bm{\gamma}(\theta)\in\argmax_{q\in Q}K(q,\theta)$ almost everywhere.
 		%We conclude that there exists an allocation $\bm \gamma^* \colon \Theta\to Q$ with $\bm \gamma ^*(\theta) \in \argmax_{\hat  q \in  Q}K(\hat q, \theta)$ a.e.~such that: $\bm q^*(\theta)=\min\{ \bm \gamma^* (\theta), q\}$ on $(0, 1)$.
 		
 	\end{proof}
 	\begin{lemma}\label{app:lem:pooling} The allocation $\bm{q}^{*}\colon\theta\mapsto\min\{\overline{\bm{\beta}}(\theta),q\}$
 		has the ``pooling property'' \citep*{toikka_ironing_2011}. \end{lemma} 
 	\begin{proof}
 		First, we note that $G(F(\theta),q)$ and $H(F(\theta),q)$ do not
 		depend on $q$, so we omit this argument (Remark \ref{app:rem:ironing}),
 		and that $J$ is ``separable'' \citep[Definition 3.1]{toikka_ironing_2011}.
 		Fix $q\in Q$, $(\theta',\theta'')\subseteq\Theta$, and suppose $G(F(\theta))<H(F(\theta))$
 		for all $\theta\in(\theta',\theta'')$. There are three cases. First,
 		if $\theta''\le T^{-}(q)$, then $\bm{q}^{*}$ is constant on $(\theta',\theta'')$
 		by the pooling property of $\overline{\bm{\beta}}$ \citep[Definition 3.5, Theorem 3.7, Corollary 3.8]{toikka_ironing_2011}.
 		Second, if $\theta'\ge T^{-}(q)$, then $\bm{q}^{*}(\theta)=q$ for
 		all $\theta\in(\theta',\theta'')$. Lastly, we consider the case in
 		which $\theta'<T^{-}(q)<\theta''$. In this case, $\overline{\bm{\beta}}(\theta)<q$
 		for $\theta\in(\theta',a')$ and $\overline{\bm{\beta}}(\theta)\ge q$
 		for $\theta\in(a'',\theta'')$, for some $(a',a'')\in(\theta',a'')\times(a',\theta'')$,
 		which contradicts the pooling property of $\overline{\bm{\beta}}$.
 		So, there exists no such interval, and $\bm{q}^{*}$ has the pooling
 		property. 
 	\end{proof}
 	
 	\subsection{Non-regular distribution}
 	
 	\label{app:sec:ironing} In this section, we maintain the additional
 	assumption that $u(q,\theta)=g(q)+\theta q$, for a strictly concave
 	and increasing $g\colon\mathbb{R}\to\mathbb{R}$, and note that $a(q)=0$.
 	We consider the cumulative virtual value $H(q,\cdot)\colon\theta\mapsto\int_{[0,\theta]}J(q,F^{-1}(\tilde{\theta}))\diff\tilde{\theta}$
 	(in the quantile space,) its lower convex envelope, $\operatorname{conv}H(q,\cdot)\colon\theta\mapsto\min\{\lambda H(q,\theta_{1})+(1-\lambda)H(q,\theta_{2})\mid(\theta_{1},\theta_{2},\lambda)\in[0,1]^{3}\ \text{and}\ \lambda\theta_{1}+(1-\lambda)\theta_{2}=\theta\}$,
 	and the right derivative of $\vex H(q,\cdot)$, $g(q,\cdot)$. The
 	\emph{ironed virtual surplus maximizer} is the largest selection 
 	from $\theta\mapsto\argmax_{q\in Q}\overline{J}(q,\theta)$, in which
 	$\overline{J}(q,\theta)\coloneqq J(0,\theta)+\int_{[0,q]}g(\tilde{q},F(\theta))\diff\tilde{q}$; we denote the ironed virtual surplus maximizer by $\overline{\bm{\beta}}$, and note that $\overline{\bm{\beta}}$ is nondecreasing \citep*{topkis_minimizing_1978,toikka_ironing_2011}.
 	We define the left-continuous inverse of $\overline{\bm{\beta}}$
 	by: $T^{-}(q)\coloneqq\inf\{\theta\in\Theta\mid \overline{\bm{\beta}}(\theta)\ge q\}$.
 	\begin{remark} \label{app:rem:ironing} We note that $\overline{J}(q,\theta)=g(q)+\overline{k}(\theta)q$
 		for some nondecreasing $\overline{k}$. So, $\overline{J}(\cdot,\theta)$
 		is strictly quasiconcave for almost every $\theta\in\Theta$ because
 		$g$ is strictly concave. The proof of Proposition \ref{app:prop:mon:ironing}
 		uses the fact that $\argmax_{q\in Q}\overline{J}(q,\theta)$ is a
 		singleton almost everywhere in the first claim. The remaining steps
 		in the proof hold essentially as stated for more general single-crossing
 		and concave--in-quality $u$. \end{remark} \begin{proposition}\label{app:prop:mon:ironing}
 		Let $q^{M}$ be the unique element of $\argmax_{q\in Q}V(q)-c(q)$.
 		Under the assumptions of this section, the allocation $\bm{q}$ is
 		monopolist if and only if: there exists a nondecreasing allocation
 		$\bm{\gamma}$ such that $\bm{\gamma}(\theta)=\overline{\bm{\beta}}(\theta)$
 		almost everywhere and $\bm{q}(\theta)=\min\{\bm{\gamma}(\theta),q^{M}\}$
 		for all $\theta$. Moreover, it holds that $q^{M}<q^{\star}$.\end{proposition} 
 	\begin{proof}
 		The proof has three main steps. First, we characterize the solutions
 		to $\mathcal{P}(q)$ for all $q\in Q$ by adapting the argument of
 		\citet[proof of Theorem 3.7]{toikka_ironing_2011}; second, we show
 		that $V$ is concave; third, we show that $q^{M}<q^{\star}$. The
 		following preliminary observations hold by known arguments. The preliminary
 		observation in the proof of Lemma \ref{app:lem:mon:dec} holds, and
 		the allocation $\bm{q}$ is monopolist iff: $\bm{q}$ solves $\mathcal{P}(q^{M})$
 		for a quality $q^{M}\in\argmax_{q\in Q}V(q)-c(q)$. Lastly, $\overline{J}(\cdot,\theta)$
 		is concave, continuously differentiable, and 
 		\begin{align*}
 			V(q)=\max_{\bm{q}\in\bm{Q}}\int_{\Theta}\overline{J}(\bm{q}(\theta),\theta)\diff F(\theta)\ \text{subject to:}\ \bm{q}(\theta)\le q\ \text{for all}\ \theta\in\Theta,
 		\end{align*}
 		by \citet[respectively, by Lemma 4.10, Lemma 4.11, and Theorem 4.5.]{toikka_ironing_2011}
 		
 		\emph{Claim: for fixed $q\in Q$, $\bm{q}$ solves $\mathcal{P}(q)$
 			iff $\bm{q}(\theta)=\min\{\bm{\gamma}(\theta),q^{M}\}$ for all $\theta$
 			and some nondecreasing $\bm{\gamma}\in\bm{Q}$ with $\bm{\gamma}(\theta)=\overline{\bm{\beta}}(\theta)$
 			almost everywhere.}
 		
 		Fix $q\in Q$ and an allocation $\bm{q}^{*}$ such that $\bm{q}^{*}(\theta)=\min\{\bm{\gamma}(\theta),q\}$
 		for all $\theta$ and some nondecreasing $\bm{\gamma}\in\bm{Q}$ such
 		that $\bm{\gamma}(\theta)=\overline{\bm{\beta}}(\theta)$ almost everywhere.
 		First, we establish that $\int_{\Theta}G(F(\theta),q)-H(F(\theta),q)\diff\bm{q}^{*}(\theta)=0$.
 		By Remark \ref{app:rem:ironing}, the maximizers of $\overline{J}(\cdot,\theta)$
 		agree almost everywhere, so the statement holds by the ``pooling
 		property'' of $\theta\mapsto\min\{\overline{\bm{\beta}}(\theta),q\}$
 		(Lemma \ref{app:lem:pooling}). By Lemma \ref{app:lem:general} and
 		the preliminary observations, $\bm{q}^{*}$ maximizes $\int_{\Theta}\overline{J}(\bm{q}(\theta),\theta)\diff F(\theta)$
 		subject to $\bm{q}(\theta)\le q$ for all $\theta\in\Theta$. Lastly,
 		the objective in $\mathcal{P}(q)$ can be expressed as $\int_{\Theta}\overline{J}(\bm{q}(\theta),\theta)\diff F(\theta)+\int_{\Theta}G(F(\theta),q)-H(F(\theta),q)\diff\bm{q}(\theta)$,
 		by \citet[proof of Theorem 3.7,]{toikka_ironing_2011} so we conclude
 		that $\bm{q}^{*}$ solves $\mathcal{P}(q)$. Hence, it follows that:
 		if $\bm{q}\in\bm{Q}$ is such that $\bm{q}(\theta)=\min\{\bm{\gamma}(\theta),q\}$
 		for all $\theta$ and some nondecreasing $\bm{\gamma}\in\bm{Q}$ with
 		$\bm{\gamma}(\theta)=\bm{\beta}(\theta)$ almost everywhere, then
 		$\bm{q}$ solves $\mathcal{P}(q)$.
 		
 		For the other direction, suppose $\bm{q}^{M}$ solves $\mathcal{P}(q)$
 		and that there does exist a nondecreasing $\bm{\gamma}\in\bm{Q}$
 		such that $\bm{q}^{M}(\theta)=\min\{\bm{\gamma}(\theta),q\}$ for
 		all $\theta\in\Theta$ and $\bm{\gamma}(\theta)=\overline{\bm{\beta}}(\theta)$
 		almost everywhere. We consider the following two cases. First, suppose
 		that $\bm{q}^{M}$ maximizes $\int_{\Theta}\overline{J}(\bm{q}(\theta),\theta)\diff F(\theta)$
 		subject to $\bm{q}(\theta)\le q$ for all $\theta\in\Theta$; then
 		$\bm{q}^{M}(\theta)=\min\{\overline{\bm{\beta}}(\theta),q\}$ a.e.~by
 		Lemma \ref{app:lem:general}. Second, suppose that $\bm{q}^{M}$ does
 		not maximize $\int_{\Theta}\overline{J}(\bm{q}(\theta),\theta)\diff F(\theta)$
 		subject to $\bm{q}(\theta)\le q$ for all $\theta\in\Theta$; then
 		$\bm{q}^{M}$ does not solve $\mathcal{P}(q)$ by the properties of
 		$\bm{q}^{*}$ in the preceding paragraph. So, the claim holds.
 		
 		\emph{Claim: $V$ is concave and its left-derivative at $q\in(0,{\overline{q}}]$
 			is $\int_{[T^{-}(q),1]}\overline{J}_{1}(q,\theta)\diff F(\theta)$.}
 		As a first step, we establish that $V$ is continuous. It holds that
 		$|J(q,\theta)|\le\max\{\max J(Q\times\Theta),|\min J(Q\times\Theta)|\}$.
 		Consider a sequence $q_{n}\to q\in Q$, for $q_{n}\in Q,~n\in\mathbb{N}$.
 		By a dominated-convergence argument, $\lim_{n\to\infty}V(q_{n})=V(q)$.
 		So, $V$ is continuous.
 		
 		Let $q_{1}<q_{2}$, $q_{i}\in Q,~i\in\{1,2\}$, and consider the difference
 		quotient $d:=\frac{V(q_{2})-V(q_{1})}{q_{2}-q_{1}}$. We have 
 		\begin{align*}
 			d=\frac{1}{q_{2}-q_{1}}\int_{[T^{-}(q_{1}),T^{-}(q_{2})]}\overline{J}(\overline{\bm{\beta}}(\theta),\theta)-\overline{J}(q_{1},\theta)\diff F(\theta)\\
 			+\int_{(T^{-}(q_{2}),1]}\frac{\overline{J}(q_{2},\theta)-\overline{J}(q_{1},\theta)}{q_{2}-q_{1}}\diff F(\theta).
 		\end{align*}
 		We construct a lower bound for $d$. In particular, because we have
 		$\overline{J}(\overline{\bm{\beta}}(\theta),\theta)\ge\overline{J}(q_{1},\theta)$
 		for all $\theta\in\Theta$, we obtain $d\ge\int_{(T^{-}(q_{2}),1]}\frac{\overline{J}(q_{2},\theta)-\overline{J}(q_{1},\theta)}{q_{2}-q_{1}}\diff F(\theta)$.
 		We construct an upper bound for $d$. In particular, because we have
 		$\overline{J}(\overline{\bm{\beta}}(\theta),\theta)\ge\overline{J}(q_{2},\theta)$
 		for all $\theta\in\Theta$, we obtain $d\le\int_{[T^{-}(q_{1}),1]}\frac{\overline{J}(q_{2},\theta)-\overline{J}(q_{1},\theta)}{q_{2}-q_{1}}\diff F(\theta)$.
 		By left continuity of $T^{-}$, we have $\lim_{q_{1}\to q^{-}}\frac{V(q)-V(q_{1})}{q-q_{1}}=\int_{[T^{-}(q),1]}\overline{J}_{1}(q,\theta)\diff F(\theta)$,
 		implying that $V$ is left differentiable and the left derivative
 		takes the stated form. By known results in convex analysis \citep[Theorem 6.4]{hiriart_fundamentals_2001},
 		we obtain concavity of $V$ by continuity of $V$ and the monotonicity
 		of the left-derivative (monotonicity follows from the same argument
 		as concavity of $V$ in the proof of Proposition \ref{app:prop:mon:alloc}.)
 		
 		There exists a unique maximizer $q^{M}$ of $q\mapsto V(q)-c(q)$
 		by concavity of $V$, and $q^{\star}>\in(0,\overline{q})$ by the
 		same argument as in Proposition \ref{app:prop:eff:alloc}. To complete
 		the proof, we argue that $q^{M}<q^{\star}$. We fix $q\in(0,\overline{q})$
 		in what follows and define the left-derivative of $V$ at $q$ as
 		$\partial_{-}V(q)$. Because $\bm{a}(q)=0$ for all $q\in Q$, it
 		suffices to show that there exists $s\in\Theta$ such that $\partial_{-}V(q)=\int_{[s,1]}J_{1}(q,\theta)\diff F(\theta)$;
 		in fact, it holds that $u_{1}(q,\theta)-k'(q)-J_{1}(q,\theta)=\frac{1-F(\theta)}{F'(\theta)}u_{12}(q,\theta)$.
 		First, if $T^{-}(q)=1$, then $\partial_{-}V(q)=0$, so the claim
 		holds. Second, if $T^{-}(q)=0$, then 
 		\begin{align*}
 			\int_{[T^{-}(q),1]}\overline{J}_{1}(q,\theta)\diff F(\theta)=\operatorname{conv}H(q,1)-\operatorname{conv}H(q,0)=H(q,1)-H(q,0).
 		\end{align*}
 		in which the first equality holds by definition of $\overline{J}$,
 		the second holds because $H(q,\cdot)$ equals $\operatorname{conv}H(q,\cdot)$
 		at the endpoints of $[0,1]$. Hence, the claim holds. Third, we claim
 		that: if $T^{-}(q)\in(0,1)$, then $\operatorname{conv}H(q,T^{-}(q))=H(q,T^{-}(q))$.
 		First, suppose that $T^{-}(q)\in(0,1)$ and $\operatorname{conv}H(q,T^{-}(q))<H(q,T^{-}(q))$.
 		Then, by continuity of $H(q,\cdot)$ and $\operatorname{conv}H(q,\cdot)$,
 		there exists an open interval $I\subseteq\Theta$ such that $H(q,\theta)>\vex H(q,\theta)$
 		for all $\theta\in I$. Then, by the results in \citet[Footnote 12]{toikka_ironing_2011},
 		$\overline{\bm{\beta}}$ is constant on $I$, which contradicts the
 		definition of $T^{-}(q)$ and the fact that $T^{-}(q)\in I$. Hence,
 		if $T^{-}(q)\in(0,1)$, then $\operatorname{conv}H(T^{-}(q),q)=H(T^{-}(q),q)$.
 		As a result, $\int_{[T^{-}(q),1]}\overline{J}_{1}(q,\theta)\diff F(\theta)=H(q,1)-H(q,T^{-}(q))$,
 		so, the proof is complete. 
 	\end{proof}
 	
 	\subsection{Additional results for Section \ref{subsec:competition}}\label{app:supp:sec:competition}
 	
 	For the proofs of the following results, we define $W^{\star}\colon q\mapsto\int_{\Theta}u(q,\theta)-k(q)\diff F(\theta)-c(q)$, $\theta_0=\int_\Theta\theta\diff F(\theta)$, $\mathbb{E}_{n}$ refers to the distribution induced by the mixed
 	strategies defined in Proposition \ref{prop:competition} for an equilibrium
 	with $n$ active firms, in which the highest cap is denoted by $\hat x$ and the second-highest one by $\hat y$.
 	
 	\begin{proposition} The following welfare comparisons hold. 
 		\begin{enumerate}
 			\item Welfare is decreasing in the intensity of competition, that is, if  $n\ge2$, then $ \mathbb{E}_{n+1}\{W(\bm{q}[\hat{x},\hat{y}])\}\le\mathbb{E}_{n}\{W(\bm{q}[\hat{x},\hat{y}])\}$.
 			\item If the monopolist fully bunches, then monopoly dominates duopoly, that
 			is, if $q^{M}\le\bm{\beta}(0)$, then $W(\bm{q}^{M})>\mathbb{E}_{2}\{W(\bm{q}[\hat{x},\hat{y}])\}$.
 			\item If the monopolist does not fully bunch, costs are sufficiently close
 			to fixed, and $\theta$ is uniformly distributed, then duopoly dominates monopoly, that is, if there exist $\alpha>0$ and $a>0$ such that $a>\bm \beta(0)$ and, for all $q$, $c=\left(q/a\right)^{\alpha}$, then $\lim_{\alpha\to\infty}W(\bm{q}^{M})<\lim_{\alpha\to\infty}\mathbb{E}_{2}\{W(\bm{q}[\hat{x},\hat{y}])\}$.
 		\end{enumerate}
 	\end{proposition}
 	
 	The proposition is a consequence of the following three results.
 	
 	\begin{proposition} For all $n\ge 2$, it holds that $\mathbb{E}_{n+1}\{W(\bm{q}[\hat{x},\hat{y}])\}\le\mathbb{E}_{n}\{W(\bm{q}[\hat{x},\hat{y}])\}$.
 	\end{proposition} 
 	\begin{proof}
 		\emph{Step 1: preliminary observations about information rents and
 			order statistics.} Fix a competitive $n$ equilibrium that is symmetric.
 		Every firm makes zero profits. The utility of type $\theta$
 		at the pricing stage is
 		\begin{align*}
 			R(x,y,\theta)=\int_{[0,\theta]}\max\{\min\{\bm{\beta}(\tilde{\theta}),x\},y\}\diff\tilde{\theta},
 		\end{align*}
 		given the realizations of $x$ and $y$ of, resp.,~the first- and
 		second-order statistics from $n$ i.i.d.~draws distributed according
 		to the distribution function $G$ (Proposition \ref{prop:competition}).
 		Note that $R(x,y,\theta)$ is nondecreasing in $x$ and $y$, and
 		is increasing in $x$ if $\theta\in(b(x),1]$ and
 		increasing in $y$ if $\theta\in[0,b(y))$.
 		
 		The conditional distribution of the second-order statistic $\hat{y}$
 		given that the first-order statistic's realization is $x$ is given
 		by $G(y\mid x):=\left(G(y)/G(x)\right)^{n-1},~y\in Q$. So, by Proposition
 		\ref{prop:competition}, the conditional distribution is given by
 		$G(y\mid x)=(c'(y)/V'(y))/(c'(x)/V'(x))$ and is constant in $n$.
 		The conditional distribution given by $G(y\mid x)$ is also increasing
 		in the FOSD order as $x$ increases. Specifically, it holds that $G(y\mid x)\le G(y\mid x')$
 		for all $y,x,x'\in Q$ with $x\ge x'$, and with strict inequality
 		if, in addition, we consider the interior of the relevant supports;
 		i.e., for all $y,x,x'\in(0,{q}^{M})$ with $x<x'$ (Proposition
 		\ref{prop:competition}).
 		
 		The distribution of the first-order statistic $\hat{x}$, instead,
 		is given by $\left(G(x)\right)^{n}$. So, by Proposition \ref{prop:competition},
 		the distribution of $\hat{x}$ is nonincreasing in the FOSD order
 		as $n$ increases. Specifically, it holds that $(G(x))^{n}=(c'(x)/V'(x))^{n/(n-1)}$,
 		and so 
 		\begin{align*}
 			(c'(x)/V'(x))^{\frac{n}{n-1}}\le(c'(x)/V'(x))^{\frac{n+1}{n}}.
 		\end{align*}
 		for all $x\in Q$ and $n\ge2$, with strict inequality if $x\in(0,{q}^{M})$.
 		
 		\emph{Step 2: The conditional expectation of welfare given $x$ is
 			monotone in $x$ and constant in $n$.} By the conclusion of Step 1, and first-order--dominance (FOSD) order
 		the conditional expectation of information rents given $\hat{x}=x$
 		is constant in $n$ and nondecreasing in $x$. Specifically, the conditional
 		expectation of information rents given $\hat{x}=x$ is 
 		\begin{align*}
 			R(x):=\int_{[0,{q}^{M}]}R(x,y,\theta)\diff\left(G(y)/G(x)\right)^{n-1},
 		\end{align*}
 		and $R(x)$ is constant in $n$ and nondecreasing in $x$. Moreover,
 		$R(x)>R(x')$ if $0<b(\hat{y})<b(x')$ and $b(x)<1$
 		with positive probability---as induced by the distribution given
 		by the distribution function $y\mapsto\left(G(y)/G(x)\right)^{n-1}$
 		for $\hat{y}$ given $\hat{x}=x$.
 		
 		\emph{Step 3: The expected welfare is monotone in $n$.} By the above
 		observations (Step 1) and known results about FOSD , the expected
 		information rents are nonincreasing in $n$. Specifically, define
 		\begin{align*}
 			R_{n}:=\int_{[0,{q}^{M}]}R(x)\diff(G(x))^{n},
 		\end{align*}
 		and observe that $R_{n}=\mathbb{E}_{n}\{W(\bm{q}[\hat{x},\hat{y}])\}$;
 		we have that $R_{n}\ge R_{n+1}$ for all $n\ge2$. Moreover, monotonicity
 		is strict if $0<b(\hat{y})<b(x')$ and $b(x)<1$
 		with positive probability given $n$ and $n+1$. 
 	\end{proof}
 	\begin{proposition} If  $q^{M}\le\bm{\beta}(0)$, then $W(\bm{q}^{M})>\mathbb{E}_{2}\{W(\bm{q}[\hat{x},\hat{y}])\}$.
 	\end{proposition} 
 	\begin{proof}
 		Note that: the utility from a free good of quality $y$ to type 0
 		is $g(y)$, and the marginal revenue under full bunching is given
 		by $g'$.
 		
 		\emph{Step 1: Welfare under full bunching and monopoly.} Under full
 		bunching at $q^{M}$, welfare is given by $W^{\star}(q^{M})$, that
 		is $W^{\star}(q^{M})=\int_{[0,q^{M}]}w^{\star}(q)\diff q$, using
 		$w^{\star}(q):=g'(q)+\theta_0-c'(q)$.
 		
 		\emph{Step 2: Welfare under full bunching and competition.} Fix a
 		competitive and symmetric $n$ equilibrium and suppose $\bm{\beta}(0)\ge x$.
 		In this case, using full bunching, we
 		obtain $V'(q)=g'(q)$ for all $q\in(0,q^{M}]$. Welfare conditional
 		on $x$ and $y$ is $R(x,y,\theta)=g(x)+\theta x-(g(x)-g(y))$. So,
 		the conditional expectation of welfare given $x$, $R(x)\coloneqq  \int_{[0,q^{M}]}R(x,y,\theta)\diff\left(G(y)/G(x)\right)^{n-1}$, is 
 		\begin{align*}
			 R(x)=\int_{[0,q^{M}]}g(y)\diff\left(\left(G(y)/G(x)\right)^{n-1}\right)+x\theta_0.
 		\end{align*}
 		Using the formula for $G$ in Proposition \ref{prop:competition}
 		and integration by parts, we express expected welfare under full bunching and competition
 		as 
 		\begin{align*}
 			R_{n} & :=\int_{[0,q^{M}]}R(x)\diff(G(x))^{n},\\
 			& =\int_{[0,q^{M}]}g(x)-c(x)\frac{g'(x)}{c'(x)}+x\theta_0\diff(c'(x)/V'(x))^{n/(n-1)}.
 		\end{align*}
 		
 		\emph{Step 3: Welfare under full bunching and duopoly.} We express
 		expected welfare under full bunching and a $2$ equilibrium that is
 		competitive and symmetric as 
 		\begin{align*}
 			R_{2} 
 			& =g(q^{M})-c(q^{M})\frac{g'(q^{M})}{c'(q^{M})}+\theta_0q^M-\\
 			&\int_{[0,q^{M}]}\left[\frac{\partial}{\partial x}\left(g(x)-c(x)\frac{g'(x)}{c'(x)}\right)+\theta_0 \right](c'(x)/V'(x))^{2}\diff x,
 		\end{align*}
 		using integration by parts. Using $V'(x)=g'(x)$ and $\frac{c'({q}^{M})}{g'({q}^{M})}=1$,
 		we have 
 		\begin{align*}
 			W^{\star}(q^{M})-R_{2}=\int_{[0,q^{M}]}\left[\frac{\partial}{\partial x}\left(g(x)-c(x)\frac{g'(x)}{c'(x)}\right)+\theta_0\right](c'(x)/V'(x))^{2}\diff x.
 		\end{align*}
 		The claim holds because $\theta_0>0$ and $
 			\frac{\partial}{\partial x}\left(g(x)-c(x)\frac{g'(x)}{c'(x)}\right)=-c(x)\frac{g''(x)c'(x)-g'(x)c''(x)}{(c'(x))^{2}}\ge0$.
 	\end{proof}
 	\begin{proposition} Suppose that: (1) for all $q$, $g(q)=0$, (2)  there exist $\alpha>1,~a>0$, such that, for all $q$,  $c(q)=(q/a)^{\alpha}$, (3) $\theta$ is uniformly distributed. It holds that: $\mathbb{E}_{2}\{W(\bm{q}[\hat{x},\hat{y}])\}-W(\bm{q}^{M})\to\frac{1}{8}a$ as $\alpha\to\infty$.
 	\end{proposition} 
 	\begin{proof}
 		We establish the result using the following claims.
 		
 		(0) The following equality holds:
 		\begin{align*}
 			\mathbb{E}_{2}\{W(\bm{q}[\hat{x},\hat{y}])\}-W(\bm{q}^{M})=\mathbb{E}_{2}\left\{ \frac{1}{8}(x+y-q^{M})+\frac{1}{4}y\right\} -\frac{1}{4}\int_{Q}q\diff H(q),
 		\end{align*}
 		defining $H\colon Q\mapsto\mathbb{R}:q\mapsto\min\left\{ \frac{c'(q)}{V'(q)},1\right\} $.
 		To see that the equality holds: consumer surplus of a monopolist is
 		\begin{align*}
 			U^{M}:= \int_{\Theta}q^{M}(\theta-\varphi^{-1}(0))_{+}\diff F(\theta)	= q^{M}/8;
 		\end{align*}
 		the consumer surplus in oligopoly is $	U^{O}(x,y)= y\int_{[0,\varphi^{-1}(0)]}\theta\diff F(\theta)+\int_{(\varphi^{-1}(0),1]}x\theta-(x-y)\varphi^{-1}(0)\diff F(\theta)$, so that
 		\begin{align*}
 			U^{O}(x,y)
 			= y/8+x3/8-(x-y)/4
 			= x/8+3y/8;
 		\end{align*}
 		hence, the change in consumer surplus is $	 U^{O}(x,y)-U^{M} = \frac{1}{8}(x+y-q^{M})+\frac{1}{4}y$;
 		the monopoly profits are  $\Pi^{M} =(1-\varphi^{-1}(0))\varphi^{-1}(0)q^{M}-c(q^{M})$, so that:
 		\begin{align*}
 			\Pi^{M}  
 			& =\frac{1}{4}q^{M}-c(q^{M})\\
 			& =\frac{1}{4}\int_{[0,q^{M}]}1-H(q)\diff q\\
 			& =\frac{1}{4}\int_{Q}q\diff H(q).
 		\end{align*}
 		
 		(1) $q^{M}\to a\ \text{as}\ \alpha\to\infty$. By direct calculation:
 		$q^{M}=\left(\frac{a^{\alpha}}{4\alpha}\right)^{\frac{1}{\alpha-1}}$.
 		The result follows because $\log\left(4\alpha\right)^{\frac{1}{\alpha-1}}\to0$
 		as $\alpha\to\infty$.
 		
 		(2) The r.v.~$\hat{y}\mid\hat{x}=x$ tends to $x$ weakly as $\alpha\to\infty$.
 		Define the r.v.~$Y_{x}$ on $Q$ with the conditional distribution
 		of $\hat{y}$ given that $\hat{x}=x$, for $x\in Q$. The result follows
 		because the distribution function of $Y_{x}$ is $q\mapsto\left(\min\left\{ \frac{c'(q)}{c'(x)},1\right\} \right)^{\alpha-1}$.
 		Specifically, it holds that $Y_{x}$ approaches the constant $x$
 		in probability as $\alpha\to\infty$, for all $x\in Q$.
 		
 		(3) The r.v.~$Z$ with distribution function $H$ tends to $a$ weakly
 		as $\alpha\to\infty$. Define the r.v.~$Z$ on $Q$ with the distribution
 		function $H$. The result follows because, by direct calculation,
 		$H\colon q\mapsto\frac{\alpha}{a}\left(\min\left\{ \frac{q}{a},1\right\} \right)^{\alpha-1}$.
 		
 		(4) The r.v.~$\hat{x}$ tends to $a$ weakly as $\alpha\to\infty$.
 		Define the r.v.~$X$ with distribution function $q\mapsto (H(q))^2$.
 		The result follows from part (3).
 		
 		By properties of weak convergence, and the fact that the support of
 		the relevant random variables are contained in $Q$, the result
 		follows. 
 	\end{proof}
 	
 	\subsection{On the expression for marginal revenues}
 	In this section, we provide more details for the argument that $b(q)$
 	is the type $\theta$ that maximizes the marginal-revenue expression
 	$(1-F(\theta))(g'(q)+\theta)$; most of these details appear, e.g.,
 	in \citet*{wilson_nonlinear_1993}.
 	
 	A different way to formalize the monopolist problem is as a choice
 	of a tariff, i.e., a function $T\colon Q\to\mathbb{R}$. The monopolist
 	solves the problem $\mathcal{P}'$ 
 	\begin{align*}
 		\sup_{\bm{q},\,T} & \int_{\Theta}T(\bm{q}(\theta))\diff F(\theta)-c(\sup\bm{q})\;\text{subject to:}\;\\
 		& \text{for all}\;(\theta,q)\in\Theta\times Q,\;u(\bm{q}(\theta),\theta)-T(\bm{q}(\theta))\ge u(q,\theta)-T(q),\\
 		&u(\bm{q}(\theta),\theta)-T(\bm{q}(\theta))\ge0.
 	\end{align*}
 	This problem is equivalent to the original monopolist problem $\mathcal{P}^{M}$
 	\citep*{guesnerie_complete_1984}. The unique solution $(\bm{q}^{M},t^{M})$
 	to $\mathcal{P}^{M}$, characterized in Proposition \ref{app:prop:mon:alloc},
 	can be implemented by a tariff $T^{M}$ such that $T(q)=t^{M}(b(q))$
 	if $q\in\bm{q}(\Theta)$. Additionally, if the hazard rate of $F$
 	is increasing, then $T^{M}$ is concave. In what follows, we consider
 	$\mathcal{P}'$ under the restriction that: $T$ is concave, and,
 	for all $\theta$, $u(q,\theta)-T(q)$ is convex as a function of
 	$q$ in a neighborhood of $\bm{q}^{M}(\theta)$. This condition holds
 	if $T=T^{M}$, so is without loss in our setup, and is common for
 	this direct approach to screening, see Section 8.1 in \citet*{wilson_nonlinear_1993}
 	for sufficient conditions in terms of the primitives.
 	
 	After a change of variables, the revenues are $\int_{[0,\sup\bm{q}(\Theta)]}T(q)\diff H(q)$,
 	letting $H(q)\coloneqq F(\partial_{+}T(q)-g'(q))$, for the right-derivative
 	$\partial_{+}T$ of $T$. Integrating by parts, we obtain $\int_{[0,\sup\bm{q}(\Theta)]}T'(q)(1-H(q))\diff q$
 	(\citealp[Lemma 2]{machina_expected_1982}, for which we use concavity
 	of $T$.) Hence, we find the optimal tariff $T$ by setting its derivative
 	$T'$ to solve the revenue maximization ``pointwise,'' up to a term
 	that is invariant in quality. The optimal price schedule $T'$ satisfies
 	$T'(q)\in\argmax_{p\in\mathbb{R}_{+}}p(1-F(p-g'(q)))$; intuitively,
 	any marginal quality increment is priced as if it constitutes an individual
 	market.
 	
 	From the solution to the problem $\mathcal{P}(\overline{q})$ and
 	the fact that the optimal $T$ has $T(q)=t^{M}(b(q))$, we know that
 	$T'(q)$ is equal to $g'(q)+b(q)$. Moreover, if $y^{*}=p^{*}-g'(q)$,
 	then $p^{*}$ maximizes $p\mapsto p(1-F(p-g'(q)))$ iff $y^{*}$ maximizes
 	$y\mapsto(y+g'(q))(1-F(y))$. So, the reason why $b(q)$ maximizes
 	the marginal-revenue expression $(g'(q)+\theta)(1-F(\theta))$ by
 	choice of type $\theta$ is the fact that the marginal utility of
 	type $b(q)$ is the optimal price of a marginal quality increment
 	from $q$.
 	
 	\begin{remark} The formula for $V'$ and the property of $b$ described
 		in the main text can be extended. Define $s(p,q)=\inf\{\theta\in\Theta\mid u_{1}(q,\theta)\ge p\}$
 		and let $(\bm{q}^{*},t^{*})$ solve the monopolist problem $\mathcal{P}(q)$.
 		It holds that 
 		\begin{enumerate}
 			\item $V'(q)=u_{1}(q,b(q))(1-F(b(q)))$; 
 			\item $u_{1}(q,b(q))\in\argmax_{p\in\mathbb{R}_{+}}p(1-F(s(p,q)))$. 
 		\end{enumerate}
 		First, by Proposition \ref{app:prop:mon:alloc} and the envelope theorem,
 		we have that $t^{*}(b(q))=u(q,b(q))-\int_{[0,b(q)]}u_{2}(\bm{q}^{*}(\tilde{\theta}),\tilde{\theta})\diff\tilde{\theta}$
 		\citep*{milgrom_envelope_2002}. Hence, by continuous differentiability
 		of $b$ on $(\bm{\beta}(0),\infty)$, we have $(t^{*}\circ b)'(q)=u_{1}(q,b(q))$
 		for all $q>\bm{\beta}(0)$. Therefore, by the expression for $V$
 		derived in the preceding paragraph, we have $V'(q)=(1-F(b(q)))u_{1}(q,b(q))$.
 	\end{remark}
 
% \clearpage 
	\newpage 
 	\bibliographystyle{te}
 	\addcontentsline{toc}{section}{\refname}\bibliography{DallAraSartoribib}
 	
 \end{document}